\documentclass[reqno,12pt]{preprint}
\usepackage[full]{textcomp}
\usepackage[osf]{newtxtext}
\usepackage{cabin}

\usepackage[varqu,varl]{inconsolata}
\usepackage[cal=boondoxo]{mathalfa}

\usepackage{comment}

\usepackage{amsmath}
\usepackage{amssymb}
\usepackage{amsfonts}
\usepackage{hyperref}
\usepackage{breakurl}

\usepackage{mathrsfs}
\usepackage{enumitem}

\usepackage{tikz}
\usepackage{dsfont}
\usepackage{mhequ} 
\usepackage{mhsymb} 
\usepackage{mhenvs} 
\usepackage{array}
\usepackage{wrapfig}
\usepackage[font=footnotesize,labelfont=sc]{caption}
\usepackage{microtype}

\usepackage{wasysym}
\usepackage{centernot}
\usetikzlibrary{decorations.markings}

\setlist{noitemsep,nolistsep,leftmargin=1.7em}
\usetikzlibrary{snakes}
\usetikzlibrary{decorations}
\usetikzlibrary{positioning}
\usetikzlibrary{shapes}

\DeclareFontFamily{U}{mathx}{\hyphenchar\font45}
\DeclareFontShape{U}{mathx}{m}{n}{
      <5> <6> <7> <8> <9> <10>
      <10.95> <12> <14.4> <17.28> <20.74> <24.88>
      mathx10
      }{}
\DeclareSymbolFont{mathx}{U}{mathx}{m}{n}
\DeclareMathSymbol{\bigtimes}{1}{mathx}{"91}

\def\emptyset{{\centernot\ocircle}}

\definecolor{darkred}{rgb}{0.7,0.1,0.1}
\definecolor{darkblue}{rgb}{0.1,0.1,0.8}
\definecolor{darkgreen}{rgb}{0.1,0.7,0.1}
\def\Sym{\mathfrak{S}}

\def\CH{\mathcal{H}}
\def\CJ{\mathcal{J}}

\providecommand{\figures}{false}
{ \ifthenelse{\equal{\figures}{false}} {#1}{\[ {\rm Figure \ missing !} \]} }{}
\def\id{\mathrm{id}}

\definecolor{connection}{rgb}{0.7,0.1,0.1}
\colorlet{symbols}{black!50}
\colorlet{boundary}{green!75!black}

\tikzset{
root/.style={circle,fill=black!50,inner sep=0pt, minimum size=3mm},
        dot/.style={circle,fill=black,inner sep=0pt, minimum size=1.5mm},
        dotred/.style={circle,fill=black!50,inner sep=0pt, minimum size=2mm},
        var/.style={circle,fill=black!10,draw=black,inner sep=0pt, minimum size=3mm},
        kernel/.style={semithick,shorten >=2pt,shorten <=2pt},
        kernels/.style={snake=zigzag,shorten >=2pt,shorten <=2pt,segment amplitude=1pt,segment length=4pt,line before snake=2pt,line after snake=5pt,},
        rho/.style={densely dashed,semithick,shorten >=2pt,shorten <=2pt},
           testfcn/.style={dotted,semithick,shorten >=2pt,shorten <=2pt},
        renorm/.style={shape=circle,fill=white,inner sep=1pt},
        labl/.style={shape=rectangle,fill=white,inner sep=1pt},
        xic/.style={very thin,circle,fill=symbols,draw=black,inner sep=0pt,minimum size=1.2mm},
        xi/.style={very thin,circle,fill=blue!10,draw=black,inner sep=0pt,minimum size=1.2mm},
        xix/.style={crosscircle,fill=blue!10,draw=black,inner sep=0pt,minimum size=1.2mm},
	xib/.style={very thin,circle,fill=blue!10,draw=black,inner sep=0pt,minimum size=1.6mm},
	xie/.style={very thin,circle,fill=green!50!black,draw=black,inner sep=0pt,minimum size=1.6mm},
	xid/.style={very thin,circle,fill=symbols,draw=black,inner sep=0pt,minimum size=1.6mm},
	xibx/.style={crosscircle,fill=blue!10,draw=black,inner sep=0pt,minimum size=1.6mm},
	kernels2/.style={very thick,draw=connection,segment length=12pt},
	not/.style={thin,circle,fill=symbols,draw=connection,fill=connection,inner sep=0pt,minimum size=0.5mm},
	>=stealth,
        }
        
\newlength\mylen
\tikzset{
markstart/.style={
  decoration={
    markings,
    mark=at position 0.5 with {
      \node[draw=none,inner sep=0pt,fill=none,text width=0pt,minimum size=0pt] {\global\setlength\mylen{\pgfdecoratedpathlength}};
    },
  },
  preaction={decorate},
  postaction={
  	draw=boundary,ultra thick,>=space,
    dash pattern=on 0.4\mylen off 0.8\mylen,dash phase=0\mylen
  },
  postaction={
  	draw=black,
    dash pattern=on 0.6\mylen off 0.4\mylen,dash phase=0.6\mylen
  },
  }
}
\tikzset{
markend/.style={
  decoration={
    markings,
    mark=at position 0.5 with {
      \node[draw=none,inner sep=0pt,fill=none,text width=0pt,minimum size=0pt] {\global\setlength\mylen{\pgfdecoratedpathlength}};
    },
  },
  preaction={decorate},
  postaction={
  	draw=black,
    dash pattern=on 0.6\mylen off 0.8\mylen,dash phase=0
  },
  postaction={
  	draw=boundary,ultra thick,
    dash pattern=on 0.4\mylen off 0.8\mylen,dash phase=0.6\mylen
  },
  }
}
\tikzset{
markboth/.style={
  decoration={
    markings,
    mark=at position 0.5 with {
      \node[draw=none,inner sep=0pt,fill=none,text width=0pt,minimum size=0pt] {\global\setlength\mylen{\pgfdecoratedpathlength}};
    },
  },
  preaction={decorate},
  postaction={
  	draw=black,
    dash pattern=on 0.2\mylen off 0.8\mylen,dash phase=0.4\mylen
  },
  postaction={
  	draw=boundary,ultra thick,
    dash pattern=on 0.4\mylen off 0.8\mylen,dash phase=0.6\mylen
  },
  postaction={
  	draw=boundary,ultra thick,
    dash pattern=on 0.4\mylen off 0.8\mylen,dash phase=0
  },
  }
}

\makeatletter
\def\DeclareSymbol#1#2#3{\expandafter\gdef\csname MH@symb@#1\endcsname{\tikz[baseline=#2,scale=0.15,draw=symbols,line join=round]{#3}}\expandafter\gdef\csname MH@symb@#1s\endcsname{\scalebox{0.7}{\tikz[baseline=#2,scale=0.15,draw=symbols,line join=round]{#3}}}}
\def\<#1>{\csname MH@symb@#1\endcsname}

\def\DeclarePicture#1#2#3{\expandafter\gdef\csname MH@symb@#1\endcsname{\tikz[baseline=#2,line join=round,style={thick}]{#3}}}
\makeatother

\DeclareSymbol{Xi22}{0.5}{\draw (0,0) node[xi] {} -- (-1,1) node[not] {} -- (0,2) node[xi] {};}

\DeclarePicture{triangle}{0.5}{%
\node[dot] (l) at (-0.4,0.5) {};
\node[dot] (r) at (0.4,0.5) {};
\node[dot] (d) at (0,0) {};
\draw[->,blue] (l) -- (r);
\draw[->] (r) -- (d);
\draw[->] (d) -- (l);
\draw[red] (d) -- ++(-90:0.3);
}
\DeclarePicture{line}{0.4cm}{%
\node[dot] (l) at (-0.4,0.5) {};
\node[dot] (r) at (0.4,0.5) {};
\draw[->,blue] (l) -- (r);
}
\DeclarePicture{loop}{0.5}{%
\node[dot] (u) at (0,0.5) {};
\node[dot] (d) at (0,0) {};
\draw[->] (d) to[bend left=70] (u);
\draw[->] (u) to[bend left=70] (d);
\draw[red] (d) -- ++(-90:0.3);
}
\DeclarePicture{bareloop}{0.1cm}{%
\node[dot] (u) at (0,0.5) {};
\node[dot] (d) at (0,0) {};
\draw[->] (d) to[bend left=70] (u);
\draw[->] (u) to[bend left=70] (d);
}
\DeclarePicture{baretriang}{0.1cm}{%
\node[dot] (l) at (-0.4,0.5) {};
\node[dot] (r) at (0.4,0.5) {};
\node[dot] (d) at (0,0) {};
\draw[->,blue] (l) -- (r);
\draw[->] (r) -- (d);
\draw[->] (d) -- (l);
}
\DeclarePicture{leg}{-0.2cm}{%
\node[dot] (d) at (0,0) {};
\draw[red] (d) -- ++(-90:0.3);
}
\DeclarePicture{markedtriang}{0.1cm}{%
\node[dot] (l) at (-0.4,0.5) {};
\node[dot,boundary] (r) at (0.4,0.5) {};
\node[dot] (d) at (0,0) {};
\draw[->] (l) -- (r);
\draw[->] (r) -- (d);
\draw[->] (d) -- (l);
}
\DeclarePicture{markedloop}{-0.1cm}{%
\node[dot] (u) at (0.5,0) {};
\node[dot,boundary] (d) at (0,0) {};
\draw[->] (d) to[bend left=70] (u);
\draw[->] (u) to[bend left=70] (d);
}
\DeclarePicture{markedboth}{0.1cm}{%
\node[dot] (l) at (-0.4,0.5) {};
\node[dot,boundary] (r) at (0.4,0.5) {};
\node[dot] (d) at (0,0) {};
\draw[->] (l) -- (r);
\draw[->] (r) -- (d);
\draw[->] (d) -- (l);
\node[dot] (u) at (0.9,0.5) {};
\draw[->] (r) to[bend left=70] (u);
\draw[->] (u) to[bend left=70] (r);
}

\def\CH{\mathcal{H}}
\def\CP{\mathcal{P}}
\def\CW{\mathcal{W}}
\def\CG{\mathcal{G}}
\def\CJ{\mathcal{J}}
\def\CA{\mathcal{A}}
\def\CE{\mathcal{E}}
\def\CC{\mathcal{C}}
\def\CQ{\mathcal{Q}}
\def\CB{\mathcal{B}}
\def\CM{\mathcal{M}}
\def\CT{\mathcal{T}}

\def\Labn{\mathfrak{n}}

\def\Labhom{\mathfrak{t}}
\def\Lab{\mathfrak{L}}

\def\Deltam{\Delta^{\!-}}
\def\Deltamp{\Delta^{\!\prime}}
\def\Deltak#1{\Delta^{\!\prime\,(#1)}}

\def\CS{\mathcal{S}}
\def\CR{\mathcal{R}}
\def\CF{\mathcal{F}}
\def\kk{\mathfrak{n}}
\def\mm{\mathfrak{m}}
\def\bn{\mathbf{n}}
\def\bT{\mathbf{T}}
\def\c{\mathfrak{d}}
\def\K{\mathfrak{K}}
\def\Contr{\mathfrak{G}}
\def\Del{\mathop{\mathrm{Del}}}

\def\line{%
\begin{tikzpicture}[style={thick}]
\node[dot] (l) at (0,0) {};
\node[dot] (r) at (1,0) {};
\draw[->] (l) -- (r) node[midway,above=-0.1] {\tiny$\Labhom$};
\draw[thick,red] (l) -- ++(180:0.5) node[midway,above=-0.1] {\tiny$0$};
\draw[thick,red] (r) -- ++(0:0.5) node[midway,above=-0.1] {\tiny$0$};
\end{tikzpicture}%
}
\def\linev{%
\begin{tikzpicture}[style={thick},baseline=-0.1cm]
\node[dot] (l) at (0,-0.3) {};
\node[dot] (r) at (0.4,0.3) {};
\draw[->] (l) -- (r) node[midway,left=-0.05] {\tiny$\Labhom$};
\draw[thick,red] (l) -- ++(180:0.5) node[midway,above=-0.1] {\tiny$0$};
\draw[thick,red] (r) -- ++(0:0.5) node[midway,above=-0.1] {\tiny$0$};
\end{tikzpicture}%
}

\newenvironment{DIFnomarkup}{}{} 

\newtheorem{property}{Property}

\newfont{\indic}{bbmss12}

\def\one{\mathbf{1}}

\def\T{\mathbf{T}}

\def\Deltam{\Delta^{\!-}}

\def\id{\mathrm{id}}

\def\NN{\mathfrak{N}}

\def\S{\mathbf{S}}
\def\M{\mathbb{M}}
\def\out{\mathop{\mathrm{out}}}
\newcommand{\inte}{\mathrm{int}}
\newcommand{\exte}{\mathrm{ext}}
\newcommand{\scale}{\mathrm{scale}}
\def\BPHZ{\textnormal{\tiny \textsc{bphz}}}
\def\full{\textnormal{\tiny {full}}}
\def\mmax{\textnormal{\tiny {max}}}
\def\cl{\mathrm{cl}}
\def\uuparrow{{\uparrow\uparrow}}

\def\simnot{\stackrel{\vbox to 0.15em{\hbox{\kern0.07em$^\circ$}}}{\sim}}

\begin{document}

\title{An analyst's take on the BPHZ theorem}
\author{M. Hairer}
\institute{Mathematics Institute, Imperial College London\\ \email{m.hairer@imperial.ac.uk}}

\maketitle

\begin{abstract}
We provide a self-contained formulation of the BPHZ theorem 
in the Euclidean context, which yields a systematic procedure 
to ``renormalise'' otherwise divergent integrals appearing in 
generalised convolutions of functions with a singularity of prescribed order 
at their origin. We hope that the formulation given in this article will
appeal to an analytically minded audience and that it will help to clarify
to what extent such renormalisations are arbitrary (or not). In particular,
we do not assume any background whatsoever in quantum field theory and we
stay away from any discussion of the physical context in which such problems
typically arise.
\end{abstract}

\setcounter{tocdepth}{2}
\tableofcontents

\section{Introduction}

The BPHZ renormalisation procedure named after 
Bogoliubov, Parasiuk, Hepp and Zimmerman \cite{BP,hepp1966,Zimmermann} (but see also 
the foundational results by Dyson and Salam \cite{Dyson,Dyson2,Salam,Salam2}) provides a consistent way
to renormalise probability amplitudes associated to Feynman diagrams in perturbative
quantum field theory (pQFT). The main aim of this article is to provide an
analytical result, Theorem~\ref{theo:BPHZ} below, which is a general form of the ``BPHZ theorem'' in the
Euclidean context. To a large extent, this theorem has been part of the folklore of
mathematical physics  since the publication of the abovementioned works (see for example 
the article \cite{FMRS85} which gives rather sharp analytical bounds and is close in formulation 
to our statement, as well as the series 
of articles \cite{CK,CK1,CK2} which elucidate some of the algebraic aspects of the theory, but focus
on dimensional regularisation which is not available in the general context considered here), but it seems
difficult to find precise  analytical statements in the literature that go beyond the specific
context of pQFT.
One reason seems to be that, in the context of the perturbative expansions arising in pQFT, 
there are three related problems. The first is to control the small-scale behaviour of the 
integrands appearing in Feynman diagrams (the ``ultraviolet behaviour''), the second is to control
their large scale (``infrared'') behaviour, and the final problem is to show that 
the renormalisation required to deal with the first problem can be implemented 
by modifying (in a scale-dependent way) the finitely many coupling constants appearing in 
the Lagrangian of the theory at hand, so that one still has a physical theory.

The approach we take in the present article is purely analytic and completely unrelated to any physical 
theory, so we do not worry about the potential physical interpretation of the renormalisation procedure. We do 
however show in Section~\ref{sec:props} that it has a number of very nice mathematical properties so that the
renormalised integrals inherit many natural properties from their unrenormalised counterparts. 
We also completely discard
the infrared problem by assuming that all the kernels (``propagators'') under consideration are 
compactly supported. For the reader who might worry that this could render our main result
all but useless, we give a simple separate argument showing how kernels 
with algebraic decay at infinity can be dealt with as well.
Note also, that contrary to much of the related theoretical and mathematical physics literature,
all of our arguments take place in configuration space, rather than in Fourier space.
In particular, the analysis presented in this article shares similarities with a number
of previous works, see for example \cite{MR0342091,CalRiv,FMRS85} and references therein.

The approach taken here is informed by some results recently obtained in the context of
the analysis of rough stochastic PDEs in \cite{Proceedings,BHZalg,BPHZana}. Indeed, 
the algebraic structure described in Sections~\ref{sec:Hopf} and~\ref{sec:twisted} below is 
very similar to the one described in \cite{Proceedings,BHZalg}, with the exception 
that there is no ``positive renormalisation'' in the present context. 
In this sense, this article can be seen as a perhaps gentler introduction to these results, with
the content of Section~\ref{sec:BPHZ} roughly parallel to  \cite{BHZalg},
while the content of Section~\ref{sec:analytical} is rather close to that of \cite{BPHZana}.
In particular, Section~\ref{sec:BPHZ} is rather algebraic in nature and allows to conceptualise
the structure of the counterterms appearing in the renormalisation procedure, while 
Section~\ref{sec:analytical} is rather analytical in nature and contains the multiscale
analysis underpinning our main continuity result, Theorem~\ref{theo:BPHZ}.
Finally, in Section~\ref{sec:largeScale}, we deal with kernels exhibiting only algebraic decay
at infinity. While the conditions given in this section are sharp in the absence of
any cancellations in the large-scale behaviour, we do not introduce an 
analogue of the ``positive renormalisation'' of \cite{BPHZana}, so that the argument remains
relatively concise.

\subsection*{Acknowledgements}

{\small
The author would like to thank Ajay Chandra and Philipp Sch\"onbauer 
for several useful discussions during the preparation of this article.
Financial support through ERC consolidator grant 615897 and a Leverhulme 
leadership award is gratefully acknowledged.
}

\section{An analytical form of the BPHZ theorem}
\label{sec:BPHZ}

Fix a countable set $\Lab$ of labels, a map 
$\deg \colon \Lab \to \R$, and an integer dimension $d > 0$. We assume that the set
of labels has a distinguished element which we denote by $\delta \in \Lab$ satisfying $\deg\delta = -d$ and
that, for every multiindex $k$, there is an
injective map $\Labhom \mapsto \Labhom^{(k)}$ on $\Lab$ with $\Labhom^{(0)} = \Labhom$ and such that
\begin{equ}[e:propLab]
(\Labhom^{(k)})^{(\ell)} = \Labhom^{(k+\ell)}\;,\qquad \deg \Labhom^{(k)} = \deg \Labhom - |k|\;.
\end{equ}
We also set $\Lab_\star = \Lab \setminus \{\delta^{(k)}\,:\, k \in \N^d\}$
and we assume that there is a finite set $\Lab_0 \subset \Lab$ such that
every element of $\Lab$ is of the form $\Labhom^{(k)}$ for some $k \in \N^d$ and
some $\Labhom \in \Lab_0$. We then give the following definition.

\begin{definition}
A \textit{Feynman 
diagram} is a finite directed graph $\Gamma = (\CV,\CE)$
endowed with the following additional data:
\begin{claim}
\item An ordered set of distinct vertices $\CL = \{[1],\ldots,[k]\} \subset \CV$ 
such that each vertex in $\CL$ has exactly one outgoing edge (called a ``leg'') 
and no incoming edge, and such that each connected
component of $\Gamma$ contains at least one leg. We will frequently use
the notation $\CV_\star = \CV \setminus \CL$, as well as 
$\CE_\star \subset \CE$ for the edges that are not legs.
\item A decoration $\Labhom\colon \CE \to \Lab$ of the edges of $\Gamma$ such that 
$\Labhom(e) \in \Lab_\star$ if and only if $e \in \CE_\star$.
\end{claim}
\end{definition}

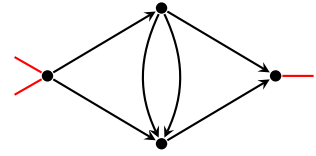
\begin{wrapfigure}{R}{5cm}
\begin{center}
\vspace{-1em}
\begin{tikzpicture}[style={thick}]
\node[dot] (l) at (0,0) {};
\node[dot] (u) at (1.5,.9) {};
\node[dot] (d) at (1.5,-.9) {};
\node[dot] (r) at (3,0) {};
\draw[->] (l) -- (u);
\draw[->] (u) -- (r);
\draw[->] (u) to[bend left=25] (d);
\draw[->] (u) to[bend right=25] (d);
\draw[->] (l) -- (d);
\draw[->] (d) -- (r);
\draw[thick,red] (l) -- ++(150:0.5);
\draw[thick,red] (l) -- ++(-150:0.5);
\draw[thick,red] (r) -- ++(0:0.5);
\end{tikzpicture}
\vspace{-1em}
\end{center}
\caption{A Feynman diagram.}\label{fig:FD}
\end{wrapfigure}

We will always use the convention of \cite{KPZJeremy} that $e_-$ and $e_+$ are 
the source and target of an edge $e$, so that $e = (e_- \to e_+)$. We also label legs 
in the same way as the corresponding element in $\CL$, i.e.\ we call the unique
edge incident to the vertex $[j]$ the $j$th leg of $\Gamma$.
The way we usually think of Feynman diagrams is as labelled graphs $(\CV_\star,\CE_\star)$ with 
a number of legs attached to them, where the legs are ordered and each leg is assigned 
a $d$-dimensional multiindex.
An example of Feynman diagram with 3 legs 
is shown in Figure~\ref{fig:FD}, with legs drawn in red and decorations suppressed.
We do not draw the arrows on the legs since they are always incoming by definition. 
In this example, $|\CV| = 7$  and $|\CV_\star| = 4$. 

Write now $\S = \R^d$, and assume that we are given a kernel $K_\Labhom \colon \S \to \S$
for every $\Labhom \in \Lab_\star$, such that $K_\Labhom$ exhibits a behaviour of
order $\deg \Labhom$ at the origin but is smooth otherwise. For simplicity, we also
assume that these kernels are all compactly supported, say in the unit ball.
More precisely, we assume that for every $\Labhom \in \Lab_\star$ and every $d$-dimensional
multiindex $k$ there exists a constant $C$ such that  one has the bound
\begin{equ}[e:propKer]
|D^k K_\Labhom(x)| \le C |x|^{\deg \Labhom - |k|}\one_{|x| \le 1}\;,\qquad \forall x \in \S\;.
\end{equ}
We also extend $K$ to all of $\Lab$ by using the convention that $K_\delta = \delta$,
a Dirac mass at the origin, and we impose that for every multiindex $k$ and label $\Labhom \in \Lab$, one has
\begin{equ}[e:propK]
K_{\Labhom^{(k)}} = D^k K_\Labhom\;.
\end{equ}
Note that \eqref{e:propKer} is compatible with \eqref{e:propLab} so that non-trivial 
(i.e.\ not just vanishing or smooth near the origin) kernel assignments do actually exist.
To some extent it is also compatible with the convention $K_\delta = \delta$ and
$\deg \delta = -d$ since the ``delta function'' on $\R^d$ is obtained as a distributional limit of functions
satisfying a uniform bound of the type \eqref{e:propKer} with $\deg\Labhom = -d$.
Given all this data, we would now like to associate to each Feynman diagram $\Gamma$ with $k$ legs
a distribution $\Pi \Gamma$ on $\S^k$ by setting
\begin{equ}[e:eval]
\bigl(\Pi \Gamma\bigr)(\phi) = \int_{\S^{\CV}} \prod_{e \in \CE} K_{\Labhom(e)}(x_{e_+} - x_{e_-}) \phi(x_{[1]},\ldots,x_{[k]})\,dx\;.
\end{equ}
Note that of course $\Pi \Gamma$ does not just depend on the combinatorial data
$\Gamma = (\CV,\CE,\CL,\Labhom)$, but also on the analytical data $(K_\Labhom)_{\Labhom \in \Lab_\star}$.
We sometimes suppress the latter dependency on our notation in order to keep it light, but it will 
be very useful later on to also allow ourselves to vary the kernels $K_\Labhom$.
We call the map $\Pi$ a ``valuation''.

The problem is that on the face of it, the definition \eqref{e:eval} does not
always make sense.
The presence of the (derivatives of) delta functions is not
a problem: writing $v_i \in \CV_\star$ for the unique vertex such that 
$([i] \to v_i) \in \CE$ and $\ell_i$ for the multiindex such that the
label of this leg is $\delta^{(\ell_i)}$, we can rewrite \eqref{e:eval} as 
\begin{equ}[e:evalBis]
\bigl(\Pi \Gamma\bigr)(\phi) = \int_{\S^{\CV_\star}} \prod_{e \in \CE_\star} K_{\Labhom(e)}(x_{e_+} - x_{e_-}) 
\bigl(D_1^{\ell_1}\cdots D_k^{\ell_k} \phi\bigr)(x_{v_1},\ldots,x_{v_k})\,dx\;.
\end{equ}
The problem instead is the possible lack of integrability of the integrand appearing
in \eqref{e:evalBis}. For example, the simplest
nontrivial Feynman diagram with two legs is given by
$\Gamma = \line$
which, by \eqref{e:evalBis}, should be associated to the distribution
\begin{equ}[e:evalKernel]
\bigl(\Pi \Gamma\bigr)(\phi) = \int_{\S^2} K_\Labhom(y_1 - y_0) \phi(y_0,y_1)\,dy\;.
\end{equ}
If it happens that $\deg \Labhom < -d$, then $K_\Labhom$ is non-integrable in general,
so that this integral may not converge. It is then natural to modify our definition,
but ``as little as possible''. In this case, we note that if the test function $\phi$ happens
to vanish near the diagonal $y_1 = y_0$, then the singularity of $K_\Labhom$ does not matter
and \eqref{e:evalKernel} makes perfect sense.
We would therefore like to find a distribution $\Pi \Gamma$ which agrees with 
\eqref{e:evalKernel} on such test functions but still yields finite values for
\textit{every} test function $\phi$. One way of achieving this is to set
\begin{equ}[e:renormSimple]
\bigl(\Pi \Gamma\bigr)(\phi) = \int_{\S^2} K_\Labhom(y_1-y_0) \Bigl(\phi(y_0,y_1) - \sum_{|k| + \deg \Labhom \le -d} {(y_1 - y_0)^k \over k!} D_2^k \phi(y_0,y_0) \Bigr)\,dy\;.
\end{equ} 
At first glance, this doesn't look very canonical since it seems that the variables $y_0$ and $y_1$
no longer play a symmetric role in this expression. However, it is an easy exercise to see that 
the \textit{same} distribution can alternatively also be written as
\begin{equ}
\bigl(\Pi \Gamma\bigr)(\phi) = \int_{\S^2} K_\Labhom(y_1-y_0) \Bigl(\phi(y_0,y_1) - \sum_{|k| + \deg \Labhom \le -d} {(y_0 - y_1)^k \over k!} D_1^k \phi(y_1,y_1) \Bigr)\,dy\;.
\end{equ} 
The BPHZ theorem is a far-reaching generalisation of this construction. To formalise what we
mean by this, write $\CK^-_\infty$ for the space of all smooth kernel assignments as
above (compactly supported in the unit ball and
satisfying \eqref{e:propK}). When endowed with the system of 
seminorms given by the minimal constants in \eqref{e:propKer}, its completion $\CK^-_0$ is a Fr\'echet 
space.

With these notations, a ``renormalisation procedure'' is a map
$K \mapsto \Pi^K$ turning a kernel assignment $K \in \CK^-_0$ into 
a valuation $\Pi^K$. The purpose of the BPHZ theorem is to argue that the 
following question can be answered positively.

\medskip{\parindent0em\textbf{Main question:}\label{mainq} Is there a consistent renormalisation procedure 
such that, for every Feynman diagram, $\Pi \Gamma$ can be 
interpreted as a ``renormalised version'' of \eqref{e:eval}?}\medskip

As stated, this is a very loose question since we have not specified what we mean
by a ``consistent'' renormalisation procedure and what properties we would like a valuation to
have in order to be a candidate for an interpretation of \eqref{e:eval}.
One important property we would like a good renormalisation procedure to have
is the continuity of the map $K \mapsto \Pi^K$. In this way, we can always reason on
smooth kernel assignments $K \in \CK^-_\infty$ and then ``only'' need to show that the 
procedure under consideration extends continuously to all of $\CK^-_0$. 
Furthermore, we would like $\Pi^K$ to inherit as many properties as possible from
its interpretation as the formal expression \eqref{e:eval}.
Of course, as
already seen, the ``na\"ive'' renormalisation procedure given by \eqref{e:eval}
itself does \textit{not} have the required continuity property, so we will have to modify it. 

\subsection{Consistent renormalisation procedures}

The aim of this section is to collect and formalise a number of properties of \eqref{e:eval}
which then allows us to formulate precisely what we mean by a ``consistent'' renormalisation procedure.
Let us write $\CT$ for the free (real) vector space generated by
all Feynman diagrams. This space comes with a natural grading and we write
$\CT_k \subset \CT$ for the subspace generated by diagrams with $k$ legs.
Note that $\CT_0 \approx \R$ since there is exactly one Feynman diagram with $0$ legs,
which is the empty one.

Write $\CS_k$ for the space of all distributions on $\S^k$
that are translation invariant in the sense that, for $\eta \in \CS_k$, $h \in \S$,
and any test function $\phi$, one has $\eta(\phi) = \eta(\phi\circ \tau_h)$ where 
$\tau_h(y_1,\ldots,y_k) = (y_1 + h,\ldots,y_k+h)$.
We will write $\CS_k^{(c)}\subset \CS_k$ for the subset of ``compactly supported'' distributions
in the sense that there exists
a compact set $\K \subset \S^k / \S$ such that $\eta(\phi) = 0$ as soon as $\supp \phi \cap \K = \emptyset$.

\begin{remark}\label{rem:compact}
Compactly supported distributions can be tested against any smooth function $\phi$ with the 
property that for any $x \in \S^k$, the set $\{h \in \S\,:\, \phi(\tau_h(x)) \neq 0\}$ is compact.
\end{remark}

Note that $\CS_1 \approx \R$ since translation invariant distributions in one
variable are naturally identified with constant functions. We will use the convention 
$\CS_0 \approx \R$ by identifying ``functions in $0$ variables'' with $\R$. 
We also set $\CS = \bigoplus_{k \ge 0} \CS_k$, so that a valuation $\Pi$
can be viewed as a linear map $\Pi \colon \CT \to \CS$ which respects the
respective graduations of these spaces.

Note that the symmetric group $\Sym_k$ in $k$ elements acts naturally on $\CT_k$ by
simply permuting the order of the legs. Similarly, $\Sym_k$ acts on $\CS_k$ by 
permuting the arguments of the test functions. 
Given two Feynman diagrams $\Gamma_1 \in \CT_k$ and $\Gamma_2 \in \CT_\ell$, we then 
write $\Gamma_1 \bullet \Gamma_2 \in \CT_{k + \ell}$ for
the Feynman diagram given by the disjoint union of $\Gamma_1$ and $\Gamma_2$.
Here, we renumber the $\ell$ legs of $\Gamma_2$ in an order-preserving way from
$k+1$ to $k+\ell$, so that although $\Gamma_1 \bullet \Gamma_2 \neq \Gamma_2 \bullet \Gamma_1$
in general, one has 
$\Gamma_1 \bullet \Gamma_2 = \sigma_{k,\ell}(\Gamma_2 \bullet \Gamma_1)$, where
$\sigma_{k,\ell} \in \Sym_{k+\ell}$ is the permutation that swaps
$(1,\ldots,\ell)$ and $(\ell+1,\ldots,\ell+k)$.
Given distributions $\eta_1 \in \CS_k$ and $\eta_2 \in \CS_\ell$,
we write $\eta_1 \bullet \eta_2 \in \CS_{k+\ell}$ for the distribution such that
\begin{equ}
\bigl(\eta_1 \bullet \eta_2\bigr)(\phi_1 \otimes \phi_2) = \eta_1(\phi_1)\eta_2(\phi_2)\;.
\end{equ}
Similarly to above, one has 
$\eta_1 \bullet \eta_2 = \sigma_{k,\ell}(\eta_2 \bullet \eta_1)$.
We extend $\bullet$ by linearity to all of $\CT$ and $\CS$ respectively, thus turning these
spaces into (non-commutative) algebras. This allows us to formulate the first property we would like to retain.

\begin{property}\label{prop:hom}
A consistent renormalisation procedure should produce valuations $\Pi$ that are 
graded algebra morphisms from $\CT$ to $\CS$ and such that, for every Feynman diagram $\Gamma$
with $k$ legs and every $\sigma \in \Sym_k$, one has $\Pi \sigma(\Gamma) = \sigma(\Pi \Gamma)$.
Furthermore $\Pi \Gamma \in \CS_k^{(c)}$ if $\Gamma$ is connected with $k$ legs.
\end{property}

Similarly, consider a Feynman diagram $\Gamma$ with $k \ge 2$ legs such that the label of the $k$th leg
is $\delta$ and such that the connected component of $\Gamma$ containing $[k]$ 
contains at least one other leg. Let $\Del_k \Gamma$ be the Feynman diagram identical to $\Gamma$, but with the 
$k$th leg removed. If the label of the $k$th leg is $\delta^{(m)}$ with $m \neq 0$, we set
$\Del_k \Gamma = 0$. If we write $\iota_k$ for the natural injection of smooth functions on $\S^{k-1}$ to
functions on $\S^{k}$ given by $(\iota_k \phi)(x_1,\ldots,x_k) = \phi(x_1,\ldots,x_{k-1})$,
we have the following property for \eqref{e:eval} which is very natural
to impose on our valuations..

\begin{property}\label{prop:restr}
A consistent renormalisation procedure should produce valuations $\Pi$
such that for any connected $\Gamma$ with $k$ legs, one has
$
(\Pi \Del_k \Gamma)(\phi) = (\Pi \Gamma)(\iota_k\phi)
$
for all compactly supported test functions $\phi$ on $\S^{k-1}$.
\end{property}

(Note that the right hand side is well-defined by Remark~\ref{rem:compact}
even though $\iota_k \phi$ is no longer compactly supported.)
To formulate our third property, it will be useful to have a notation for our test functions. 
We write $\CD_k$ for
the set of all $\CC^\infty$ functions on $\S^k$ with compact support. 
It will be convenient to consider the following subspaces of $\CD_k$. Let
$\CA$ be a collection of subsets of $\{1,\ldots,k\}$ such that
every set $A \in \CA$ contains at least two elements.
Then, we write 
$\CD_k^{(\CA)} \subset \CD_k$ for the set of such functions $\phi$ which vanish 
in a neighbourhood of the set
$\Delta_k^{(\CA)} \subset \S^k$ given by 
\begin{equ}[e:defDiagonal]
\Delta_k^{(\CA)} = \{y\in \S^k\,:\, \exists A \in \CA \;\text{with}\; y_i = y_j\; \forall i,j \in A\}\;.
\end{equ}
Because of this definition, we also call a collection $\CA$ as above a ``collision set''. 
Note that in particular one has $\CD_k^{(\emptyset)} = \CD_k$.

A first important question to address then concerns the conditions under which
the expression \eqref{e:eval} converges. A natural notion then is that of the 
degree of a subgraph of a Feynman diagram. In this article, we define a subgraph $\bar \Gamma \subset\Gamma$ 
to be a subset $\bar \CE$ of the collection $\CE_\star$ of internal edges
and a subset $\bar \CV \subset \CV_\star$ of the internal vertices such that 
$\bar \CV$ consists precisely of those vertices incident to at least one edge in $\bar \CE$.
(In particular, isolated nodes are not allowed in $\bar \Gamma$.) 
Given such a subgraph $\bar \Gamma$, we then set
\begin{equ}[e:degreeSubgraph]
\deg\bar \Gamma \eqdef \sum_{e \in \bar \CE} \deg\Labhom(e) + d(|\bar \CV| - 1)\;.
\end{equ}
We define the degree of the full Feynman diagram $\Gamma$ in exactly the same way,
with $\bar \CE$ and $\bar \CV$ replaced by $\CE_\star$ and $\CV_\star$.
One then has the following result initially due to Weinberg \cite{Weinberg}.
See also \cite[Thm~A.3]{KPZJeremy} for the proof of a slightly more general statement
which is also notationally closer to the setting considered here.

\begin{proposition}
If $\Gamma$ is a Feynman diagram with $k$ legs such that $\deg \bar \Gamma > 0$ for every 
subgraph $\bar \Gamma \subset \Gamma$, then the integral in \eqref{e:evalBis} is absolutely convergent
for every $\phi \in \CD_k$. \qed
\end{proposition}

We will henceforth call a subgraph $\bar \Gamma \subset \Gamma$ \textit{divergent}
if $\deg \bar \Gamma \le 0$.
A virtually identical proof actually yields the following refined statement
which tells us very precisely where exactly there is a need for renormalisation.
\begin{proposition}\label{prop:WeinFancy}
Let $\Gamma$ be a Feynman diagram with $k$ legs and let $\CA$ be a collision set
such that, for every connected divergent subgraph $\bar \Gamma \subset \Gamma$,
there exists $A \in \CA$ such that every leg in $A$ is adjacent to $\bar \Gamma$. 
Then \eqref{e:evalBis} is absolutely convergent
for every $\phi \in \CD_k^{(\CA)}$.
\end{proposition}

\begin{remark}
Here and below we say that an edge $e$ is adjacent to a subgraph $\bar \Gamma \subset \Gamma$ 
(possibly itself consisting only of a single edge) if
$e$ is not an edge of $\bar \Gamma$, but shares a vertex with such an edge. 
\end{remark}

\begin{proof}
Since the main idea will be useful in the general result, we sketch it here.
Note first that we can assume without loss of generality that, for every $A \in \CA$,
the vertices of $\CV_\star$ to which the legs in $A$ are attached are all distinct, since
otherwise \eqref{e:evalBis} vanishes identically for $\phi \in \CD_k^{(\CA)}$.

The key remark is that, for every configuration of points $x \in \S^{\CV_\star}$ we
can find a binary tree $T$ with leaves given by $\CV_\star$ and a label $\bn_u \in \N$ for
every inner vertex $u$ of $T$ in such a way that $\bn$ is increasing when going from the root to 
the leaves of $T$ and, for any $v, \bar v \in \CV_\star$, one has
\begin{equ}[e:boundDist]
C^{-1}2^{-\bn_u} \le \|x_v - x_{\bar v}\| \le C2^{-\bn_u}\;,
\end{equ}
where $u = v \wedge \bar v$ is the least common ancestor of $v$ and $\bar v$ in $T$. Here, the constant $C$ only
depends on the size of $\CV_\star$.
(Simply take for $T$ the minimal spanning tree of the point configuration.)
Writing $\bT = (T,\bn)$ for this data, we then let $D_{\bT} \subset \S^{\CV_\star}$ be the
set of configurations giving rise to the data $\bT$. By analogy with the construction of
\cite{hepp1966}, we call $D_{\bT}$ a ``Hepp sector''.

\begin{remark}
While the type of combinatorial data $(T,\bn)$ used to index Hepp sectors is identical to 
that appearing in ``Gallavotti-Nicolo trees'' \cite{GN1,GN2} and the meaning of the index $\bn$ is similar in
both cases, there does not appear to be a direct analogy between the terms indexed by this data in
both cases.
\end{remark}

\begin{remark}\label{rem:ultrametric}
Thanks to the tree structure of $T$, the quantity $d_\T$ given by $d_\T(v,\bar v) = 2^{-\bn_u}$ 
as above is an ultrametric. 
\end{remark}

Writing $\bn(e)$ for the value of $\bn_{e^\uparrow}$, with $e^\uparrow = e_-\wedge e_+$,
the integrand of \eqref{e:evalBis} is then bounded by some constant times
$\prod_{e \in \CE_\star} 2^{- \bn(e)\deg \Labhom(e)}$.
Identifying $T$ with its set of internal nodes, one can also show that the measure
of $D_{\bT}$ is bounded by $\prod_{u \in T} 2^{-d\bn_u}$. 
Finally, by the definition of $\CD_k^{(\CA)}$, there exists a constant $N_0$ such that
the integrand vanishes on sets $D_{\bT}$ such that $\sup_{A\in \CA} \bn_{A^\uparrow} \ge N_0$,
where $A^\uparrow$ is the least common ancestor in $T$ of the collection of elements of $\CV_\star$ 
incident to the legs in $A$. Writing
\begin{equ}
\CT_\CA = \{(T,\bn)\,:\, \sup_{A\in \CA} \bn_{A^\uparrow} < N_0\}\;,
\end{equ}
we conclude that \eqref{e:evalBis} is bounded by some constant multiple of
\begin{equ}[e:bigSum]
\sum_{\bT \in \CT_\CA} \prod_{u \in T} 2^{-\eta_u}\;,\qquad \eta = d + \sum_{e \in \CE_\star}\one_{e^\uparrow} \deg\Labhom(e)  \;.
\end{equ}
We now note that the assumption on $\CA$ guarantees that, for every node $u \in T$, one
has either $\sum_{v \ge u} \eta_v > 0$, or there exists some $A \in \CA$ such that $u \le A^\uparrow$.
In the latter case, $\bn_u$ is bounded from above by $N_0$. Furthermore, as a consequence of the fact that
each connected component of $\Gamma$ has at least one leg and the kernels $K_\Labhom$ are compactly supported,
\eqref{e:evalBis} vanishes on all Hepp sectors with some $\bn_u$ sufficiently negative.
Combining these facts, and performing the sum in \eqref{e:bigSum} ``from the leaves inwards''
as in \cite[Lem.~A.10]{KPZJeremy}, 
it is then straightforward to see that it does indeed converge, as claimed.  
\end{proof}

In other words, Proposition~\ref{prop:WeinFancy} tells us that the only 
region in which the integrand of \eqref{e:evalBis} diverges in a non-integrable
way consists of an arbitrarily small neighbourhood of those points $x$ for 
which there exists a divergent subgraph $\bar \Gamma = (\bar \CV, \bar \CE)$ 
such that $x_u = x_v$ for all vertices $u,v \in \bar \CV$. 
It is therefore very natural to impose the following.

\begin{property}\label{prop:conv}
A consistent renormalisation procedure should produce valuations $\Pi$ that 
agree with \eqref{e:eval} for test functions and Feynman diagrams satisfying
the assumptions of Proposition~\ref{prop:WeinFancy}.
\end{property}

Finally, a natural set of relations of the canonical valuation $\Pi$
given by \eqref{e:eval} which we would like to retain is those given by integration
by parts. In order to formulate this, it is convenient to introduce the notion
of a \textit{half-edge}.
A half-edge is a pair $(e,v)$ with $e\in \CE$ and  $v \in \{e_+,e_-\}$. It is
said to be \textit{incoming} if $v = e_+$ and \textit{outgoing} if $v = e_-$.
Given an edge $e$, we also write $e_\leftarrow$ and $e_\rightarrow$ for the two half-edges
$(e,e_-)$ and $(e,e_+)$.
Given a Feynman diagram $\Gamma$, a half-edge $(e,v)$, and $k \in \N^d$,
we then write $\d_{(e,v)}^k \Gamma$ for the element of $\CT$ obtained
from $\Gamma$ by replacing the decoration $\Labhom$ of the edge $e$ by
$\Labhom^{(k)}$ and then multiplying the resulting Feynman diagram
by $(-1)^{|k|}$ if the half-edge $(e,v)$ is outgoing. 
We then write $\d \CT$ for the smallest subspace of $\CT$ such that, for every Feynman diagram
$\Gamma$, every $i \in \{1,\ldots,d\}$ and every inner vertex $v \in \CV_\star$
of $\Gamma$, one has
\begin{equ}[e:IBP]
\sum_{e \sim v} \d_{(e,v)}^{\delta_i} \Gamma \in \d \CT\;,
\end{equ}
where $e \sim v$ signifies that the edge $e$ is incident to the vertex $v$ 
and $\delta_i$ is the $i$th canonical element of $\N^d$.
By integration by parts, it is immediate that if the kernels $K_\Labhom$ are all smooth,
then the canonical valuation \eqref{e:eval} satisfies $\Pi \d \CT = 0$.
It is therefore natural to impose the following.

\begin{property}\label{prop:IBP}
A consistent renormalisation procedure should produce valuations $\Pi$ that 
vanish on $\d \CT$.
\end{property}

Setting $\CH = \CT / \d \CT$, we can therefore consider a valuation as a map
$\Pi \colon \CH \to \CS$.
Note that since $\d \CT$ is an ideal of $\CT$ which respects its grading, $\CH$ is again a graded algebra. 
Furthermore, since $\d \CT$ is invariant under the action of the symmetric group,
$\Sym_k$ acts naturally on $\CH_k$.
In particular, Property~\ref{prop:hom} can be formulated in $\CH$ rather than $\CT$
and it is not difficult to see that the deletion operation $\Del_k$ introduced in Property~\ref{prop:restr}
also makes sense on $\CH$. This motivates the following definition. 

\begin{definition}\label{def:consistent}
A valuation $\Pi\colon \CH \to \CS$ is \textit{consistent} for the kernel assignment $K$ if it satisfies
Properties~\ref{prop:hom}, \ref{prop:restr} and \ref{prop:conv}.
\end{definition}

\subsection{Some algebraic operations on Feynman diagrams}

In order to satisfy Property~\ref{prop:conv}, we will consider valuations that differ from
the canonical one only by counterterms of the same form, but with some of the factors of \eqref{e:eval} 
corresponding to
divergent subgraphs replaced by a suitable derivative of a delta function, just
like what we did in \eqref{e:renormSimple}. 

These counterterms can again be encoded into Feynman diagrams with the
same number of legs as the original diagram, multiplied by a suitable weight.
We are therefore looking for a procedure which, given a smooth kernel assignment $K \in \CK^-_\infty$,
builds a linear map $M^K\colon \CT \to \CT$ such that if we define a ``renormalised''
valuation $\hat \Pi^K$ by
\begin{equ}[e:renormPi]
\hat \Pi^K \Gamma = \Pi^K M^K \Gamma\;,
\end{equ}
with $\Pi^K$ the canonical valuation given by \eqref{e:eval}, then $K \mapsto \hat \Pi^K$ is
a renormalisation procedure which extends continuously to all of $\CK^-_0$.
We would furthermore like $M^K$ to differ from the identity only by terms of the 
form described above, obtained by contracting divergent subgraphs to a derivative of a delta function.

The procedure \eqref{e:renormSimple} is exactly of this form with  
\begin{equ}[e:Msimple]
M^K \linev = \linev - \sum_{|k| + \deg \Labhom \le -d}\!\!\!c_k\cdot \begin{tikzpicture}[style={thick}]
\node[dot] (l) at (0,0) {};
\draw[thick,red] (l) -- ++(180:0.5) node[midway,above=-0.1] {\tiny$0$};
\draw[thick,red] (l) -- ++(0:0.5) node[midway,above=-0.1] {\tiny$k$};
\end{tikzpicture}\;,
\qquad c_k = {1\over k!} \int_\S x^k K_\Labhom(x)\,dx\;.
\end{equ}
Note that the condition $\deg \Labhom \le -d$ which is required for $M^K$ to differ
from the identity is precisely the condition that the subgraph
\begin{tikzpicture}[style={thick}]
\node[dot] (l) at (0,0) {};
\node[dot] (r) at (1,0) {};
\draw[->] (l) -- (r) node[midway,above=-0.1] {\tiny$\Labhom$};
\end{tikzpicture} 
is divergent, which then guarantees that this example satisfies Property~\ref{prop:conv}.

It is natural to index the constants appearing in the terms of
such a renormalisation map by the corresponding subgraphs that were contracted.
These subgraphs then have no legs anymore, but may require additional decorations
describing the powers of $x$ appearing in the expression for $c_k$ above.
We therefore give the following definition, where the choice of terminology is 
chosen to be consistent with the QFT literature.

\begin{definition}
A \textit{vacuum diagram} consists of a Feynman diagram
$\Gamma = (\CV,\CE)$ with exactly one leg per connected component,
endowed additionally with a node decoration $\kk \colon \CV_\star \to \N^d$. We also impose that each
leg has label $\delta$.
We say that a connected vacuum diagram is \textit{divergent} if $\deg \Gamma \le 0$, where 
$$\deg \Gamma = \sum_{e \in \CE}\Labhom(e) + \sum_{v \in \CV} |\kk(v)| + d(|\CV|- 1)\;.$$
We extend this to arbitrary vacuum diagrams by imposing that 
$\deg (\Gamma_1 \bullet \Gamma_2) = \deg\Gamma_1 + \deg\Gamma_2$. 
\end{definition}

One should think of a connected vacuum diagram $\Gamma$ as encoding the constant
\begin{equ}[e:PiK]
\Pi_-^K \Gamma \eqdef \int_{\S^{\CV_\star\setminus \{v_\star\}}} \prod_{e \in \CE_\star} K_{\Labhom(e)}(x_{e_+} - x_{e_-}) 
\prod_{w \in \CV_\star}(x_w - x_{v_\star})^{\kk(w)} \,dx
\end{equ}
where $v_\star$ is the 
element of $\CV_\star$ that has the unique leg attached to it.
This is then extended multiplicatively to all vacuum diagrams.
In view of this, it is also natural to ignore the ordering of the legs for
vacuum diagrams, and we will always do this from now on.

Write now $\hat \CT_-$ for the algebra of all vacuum diagrams such that each connected component 
has at least one internal edge
and by $\CT_- \subset \hat \CT_-$ for the subalgebra generated by those diagrams
such that each connected components is divergent. 
Since we ignored the labelling of legs, the product $\bullet$ turns $\hat \CT_-$ into
a commutative algebra.
Note that if we write $\CJ_+ \subset \hat \CT_-$ for the ideal generated by all
vacuum diagrams $\Gamma$ with $\deg\Gamma > 0$, then we have a natural isomorphism
\begin{equ}
\CT_- \approx \hat \CT_- / \CJ_+\;.
\end{equ}

Similarly to above, it is natural to identify vacuum diagrams related to each other 
by integration by parts, but also those related by changing the location of the leg(s).
In order to formalise this, we reinterpret a connected vacuum diagram as above
as a Feynman diagram ``with $0$ legs'', but with one of the vertices being
distinguished, which is of course completely equivalent, and we write it
as $(\Gamma,v_\star,\kk)$. 
With this notation, we define $\d \hat \CT_-$ as the smallest
ideal of $\hat \CT_-$ such that, for every connected $(\Gamma,v_\star,\kk)$ 
one has the following.
\begin{claim}
\item For every vertex $v \in \CV \setminus \{v_\star\}$ and every $i\in\{1,\ldots,d\}$, one has
\begin{equ}[e:IBPbis]
\sum_{e \sim v} (\d_{(e,v)}^{\delta_i} \Gamma,v_\star,\kk) + \kk(v)_i (\Gamma, v_\star,\kk - \delta_i \one_{v}) \in \d \hat \CT_-\;,
\end{equ}
where $\one_v$ denotes the indicator function of $\{v\}$.
\item One has
\begin{equ}[e:IBProot]
\sum_{e \sim v_\star} (\d_{(e,v_\star)}^{\delta_i} \Gamma,v_\star,\kk) - \sum_{v\in \bar\CV}\kk(v)_i (\Gamma, v_\star,\kk - \delta_i \one_{v}) \in \d \hat \CT_-\;,
\end{equ}
\item For every vertex $v \in \CV$, one has
\begin{equ}[e:moveLeg]
(\Gamma,v_\star,\kk) - \sum_{\mm\colon \CV \to \N^d} (-1)^{|\mm|} \binom{\kk}{\mm} (\Gamma,v,\kk-\mm + \Sigma\mm \one_{v_\star}) \in \d \hat \CT_-\;,
\end{equ}
where $\Sigma \mm = \sum_u \mm(u)$ and we use the convention
$\mm ! = \prod_{u\in \CV} \prod_{i=1}^d \mm(u)_i!$ to define the binomial coefficients,
with the additional convention that the coefficient vanishes unless $\mm \le \kk$ everywhere.
\end{claim}

\begin{remark}\label{rem:ker}
One can verify that if $K \in \CK^-_\infty$ and $\Pi_-^K$ is given by \eqref{e:PiK},
then $\d \hat \CT_- \in \ker \Pi_-^K$.
In the case of \eqref{e:IBPbis} and \eqref{e:IBProot}, this is because the integrand is then
a total derivative with respect to $(x_v)_i$ and $(x_{v_\star})_i$ respectively.
In the case of \eqref{e:moveLeg}, this can be seen by writing
$(x_w - x_{v_\star})^{\kk(w)} = ((x_w - x_{v}) - (x_{v_\star} - x_v))^{\kk(w)}$ and applying
the multinomial theorem.
\end{remark}

\begin{remark}\label{rem:vanish}
The expressions \eqref{e:IBPbis} and \eqref{e:IBProot} are consistent with \eqref{e:IBP} in the special case $\kk = 0$.
Considering the case $v = v_\star$ in \eqref{e:moveLeg}, it is also straightforward to verify that
$(\Gamma,v_\star,\kk) \in \d \hat \CT_-$ as soon as $\kk(v_\star) \neq 0$.
\end{remark}

As before, we then write $\hat \CH_-$ as a shorthand for $\hat \CT_- / \d \hat \CT_-$ and
similarly for $\CH_-$. (This is well-defined since $\d \hat \CT_-$ does not mix elements of 
different degree.) As a consequence of Remark~\ref{rem:ker}, we see that 
every $K \in \CK^-_\infty$ yields a character $\Pi_-^K$ of $\hat \CH_-$
and therefore also of $\CH_-$.

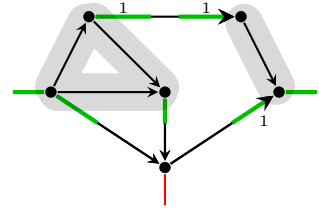
\begin{wrapfigure}{R}{5cm}
\begin{center}
\vspace{-1em}
\begin{tikzpicture}[style={thick}]
\node[dot] (l) at (0,0) {};
\node[dot] (r) at (3,0) {};
\node[dot] (ul) at (0.5,1) {};
\node[dot] (ur) at (2.5,1) {};
\node[dot] (d) at (1.5,-1) {};
\node[dot] (c) at (1.5,0) {};
\draw[->] (l) -- (ul);
\draw[->,markboth] (ul) -- (ur) node[pos=0.2,above=-0.1] {\tiny$1$} node[pos=0.8,above=-0.1] {\tiny$1$};
\draw[->] (ur) -- (r); 
\draw[->] (l) -- (c);
\draw[->] (ul) -- (c);
\draw[->,markstart] (l) -- (d);
\draw[->,markstart] (c) -- (d);
\draw[->,markend] (d) -- (r) node[pos=0.8,below right=-0.1] {\tiny$1$};
\draw[ultra thick,boundary] (l) -- ++(180:0.5);
\draw[ultra thick,boundary] (r) -- ++(0:0.5);
\draw[thick,red] (d) -- ++(-90:0.5);
\draw[line width=0.5cm,draw opacity=0.15, line cap=round] (l) -- (ul) -- (c) -- (l);
\draw[line width=0.5cm,draw opacity=0.15, line cap=round] (ur) -- (r);
\end{tikzpicture} 
\end{center}
\vspace{-1em}
\caption{Example of a subgraph (shaded) and its boundary (green).}\label{fig:subgraph}
\vspace{-1em}
\end{wrapfigure}

Given a Feynman diagram $\Gamma$ and a subgraph $\bar \Gamma \subset \Gamma$, we 
can (and will) identify $\bar \Gamma$ with an element of $\hat \CH_-$, obtained by 
setting all the node decorations to $0$. 
By \eqref{e:moveLeg} we do not need to specify where we attach leg(s) to $\bar \Gamma$
since these elements are all identified in $\hat \CH_-$.
We furthermore write $\d \bar \Gamma$ for the set of all \textit{half-edges} adjacent
to $\bar \Gamma$. 
Figure~\ref{fig:subgraph} shows an example of a Feynman diagram with a subgraph $\bar \Gamma$ shaded in grey
and $\d \bar \Gamma$ indicated in green. 
Legs can also be
part of $\d \bar \Gamma$ as is the case in our example, but they can \textit{not}
be part of $\bar \Gamma$ by our definition of a subgraph.
Note also that the edge joining the two vertices at the top appears
as two distinct half-edges in $\d\bar\Gamma$.
Given furthermore a map $\ell \colon \d \bar \Gamma \to \N^d$ (canonically extended
to vanish on all other half-edges of $\Gamma$),
we then define the following two objects.
\begin{claim}
\item A vacuum diagram $(\bar \Gamma, \pi\ell)$ which consists of the 
graph $\bar \Gamma$ endowed with the edge decoration inherited from $\Gamma$, as well
as the node decoration $\kk = \pi\ell$ given by $(\pi\ell)(v) = \sum_{e\,:\,(e,v) \in \d \Gamma}\ell(e,v)$.
\item A Feynman diagram $\Gamma / (\bar \Gamma, \ell)$ obtained by contracting 
the connected components of $\bar \Gamma$ to nodes and applying $\ell$ to the 
resulting diagram in the sense that, for edges $e \in \CE \setminus \bar \CE$
adjacent to $\d\bar \Gamma$ and with label (in $\Gamma$) given by $\Labhom$,
we replace their label by $\Labhom^{(\ell(e_\leftarrow) + \ell(e_\rightarrow))}$.
\end{claim}
In the example of Figure~\ref{fig:subgraph}, where non-zero values of $\ell$ are indicated by small labels, we have
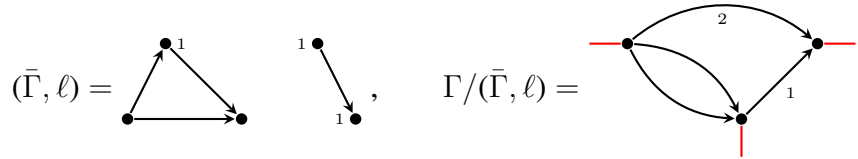
\begin{equ}
(\bar \Gamma, \ell) = 
\begin{tikzpicture}[style={thick},baseline=0.3cm]
\node[dot] (l) at (0,0) {};
\node[dot,label={[shift={(0.1,0)}]left:{\tiny$1$}}] (r) at (3,0) {};
\node[dot,label={[shift={(-0.1,0)}]right:{\tiny$1$}}] (ul) at (0.5,1) {};
\node[dot,label={[shift={(0.1,0)}]left:{\tiny$1$}}] (ur) at (2.5,1) {};
\node[dot] (c) at (1.5,0) {};
\draw[->] (l) -- (ul);
\draw[->] (ur) -- (r); 
\draw[->] (l) -- (c);
\draw[->] (ul) -- (c);
\end{tikzpicture}\;,\qquad
\Gamma / (\bar \Gamma, \ell)
= 
\begin{tikzpicture}[style={thick},baseline=-0.7cm]
\node[dot] (l) at (0,0) {};
\node[dot] (r) at (2.5,0) {};
\node[dot] (d) at (1.5,-1) {};
\draw[->] (l) to[bend left=40] node[midway,below=-0.05] {\tiny$2$} (r);
\draw[->] (l) to[bend left=30] (d);
\draw[->] (l) to[bend right=30] (d);
\draw[->] (d) -- (r) node[midway,below right=-0.1] {\tiny$1$};
\draw[thick,red] (l) -- ++(180:0.5);
\draw[thick,red] (r) -- ++(0:0.5);
\draw[thick,red] (d) -- ++(-90:0.5);
\end{tikzpicture} 
\end{equ}
where a label $k$ on an edge means that if it had a decoration $\Labhom$ in $\Gamma$, 
then it now has a decoration $\Labhom^{(k)}$.
Given a map $\ell \colon \d \bar \Gamma \to \N^d$ as above, we also write 
``$\out \ell$'' as a shorthand for the restriction of $\ell$ to outgoing half-edges.
With these notations at hand, we define a map
$\Delta \colon \CT \to \CH_- \otimes \CH$ by
\begin{equ}[e:coaction]
\Delta \Gamma = \sum_{\bar \Gamma \subset \Gamma} \sum_{\ell\colon \d\bar \Gamma \to \N^d}
{(-1)^{|\out \ell|} \over \ell!} (\bar \Gamma, \pi\ell) \otimes \Gamma / (\bar \Gamma, \ell)\;,
\end{equ}
where we use the same conventions for factorials as in \eqref{e:moveLeg}.
Note that since the right hand side is identified with an element of
$\CH_- \otimes \CH$, this sum is finite. Indeed, unless
$(\bar \Gamma, \pi\ell) \in \CT_-$, which only happens for finitely many choices of $\ell$, 
the corresponding factor is identified with $0$ in $\CH_-$.

\begin{remark}
For any fixed $\Gamma$ this sum is actually finite since there are only finitely many subgraphs and
since, for large enough $\ell$, $(\bar \Gamma, \pi \ell)$ is no longer in $\CT_-$.
\end{remark}

\begin{remark}
The factor $(-1)^{|\out \ell|}$ appearing here encodes the fact that for 
an edge $e$, having $\ell(e, u) = k$ means that in the resulting
Feynman diagram $\Gamma / (\bar \Gamma, \ell)$, one would like to replace the 
factor $K_\Labhom(x_{e_+} - x_{e_-})$ by its $k$th derivative with respect to $x_u$,
which is precisely what happens when one replaces the corresponding connected component
of $\bar \Gamma$ by a derivative of a delta function.
In the case when $u = e_-$, namely when the half-edge is outgoing, 
this is indeed the same as $(-1)^{|k|} (D^kK_\Labhom)(x_{e_+} - x_{e_-})$, while the
factor $(-1)^{|k|}$ is absent for incoming half-edges.
\end{remark}

It turns out that one has the following.

\begin{proposition}\label{prop:Delta}
The map $\Delta$ is well-defined as a map
from $\CH$ to $\CH_- \otimes \CH$. 
\end{proposition}

Before we start our proof, recall the following version of the Chu-Vandermonde identity
\begin{lemma}\label{Chu-Vandermonde}
Given finite sets $S, \bar S$ and maps 
$\pi\colon S \to \bar S$ and $\ell \colon S \to \N$, we define
$\pi_\star \ell \colon \bar S \to \N$ by $\pi_\star \ell(x) = \sum_{y \in \pi^{-1}(x)} \ell(y)$. 

Then, for every finite set $S$ and every $k \colon S \to \N$, one has the identity
\begin{equ}
\sum_{\ell\,:\, \pi_\star \ell}\binom{k}{\ell} = \binom{\pi_\star k}{\pi_\star \ell}\;,
\end{equ}
where the sum runs over all possible choices of 
$\ell$ such that $\pi_\star \ell$ is fixed.\qed
\end{lemma}

\begin{proof}[of Proposition~\ref{prop:Delta}]
We first show that for $\Gamma \in \CT$ the right hand side of \eqref{e:coaction} is well-defined
as an element of $\CH_- \otimes \CH$, which is a priori not obvious since we did not specify 
where the legs of $(\bar \Gamma, \ell)$ are attached. Our aim therefore is to show that,
for any fixed $L \in \N^d$, the expression
\begin{equ}[e:startExpr]
\sum_{\ell\colon \d\bar \Gamma \to \N^d \atop \Sigma\ell = L}
{(-1)^{|\out \ell|} \over \ell!} (\bar \Gamma, v, \pi\ell) \otimes \Gamma / (\bar \Gamma, \ell)
\end{equ}
is independent of $v \in \bar \CV$ in $\CH_- \otimes \CH$.
By Remark~\ref{rem:vanish}, we can restrict the sum over $\ell$ to those values
such that $\ell$ vanishes on the set $A_v$ of all half-edges incident to $v$ since 
$(\bar \Gamma, v, \ell) = 0$ in $\CH_-$ for those $\ell$ for which this is not the case.
Fixing some arbitrary $u \neq v$ and using \eqref{e:moveLeg} as well as Lemma~\ref{Chu-Vandermonde}, 
we then see that 
\eqref{e:startExpr} equals
\begin{equ}
\sum_{\ell\colon \d\bar \Gamma \setminus A_v \to \N^d \atop \Sigma\ell = L}
 \sum_{m \colon \d\bar \Gamma \to \N^d}{(-1)^{|\out \ell|+|m|} \over \ell!} \binom{\ell}{m} (\bar \Gamma, u, \pi\ell - \pi m + \Sigma m \one_{v}) \otimes \Gamma / (\bar \Gamma, \ell)\;.
\end{equ}
Writing $k = \ell-m$, we rewrite this expression as
\begin{equ}
\sum_{k\colon \d\bar \Gamma \setminus A_v \to \N^d \atop \Sigma k \le L}
 \sum_{m \colon \d\bar \Gamma \setminus A_v \to \N^d \atop \Sigma m = L-\Sigma k}{(-1)^{|\out k|+|\out m| + |m|} \over k! m!} (\bar \Gamma, u, \pi k + \Sigma m \one_{v}) \otimes \Gamma / (\bar \Gamma, k + m)\;.
\end{equ}
At this stage we note that, as a consequence of \eqref{e:IBP}, we have for every subset
$A \subset \d\bar \Gamma$ and every $M \in \N^d$ the identity
\begin{equ}
\sum_{m \colon \d\bar \Gamma \setminus A \to \N^d \atop \Sigma m = M}{(-1)^{|\out m|} \over m!} \Gamma / (\bar \Gamma, k + m)
= \sum_{n \colon A \to \N^d \atop \Sigma n = M} {(-1)^{|n|+|\out n|} \over n!} \Gamma / (\bar \Gamma, k + n)\;.
\end{equ}
Inserting this into the above expression and noting that for functions $n$ supported on $A_v$ one has
$\pi n = \Sigma n \one_{v}$, we conclude that it equals
\begin{equ}
\sum_{k\colon \d\bar \Gamma \setminus A_v \to \N^d \atop \Sigma k \le L}
 \sum_{n \colon A_v \to \N^d \atop \Sigma n = L-\Sigma k}{(-1)^{|\out k|+|\out n|} \over k! n!} (\bar \Gamma, u, \pi k + \pi n) \otimes \Gamma / (\bar \Gamma, k + n)\;.
\end{equ}
Setting $\ell = k+n$ and noting that $k!n! = (k+n)!$ since $k$ and $n$ have disjoint support, we see
that this is indeed equal to \eqref{e:startExpr} with $v$ replaced by $u$, as claimed.

It remains to show that $\Delta$ is well-defined on $\CH$, namely that $\Delta \tau = 0$ in
$\CH_-\otimes \CH$ for $\tau \in \d \CT$. Choose a 
Feynman diagram $\Gamma$, an inner vertex $v \in \CV_\star$, an index $i \in \{1,\ldots,d\}$,
and a subgraph $\bar \Gamma \subset \Gamma$. 
Writing $\bar A_v$ for the half-edges in $\bar \Gamma$ adjacent to
$v$ and $A_v$ for the remaining half-edges adjacent to $v$ (so that $A_v \subset \d \bar \Gamma$), 
it suffices to show that 
\begin{equs}[e:wantedVanish]
\sum_{h \in A_v} &\sum_{\ell\colon \d\bar \Gamma \to \N^d \atop
(\bar \Gamma, \ell) \in \CT_-}
{(-1)^{|\out \ell|} \over \ell!} (\bar \Gamma, \pi\ell) \otimes \d_{h}^{\delta_i}\Gamma / (\bar \Gamma, \ell)\\
&+
\sum_{h \in \bar A_v} \sum_{\ell\colon \d\bar \Gamma \to \N^d \atop
\d_{i,v}(\bar \Gamma, \ell) \in \CT_-}
{(-1)^{|\out \ell|} \over \ell!} \d_{h}^{\delta_i}(\bar \Gamma,\pi\ell) \otimes \Gamma / (\bar \Gamma, \ell)
= 0
\end{equs}
in $\CH_- \times \CH$, where we used the shorthand notation $\d_{i,v}(\bar \Gamma, \ell) \in \CT_-$
for the condition $\d_{h}^{\delta_i}(\bar \Gamma, \ell) \in \CT_-$, which is acceptable since
this condition does not depend on which half-edge $h$ one considers. 
If $v$ is not contained in $\bar \Gamma$, then the second term vanishes
and $A_v$ consists exactly of all edges adjacent to $v$ in $\Gamma / (\bar \Gamma, \ell)$,
so that the first term vanishes as well by \eqref{e:IBP}.
If $v$ is contained in $\bar \Gamma$, then
 we attach the leg of the corresponding connected component
$\bar \Gamma_0$ of $\bar \Gamma$ to $v$ itself, so that in particular the sum over $\ell$ can be restricted
to values supported on $\d\bar \Gamma \setminus A_v$. 
By \eqref{e:IBProot}, the second term is then equal to
\begin{equ}
\sum_{h \in \d\bar \Gamma_0 \setminus A_v} \sum_{\ell\colon \d\bar \Gamma \setminus A_v \to \N^d \atop
\d_{i,v}(\bar \Gamma, \ell) \in \CT_-}
{(-1)^{|\out \ell|} \over \ell!} \ell(h)_i (\bar \Gamma,v,\pi(\ell - \delta_i \one_{h})) \otimes \Gamma / (\bar \Gamma, \ell)\;,
\end{equ}
which can be rewritten as
\begin{equ}
\sum_{h \in \d\bar \Gamma_0 \setminus A_v} \sum_{\ell\colon \d\bar \Gamma \to \N^d \atop
(\bar \Gamma, \ell) \in \CT_-}
{(-1)^{|\out \ell| + \delta_{h \in \out}} \over \ell!} (\bar \Gamma,v,\pi \ell) \otimes \Gamma / (\bar \Gamma, \ell + \delta_i \one_{h})\;.
\end{equ}
Inserting this into \eqref{e:wantedVanish}, we conclude that this expression equals
\begin{equ}
\sum_{h \in \d \bar \Gamma_0} \sum_{\ell\colon \d\bar \Gamma \to \N^d \atop
(\bar \Gamma, \ell) \in \CT_-}
{(-1)^{|\out \ell|} \over \ell!} (\bar \Gamma, \pi\ell) \otimes \d_{h}^{\delta_i}\Gamma / (\bar \Gamma, \ell)
\end{equ}
which vanishes in $\CH_-\otimes \CH$ by \eqref{e:IBP} since the half-edges in $\d\bar\Gamma_0$
are precisely all the half-edges adjacent in $\Gamma / (\bar \Gamma, \ell)$ to the node that
$\bar \Gamma_0$ was contracted to. 
\end{proof}

For any element $g \colon \CH_- \to \R$ of the dual of $\CH_-$, we now have a 
linear map $M^g \colon \CH \to \CH$ by
\begin{equ}
M^g \Gamma = (g \otimes \id)\Delta \Gamma\;,
\end{equ}
which leads to a valuation $\Pi^K_g \colon \CH \to \CS$ by setting
\begin{equ}[e:defPig]
\Pi^K_g = \Pi^K \circ M^g
\end{equ}
as in \eqref{e:renormPi}, with $\Pi^K$ the canonical
valuation \eqref{e:eval}. 
Note that this is well-defined since $\Pi^K \d\CT = 0$, as already remarked.
In particular, we can also view $\Pi^K_g$ as a map from $\CT$ to $\CS$.

For any choice of $g$ (depending on the kernel assignment $K$), such 
a valuation then automatically satisfies Properties~\ref{prop:conv} and~\ref{prop:IBP}, since these were encoded
in the definition of the space $\CH$, as well as Property~\ref{prop:restr} since the action of $\Delta$
commutes with the operation of ``amputation of the $k$th leg'' on the subspace on which the latter is defined.
In general, such a valuation may
fail to satisfy Property~\ref{prop:hom},
but if we restrict ourselves to elements $g \colon \CH_- \to \R$ that are also characters,
one has
\begin{equ}
M^g (\Gamma_1 \bullet \Gamma_2) = (M^g \Gamma_1)\bullet (M^g \Gamma_2)\;.
\end{equ}
Since $\d \CT$ is an ideal, this implies that the valuation $\Pi^K_g$ is multiplicative as a 
map from $\CT$ to $\CS$, as required by Property~\ref{prop:hom}. We have therefore shown the following.

\begin{proposition}\label{prop:consistent}
For every character $g \colon \CH_- \to \R$, the valuation $\Pi^K_g$ is consistent for $K$ in the sense of
Definition~\ref{def:consistent}. \qed
\end{proposition}

Writing $\CG_-$ for the space of characters of $\CT_-$, it is therefore natural to define 
a ``consistent renormalisation procedure'' as a map
$\CR \colon \CK^-_\infty \to \CG_-$ such that the map 
\begin{equ}[e:defVal]
K \mapsto \hat \Pi^K = \Pi^K \circ M^{\CR(K)}\;,
\end{equ}
where $\Pi^K$ denotes the canonical valuation given by \eqref{e:eval}, extends continuously
to all of $\CK^-_0$. Our question now turns into the question whether such a map
exists.

\begin{remark}
We do certainly \textit{not} want to impose that $\CR$ extends continuously to
all of $\CK^-_0$ since this would then imply that $\Pi^K$ extends to all of $\CK^-_0$ which
is obviously false.
\end{remark}

\subsection{A Hopf algebra}
\label{sec:Hopf}

In this subsection, we address the following point. We have seen that every character $g$ of $\CH_-$ 
allow us to build a new valuation $\Pi_g$ from the canonical valuation $\Pi$ associated to a
smooth kernel assignment. We can then take a second character $h$ and build a new valuation
$\Pi_g \circ M^h$. It is natural to ask whether this would give us a genuinely new valuation
or whether this valuation is again of the form $\Pi_{\bar g}$ for some character $\bar g$.
In other words, does $\CG_-$ have a group structure, so that $g \mapsto M^g$ is a left 
action of this group on the space of all valuations?

In order to answer this question, we first define a map
$\Deltam \colon \hat \CT_- \to \CH_- \otimes \hat \CH_-$ in a way very similar to the map $\Delta$,
but taking into account the additional labels $\kk$:
\begin{equ}[e:coprod]
\Deltam (\Gamma,v_\star,\kk) = \sum_{\bar \Gamma \subset \Gamma} \sum_{\bar \ell\colon \d\bar \Gamma \to \N^d \atop \bar \kk\colon \bar \CV \to \N^d}
{(-1)^{|\out \bar \ell|} \over \bar \ell!}\binom{\kk}{\bar \kk} (\bar \Gamma, \bar \kk + \pi\bar \ell) \otimes (\Gamma,v_\star,\kk - \bar\kk) / (\bar \Gamma, \bar \ell)\;.
\end{equ}
Here, we define $(\Gamma,v_\star,\kk) / (\bar \Gamma, \bar \ell)$ similarly to before,
with the node-label of the quotient graph obtained by summing over the labels of all the 
nodes that get contracted to the same node.
If $\bar \Gamma$ completely contains one (or several) connected
components of $\Gamma$, then this definition could create graphs that contain isolated nodes,
which is forbidden by our definition of $\hat \CT_-$. Given \eqref{e:PiK}, it is natural
to identify isolated nodes with vanishing node-label with the empty diagram $\one$, while we identify those
with non-vanishing node-labels with $0$. In particular, it follows that
\begin{equ}
\Deltam \tau = \tau \otimes \one + \one \otimes \tau + \Deltamp\tau\;,
\end{equ}
where each of the terms appearing in $\Deltamp\tau$ is such that both factors contain at least
one edge.

Note the strong similarity with \cite[Def.~3.3]{BHZalg} which looks formally almost identical, 
but with graphs replaced by trees. As before, one then has 

\begin{proposition}
The map $\Deltam$ is well-defined both as a map $\hat \CH_- \to \CH_- \otimes \hat \CH_-$
and a map $\CH_- \to \CH_- \otimes \CH_-$. \qed
\end{proposition}

It follows immediately from the definitions that 
$\Deltam$ is multiplicative. What is slightly less obvious is that it also has
a nice coassociativity property as follows.

\begin{proposition}\label{prop:coass}
The identities
\begin{equ}[e:coassoc]
(\Deltam \otimes \id) \Delta = (\id \otimes \Delta) \Delta\;,\qquad
(\Deltam \otimes \id) \Deltam = (\id \otimes \Deltam) \Deltam
\end{equ}
hold between maps $\CB \to \CH_- \otimes \CB \otimes \CB$
for $\CB = \CH$ in the case of the first identity and
for $\CB \in \{\CH_-,\hat \CH_-\}$ in the case of the second one. 
\end{proposition}

\begin{proof}
We only verify the second identity since the first one is essentially a special case
of the second one. The difference is the presence of legs, which are never part of
the subgraphs appearing in the definition of $\Delta$, but otherwise play the same role as
a ``normal'' edge. 

Fix now a Feynman diagram $\Gamma$ as well as two subgraphs $\Gamma_1$ and $\Gamma_2$ with the property that 
each connected component of $\Gamma_1$ is either contained in $\Gamma_2$ or vertex-disjoint from it.
We also write $\bar \Gamma = \Gamma_1 \cup \Gamma_2$ and $\Gamma_{1,2} = \Gamma_1 \cap \Gamma_2$. 
There is then a natural bijection
between the terms appearing in $(\Deltam \otimes \id) \Deltam$ and those appearing in 
$(\id \otimes \Deltam) \Deltam$ obtained by noting that first extracting $\bar \Gamma$ from $\Gamma$
and then extracting $\Gamma_1$ from $\bar \Gamma$ is the same as first extracting $\Gamma_1$ from
$\Gamma$ and then extracting $\Gamma_2 / \Gamma_{1,2}$ from $\Gamma / \Gamma_1$.
It therefore remains to show that 
the labellings and combinatorial factors appearing for these terms are also the same.
This in turn is a consequence from a generalisation of the Chu-Vandermonde identity and can
be obtained in almost exactly the same way as \cite[Prop.~3.9]{BHZalg}.
\end{proof}

If we write $\one$ for the empty vacuum diagram and $\one^*$ for the 
element of $\CG_-$ that vanishes on all non-empty diagrams, then we see
that $(\CH_-, \Deltam, \bullet,\one,\one^*)$ is a bialgebra. Since it also
graded (by the number of edges of a diagram) and connected (the only diagram
with $0$ edges is the empty one), it is a Hopf algebra so that $\CG_-$ is
indeed a group with product
\begin{equ}
f \circ g \eqdef (f \otimes g)\Deltam\;,
\end{equ}
and inverse $g^{-1} = g\CA$, where $\CA$ is the antipode. The first identity in 
\eqref{e:coassoc} then implies that the map $g \mapsto M^g = (g \otimes \id)\Delta$
does indeed yield a group action on the space of valuations, thus answering 
positively the question asked at the start of this section. 

\subsection{Twisted antipodes and the BPHZ theorem}
\label{sec:twisted}

An arbitrary character $g$ of $\CH_-$ is uniquely determined by its value on
connected vacuum diagrams $\Gamma$ with $\deg\Gamma \le 0$.
Comparing \eqref{e:Msimple} with \eqref{e:coaction}, this would suggest that a natural choice
of renormalisation procedure $\CR$ is given by simply setting 
\begin{equ}
\CR(K)\Gamma = - \Pi_-^K\Gamma\;,
\end{equ}
as this would indeed reproduce the expression \eqref{e:renormSimple}. Unfortunately,
while this choice does yield valuations that extend continuously to all kernel assignments in 
$\CK^-_0$ for a class of ``simple'' Feynman diagrams, it fails to do so for all of them.

Following \cite{CK1,CK2}, a more sophisticated guess would be 
to set $\CR(K)\Gamma = \Pi_-^K\CA\Gamma$, for $\CA$ the antipode of $\CH_-$ endowed with
the Hopf algebra structure described in the previous section.
The reason why this identity also fails to do the trick can be illustrated with the
following example. Consider the case $d = 1$ and two labels with $|\Labhom_1| = -1/3$ and 
$|\Labhom_2| = -4/3$. Drawing edges decorated with $\Labhom_1$ in black and
edges decorated with $\Labhom_2$ in blue, we then consider
\begin{equ}
\Gamma = \<triangle>\;,
\end{equ}
which has degree $\deg\Gamma = 0$. Since $\Gamma$ has only one leg, the naive valuation
$\Pi^K \Gamma$ can be identified with the real number
\begin{equ}
\Pi^K \Gamma = (K_1 * K_2 * K_1)(0)\;,
\end{equ}
where we wrote $K_i \eqdef K_{\Labhom_i}$ and $*$ denotes convolution.
Since this might diverge for a generic kernel assignment in $\CK^-_0$, even if $K_2$ is
replaced by its renormalised version, there appears to be no good canonical renormalised value
for $\hat\Pi^K \Gamma$, so we would expect to just have $\hat\Pi^K \Gamma = 0$.

Let's see what happens instead if 
we choose the renormalisation procedure $\CR(K)\Gamma = \Pi_-^K\CA\Gamma$.
It follows from the definition of $\Delta$ that
\begin{equ}[e:exDelta]
\Delta \Gamma = \one \otimes \<triangle> + \<line> \otimes \<loop>
+ \<baretriang> \otimes \<leg>\;,
\end{equ}
since \<line> and \<baretriang> are the only subgraphs of negative degree, but their
degree remains above $-1$ so that no node-decorations are added.
Note furthermore that in $\CH_-$ one has the identities
\begin{equ}
\Deltam \<line> = \<line> \otimes \one + \one \otimes \<line>\;,\quad
\Deltam \<baretriang> = \<baretriang> \otimes \one + \one \otimes \<baretriang>\;.
\end{equ}
The reason why there is no additional term analogous to the middle term of
\eqref{e:exDelta} appearing in the second identity is that the corresponding factor would be
of positive degree and therefore vanishes when viewed as an element of $\CH_-$.
As a consequence, we have $\CA \tau = -\tau$ in both cases, so that the first and last terms
of \eqref{e:exDelta} cancel out and we are eventually left with
\begin{equ}
\hat \Pi^K \Gamma = - (K_1 * K_1)(0)\cdot K_2(0)\;,
\end{equ}
which is certainly not desirable since it might diverge as well.

The way out of this conundrum is to define a \textit{twisted antipode} $\hat \CA \colon \CH_- \to \hat \CH_-$
which is defined by a relation very similar to that defining the antipode, but this time
guaranteeing that the renormalised valuation vanishes on those diagrams that encode
``potentially diverging constants'' as above. Here, the renormalised valuation is defined by
setting 
\begin{equ}[e:properRenorm]
\CR(K)\Gamma = \Pi_-^K \hat \CA \Gamma\;,
\end{equ}
where $\Pi_-^K$ is defined by \eqref{e:PiK}. Writing $\CM\colon \hat \CH_-\otimes \hat \CH_- \to \hat \CH_-$
for the product, we define $\hat \CA$ to be such that 
\begin{equ}[e:twisted]
\CM (\hat \CA \otimes \id)\Deltam \Gamma = 0\;,
\end{equ}
for every non-empty connected vacuum diagram $\Gamma \in \hat \CH_-$ with $\deg\Gamma \le 0$.
At first sight, this looks exactly like the definition of the antipode. The difference is
that the map $\Deltam$ in the above expression goes from 
$\hat \CH_-$ to $\CH_- \otimes \hat \CH_-$, so that no projection onto diverging diagrams
takes place on the right factor. If we view $\CH_-$ as a subspace of $\hat \CH_-$, then the 
antipode satisfies the identity 
\begin{equ}
\CM (\CA \otimes \pi_-)\Deltam \Gamma = 0\;,
\end{equ}
where $\pi_-\colon \hat \CH_- \to \CH_-$ is the projection given by quotienting by the
ideal $\CJ_+$ generated by diagrams with strictly positive degree.
We have the following simple lemma.

\begin{lemma}
There exists a unique map $\hat \CA \colon  \CH_- \to \hat \CH_-$ satisfying \eqref{e:twisted}.
Furthermore, the map $\Pi^K_\BPHZ$ given by \eqref{e:defVal} with $\CR(K) = \Pi_-^K \hat\CA$ is indeed
a valuation.
\end{lemma}

\begin{proof}
The existence and uniqueness of $\hat \CA$ is immediate by performing an
induction over the number of edges. Defining $\Deltak{k}\colon \CH_- \to \hat \CH_-^{\otimes (k+1)}$
inductively by $\Deltak{0} = \iota$ and then 
\begin{equ}
\Deltak{k+1} = (\Deltak{k} \otimes \id)\Deltamp\iota\;,
\end{equ}
where $\iota \colon \CH_- \to \hat \CH_-$ is the canonical injection, one obtains 
the (locally finite) Neumann series
\begin{equ}[e:reprAhat]
\hat \CA = \sum_{k \ge 0} (-1)^{k+1} \CM^{(k)}\Deltak{k}\;,
\end{equ}
where $\CM^{(k)} \colon \hat \CH_-^{\otimes (k+1)} \to \hat \CH_-$ is the multiplication
operator.
The uniqueness also immediately implies
that $\hat \CA$ is multiplicative, so that $\CR(K)$ as defined above is indeed a character
for every $K \in \CK^-_\infty$. 
\end{proof}

\begin{definition}
We call the renormalisation procedure defined by $\CR(K) = \Pi_-^K \hat\CA$ the ``BPHZ renormalisation''.
\end{definition}

It follows from \eqref{e:reprAhat} that in the above example
the twisted antipode satisfies
\begin{equ}
\hat\CA \<baretriang> = - \<baretriang> + \<bareloop>\;\<line>\;,
\end{equ}
so that
\begin{equ}
(\hat\CA \otimes \id)\Delta\Gamma 
= \one \otimes \<triangle> - \<line> \otimes \<loop>
- \<baretriang> \otimes \<leg> + \<bareloop>\;\<line> \otimes \<leg>\;,
\end{equ}
which makes it straightforward to verify that indeed 
$\Pi^K_\BPHZ \Gamma = 0$.
The following general statement should make it clear that this is indeed
the ``correct'' way of renormalising Feynman diagrams.

\begin{proposition}\label{prop:poly}
The BPHZ renormalisation is characterised by the fact that, for every $k \ge 1$ and 
every connected Feynman diagram $\Gamma$ with $k$ legs and $\deg\Gamma \le 0$, there 
exists a constant $C$ such that if $\phi$ is a test function on $\S^{k}$
of the form $\phi= \phi_0\cdot \phi_1$ 
such that $\phi_1$ depends only on $x_1+\ldots+x_k$, $\phi_0$ depends only on the differences of the $x_i$, 
and there exists a polynomial 
$P$ with $\deg P + \deg \Gamma \le 0$ and
\begin{equ}[e:propPhi0]
\phi_0(x_1,\ldots,x_{k}) = P(x_2-x_1,\ldots,x_{k}-x_1)\;,\qquad |x| \le C\;,
\end{equ}
then $\bigl(\Pi^K_\BPHZ \Gamma\bigr)(\phi) = 0$.
\end{proposition}

\begin{remark}
One way to interpret this statement is that, once we have defined $\Pi^K_\BPHZ \Gamma$
for test functions in $\CD_k^{(\CA)}$ with $\CA = \{\{1,\ldots,k\}\}$, 
the canonical way of extending it to all test functions is to 
subtract from it the linear combination of derivatives of delta functions which has precisely the
same effect when testing it against all polynomials of degree at most $-\deg \Gamma$.
\end{remark}

\begin{proof}
The statement follows more or less immediately from the following observation.
Take a valuation of the form $\Pi^K_g$ as in \eqref{e:defPig} for some $K \in \CK^-_\infty$
and some $g \in \CG_-$. Fixing the Feynman diagram $\Gamma$ from the statement, we
write $\d \Gamma = \{[1],\ldots,[k]\}$ for its $k$ legs, and we fix
a function $\Labn \colon \d\Gamma \to \N^d$ with $|\Labn| + \deg\Gamma \le 0$. Write furthermore 
$\bar \Labn \colon \d\Gamma \to \N^d$ for the function such that the $\ell$th leg has label
$\delta^{(\bar \Labn([\ell]))}$. We assume without loss of generality that $\bar \Labn([1]) = 0$
since we can always reduce ourselves to this case by \eqref{e:IBP}.

Let then $P$ be given by
\begin{equ}
P(x) = P_\Labn(x) = \prod_{i=[2]}^{[k]} (x_{i} - x_{[1]})^{\Labn(i)}\;,
\end{equ}
let $\phi_0$ be as in \eqref{e:propPhi0}, and let $\phi_1$ be a test function depending only on the
sums of the coordinates and integrating to $1$.
We then claim that, writing $\bar \Gamma \subset \Gamma$ for the
maximal subgraph where we only discarded the legs and $v_\star$ for the vertex of $\Gamma$ incident to 
the first leg, one has
\begin{equ}
\bigl( \Pi_g^K \Gamma\bigr)(\phi) = \binom{\Labn}{\bar \Labn}  (g\otimes \Pi_-^K) \Deltam (\bar \Gamma,v_\star,\pi(\Labn - \bar \Labn))\;.
\end{equ}
(In particular  
$\bigl( \Pi^K_g \Gamma\bigr)(\phi) = 0$ unless $\bar \Labn \le \Labn$.) Indeed, comparing
\eqref{e:evalBis} to \eqref{e:PiK}, it is clear that this is the case when $g = \one^*$, noting that
\begin{equ}[e:derPoly]
D_2^{\bar \Labn([2])}\cdots D_k^{\bar \Labn([k])} P_\Labn =  \binom{\Labn}{\bar \Labn} P_{\Labn - \bar \Labn}\;.
\end{equ}
The general case then follows by comparing the definitions of $\Delta$ and $\Deltam$, noting that 
by \eqref{e:derPoly} the effect of the label $\bar \Labn$ in \eqref{e:coprod} is exactly the same 
of that of the components of $\ell$ supported on the ``legs'' in \eqref{e:coaction}. In other words, when
comparing the two expressions one should set $\bar \ell(h) = \ell(h)$ for the half-edges $h$ that are not legs
and $\bar \Labn(v) = \sum \ell(e,v)$, where the sum runs over all legs (if any) adjacent to $v$.

The claim now follows immediately from the definition of the twisted antipode and the BPHZ
renormalisation:
\begin{equs}
\bigl(\Pi^K_\BPHZ \Gamma\bigr)(\phi)
&= \binom{\Labn}{\bar \Labn}  (\Pi_-^K \hat\CA\otimes \Pi_-^K) \Deltam (\bar \Gamma,v_\star,\pi(\Labn - \bar \Labn)) \\
&= \binom{\Labn}{\bar \Labn}  \Pi_-^K \CM(\hat\CA\otimes \id) \Deltam (\bar \Gamma,v_\star,\pi(\Labn - \bar \Labn)) = 0\;,
\end{equs}
since the degrees of $\Gamma$ and of $(\bar \Gamma,v_\star,\pi(\Labn - \bar \Labn))$ agree
(and are negative) by definition.
\end{proof}

\section{Statement and proof of the main theorem}
\label{sec:analytical}

We now have all the definitions in place in order to be able to state the BPHZ theorem.

\begin{theorem}\label{theo:BPHZ}
The valuation $\Pi^K_\BPHZ$ 
is consistent for $K$ and extends continuously to all $K \in \CK^-_0$.
\end{theorem}

By Proposition~\ref{prop:consistent}, we only need to show the continuity part of the statement.
Before we turn to the proof, we give an explicit formula for the valuation $\Pi^K_\BPHZ$
instead of the implicit characterisation given by \eqref{e:twisted}. This is nothing but
Zimmermann's celebrated ``forest formula''. 

\subsection{Zimmermann's forest formula}

So what are
these ``forests'' appearing in the eponymous formula? Given any Feynman diagram $\Gamma$, the set 
$\CG_\Gamma^-$ of all \textit{connected} vacuum diagrams $\bar \Gamma \subset \Gamma$ 
with $\deg\bar \Gamma \le 0$ is endowed with a natural
partial order given by inclusion. A subset $\CF \subset \CG_\Gamma^-$ is called a ``forest''
if any two elements of $\CF$ are either comparable in $\CG_\Gamma^-$ or vertex-disjoint
as subgraphs of $\Gamma$. 

Given a forest $\CF$ and a subgraph $\bar \Gamma \in \CF$, we say that 
$\bar \Gamma_1$ is a \textit{child} of $\bar \Gamma$ if $\bar \Gamma_1 < \bar \Gamma$ and
there exists
no $\bar \Gamma_2 \in \CF$ with $\bar \Gamma_1 <\bar \Gamma_2 < \bar \Gamma$.
Conversely, we then say that $\bar \Gamma$ is $\bar \Gamma_1$'s parent. (The forest structure
of $\CF$ guarantees that its elements have at most one parent.)
An element without children is called a \textit{leaf} and one without parent a \textit{root}.
If we connect parents to their children in $\CF$, then it does indeed
form a forest with arrows pointing away from the roots and towards the leaves.
We henceforth write $\CF_\Gamma^{-}$ for the set of all forests for $\Gamma$.

Given a diagram $\Gamma$, we now consider the space $\CT_\Gamma$ generated by all diagrams
$\hat \Gamma$
such that each connected component has \textit{either} at least one leg \textit{or} a distinguished
vertex $v_\star$, but not both. We furthermore endow $\hat \Gamma$ with an $\N^d$-valued vertex decoration $\kk$
supported on the leg-less components and, most importantly, with a bijection 
$\tau \colon \hat\CE \to \CE$ between the edges 
of $\hat \Gamma$ and those of $\Gamma$, such that legs get mapped to legs. 
The operation of discarding $\tau$  
yields a natural injection $\CT_\Gamma \hookrightarrow \hat\CT_- \otimes \CT$
by keeping the components with a distinguished vertex in the first factor and those with legs in the
second factor. (The space $\CT_\Gamma$ itself however is not a tensor product due to the constraint that 
$\tau$ is a  bijection, which exchanges information between the two factors.) 
We can also define $\d \CT_\Gamma$ analogously to \eqref{e:IBP} and \eqref{e:IBPbis}--\eqref{e:moveLeg}, so that 
$\CH_\Gamma = \CT_\Gamma / \d \CT_\Gamma$ naturally injects into $\hat \CH_- \otimes \CH$.

Given a connected subgraph $\gamma \subset \Gamma$, we then define a contraction operator
$\CC_{\gamma}$ acting on $\CH_\Gamma$ in the following way.
Given an element $(\hat \Gamma,\kk) \in \CT_\Gamma$, we write $\hat \gamma$ for the subgraph of 
$\hat \Gamma$ such that $\tau$ is a bijection between the edges of $\hat \gamma$ and those of $\gamma$.
If $\hat \gamma$ is not connected, then we set $\CC_{\gamma} (\hat \Gamma,\kk) = 0$.
Otherwise, we set as in \eqref{e:coprod}
\begin{equ}[e:contraction]
\CC_{\gamma} (\hat \Gamma,\kk) = 
 \sum_{\bar \ell\colon \d \gamma \to \N^d \atop \bar \kk\colon \CV_\gamma \to \N^d}
{(-1)^{|\out \bar \ell|} \over \bar \ell!}\one_{\deg (\hat \gamma, \bar \kk + \pi\bar \ell) \le 0}\binom{\kk}{\bar \kk} (\hat \gamma, \bar \kk + \pi\bar \ell) \cdot (\hat \Gamma,\kk - \bar\kk) / (\hat \gamma, \bar \ell)\;,
\end{equ}
with the obvious bijections between the edges of $\hat \gamma \cdot \hat \Gamma / \hat \gamma$ and  
those of $\Gamma$. This time we explicitly include the restriction to terms such that
$\deg (\hat \gamma, \bar \kk + \pi\bar \ell) \le 0$, which replaces the projection to $\CH_-$ in
\eqref{e:coprod}.
An important fact is then the following.

\begin{lemma}\label{lem:commute}
Let $\gamma_1, \gamma_2$ be two subgraphs of $\Gamma$ that are vertex-disjoint 
and let $\hat \Gamma \in \CT_\Gamma$ be such that $\hat \gamma_1$ and $\hat \gamma_2$ are
vertex disjoint. Then $\CC_{\gamma_1} \CC_{\gamma_2}\hat \Gamma = \CC_{\gamma_2} \CC_{\gamma_1}\hat \Gamma$.\qed
\end{lemma}

We will use the natural convention that $\emptyset \in \CF_\Gamma^{-}$.
For any $\CF \in \CF_\Gamma^{-}$, we then write $\CC_\CF \Gamma$ for the 
element of $\CH_\Gamma$ defined recursively in the following way. 
For $\CF = \emptyset$, we set $\CC_\emptyset \Gamma = \Gamma$.
For non-empty $\CF$, we write $\rho(\CF) \subset \CF$ for the set of roots of $\CF$ 
and we set recursively
\begin{equ}
\CC_\CF \Gamma
= \CC_{\CF \setminus \rho(\CF)} \prod_{\gamma \in \rho(\CF)} \CC_\gamma \Gamma\;.
\end{equ}
The order of the product doesn't matter by Lemma~\ref{lem:commute}, since the 
 roots of $\CF$ are all vertex-disjoint.
With these notations at hand, Zimmermann's forest formula \cite{Zimmermann} then reads

\begin{proposition}\label{prop:forest}
The BPHZ renormalisation procedure is given by the identity
\begin{equ}[e:forest]
(\hat \CA\otimes \id)\Delta \Gamma = \CR \Gamma \eqdef \sum_{\CF \in \CF_\Gamma^-} (-1)^{|\CF|} \CC_\CF \Gamma\;,
\end{equ}
where we implicitly use the injection $\CH_\Gamma  \hookrightarrow \hat \CH_- \otimes \CH$ 
for the right hand side.
\end{proposition}

\begin{proof}
This follows from the representation \eqref{e:reprAhat}. Another way of seeing it is to 
first note that $\CR$ is indeed of the form $(\CB \otimes \id)\Delta \Gamma$ for \textit{some}
$\CB \colon \CH_-\to \hat \CH_-$ and to then make use of
the characterisation \eqref{e:twisted} of the twisted antipode $\hat \CA$. This implies that
it suffices to show that $\CR \Gamma = 0$ 
for every connected $\Gamma$ with a distinguished vertex and a node-labelling such that 
$\deg\Gamma \le 0$.

The idea is to observe that $\CF_\Gamma^-$ can be partitioned into two disjoint sets that are in
bijection with each other: those that contain $\Gamma$ itself and the complement $\hat \CF_\Gamma^-$
of those forest that don't.
Furthermore, it follows from the definition that $\CC_\Gamma \Gamma = \Gamma$, so that 
\begin{equ}
\sum_{\CF \in \CF_\Gamma^-} (-1)^{|\CF|} \CC_\CF \Gamma
= \sum_{\CF \in \hat \CF_\Gamma^-} (-1)^{|\CF|} \bigl(\CC_\CF \Gamma - \CC_{\CF \cup \{\Gamma\}} \Gamma)
= \sum_{\CF \in \hat \CF_\Gamma^-} (-1)^{|\CF|} \bigl(\CC_\CF \Gamma - \CC_{\CF} \Gamma)\;,
\end{equ}
which vanishes thus completing the proof.
\end{proof}

In order to analyse \eqref{e:forest}, it will be very convenient to have ways of resumming its terms
in order to make cancellations more explicit.
These resummations are based on the following trivial identity. Given a finite set $A$ and operators 
$X_i$ with $i \in A$, one has 
\begin{equ}[e:prod]
\prod_{i\in A}(\id -X_i) = \sum_{B \subset A} (-1)^B \prod_{j\in B}X_j\;,
\end{equ}
provided that the order
in which the operators are composed is the same in each term and that the empty product is interpreted
as the identity. The right hand side of this expression is 
clearly reminiscent of \eqref{e:forest} while the left hand side encodes cancellations if the $X_i$ are close
to the identity in some sense. If $\CG_\Gamma^-$ itself happens to 
be a forest, then $\CF_\Gamma^-$ consists simply of all subsets of  $\CG_\Gamma^-$, so that 
one can indeed write
\begin{equ}[e:simpleCase]
(\hat \CA\otimes \id)\Delta  \Gamma = \CR_{\CG_\Gamma^-} \Gamma\;,
\end{equ}
where $\CR_\CF \Gamma$ is defined by $\CR_\emptyset \Gamma = \Gamma$ and then via the recursion
\begin{equ}[e:defRGamma]
\CR_\CF \Gamma
= \CR_{\CF \setminus \rho(\CF)} \prod_{\gamma \in \rho(\CF)} (\id - \CC_\gamma) \Gamma\;.
\end{equ}
In general however this is not the case, and this is precisely the problem of
``overlapping divergences''. In order to deal with this, we introduce the following variant
of \eqref{e:simpleCase} which still works in the general case. To formulate it, we introduce
the notion of a ``forest interval'' $\M$ for $\Gamma$ which is a subset of $\CF_\Gamma^-$
of the form $[\underline \M,\overline \M]$ in the sense that it consists precisely 
of all those forests $\CF \in \CF_\Gamma^-$ such that $\underline \M \subset \CF \subset \overline\M$.
An alternative description of $\M$ is that there is a forest $\delta(\M) = \overline\M \setminus \underline\M$ 
disjoint from $\underline \M$ and such that
$\M$ consists of all forests of the type $\underline \M \cup \CF$ with $\CF \subset \delta(\M)$.
Given a forest interval, we define an operation $\CR_\M$ which renormalises all subgraphs
in $\delta(\M)$ and contracts those subgraphs in $\underline \M$. In other words, we set
$\CR_\M = \CR_\M^{\overline \M}$, where $\CR_\M^{\CF}$ is defined recursively by
\begin{equ}
\CR_\M^\CF \Gamma
= \CR_\M^{\CF \setminus \rho(\CF)} \prod_{\gamma \in \rho(\CF)} \CC_\gamma^\sharp \Gamma\;,\qquad
\CC_\gamma^\sharp = 
\left\{\begin{array}{cl}
	\id - \CC_\gamma & \text{if $\gamma \in \delta(\M)$,} \\
	- \CC_\gamma & \text{otherwise.}
\end{array}\right.
\end{equ}
This definition is consistent with \eqref{e:defRGamma} in the sense that one
has $\CR_\CF = \CR_\M$ for $\M = [\emptyset,\CF]$. Combining Proposition~\ref{prop:forest} with \eqref{e:prod},
we then obtain the following alternative characterisation of our renormalisation map.

\begin{lemma}\label{lem:resum}
Let $\Gamma$ be a Feynman diagram and let $\CP$ be a partition of $\CF_\Gamma^-$
consisting of forest intervals. Then, one has the identity
$
(\hat \CA\otimes \id)\Delta  \Gamma = \sum_{\M \in \CP} \CR_{\M} \Gamma
$.\qed
\end{lemma}

\subsection{Proof of the BPHZ theorem, Theorem~\ref{theo:BPHZ}}

We now have all the ingredients in place to prove Theorem~\ref{theo:BPHZ}.
We only need to show that for every (connected) Feynman diagram $\Gamma$
there are constants $C_\Gamma$ and $N_\Gamma$ such that for every test function $\phi$ with 
compact support in the ball of radius $1$ 
one has the bound
\begin{equ}[e:mainBound]
\bigl|\bigl(\Pi^K_\BPHZ \Gamma\bigr)(\phi)\bigr| \le C_\Gamma \prod_{e \in \CE} |K_{\Labhom(e)}|_{N_\Gamma}
 \sup_{|k| \le N_\Gamma} \|D^{(k)}\phi\|_{L^\infty}\;,
\end{equ}
where $|K_\Labhom|_N$ denotes the smallest constant $C$ such that  \eqref{e:propKer}
holds for all $|k| \le N$. 

The proof of \eqref{e:mainBound} follows the same lines as that of the main result in 
\cite{BPHZana}, but with a  number of considerable simplifications:
\begin{claim}
\item There is no
``positive renormalisation'' in the present context so that we do not need to worry about
overlaps between positive and negative renormalisations. As a consequence, we also do not
make any claim on the behaviour of \eqref{e:mainBound} when rescaling the test function. 
In general, it is \textit{false} that \eqref{e:mainBound} obeys the naive power-counting
when $\phi$ is replaced by $\phi^\lambda$ and $\lambda\to0$ as in \cite[Lem.~A.7]{KPZJeremy}. 
\item The BPHZ renormalisation procedure studied in the present article is directly formulated at the level of
graphs. In \cite{BHZalg,BPHZana} on the other hand, it is formulated at the level of trees 
(which are the objects indexing a suitable family of stochastic processes) and then has to be translated into
a renormalisation procedure on graphs which, depending on how trees are glued together in order to form
these graphs, 
creates additional ``useless'' terms.
\item We only consider kernels with a single argument, corresponding to ``normal'' edges in
our graphs, while  \cite{BPHZana} deals with non-Gaussian processes which then gives rise to 
Feynman diagrams containing some ``multiedges''.
\end{claim}
We therefore only give an overview of the main steps, but we hope that the style of our exposition is
such that the interested reader will find it possible to fill in the missing details without 
undue effort.

As in the proof of Proposition~\ref{prop:WeinFancy},
we break the domain of integration into Hepp sectors $D_\bT$ and we estimate 
terms separately on each sector. The main trick is then to resum the
terms as in Lemma~\ref{lem:resum}, but by using a partition $\CP_\bT$ that is
adapted to the Hepp sector $\bT$ in such a way that the occurrences of
$(\id - \CC_\gamma)$ create cancellations that are useful on $D_\bT$.

In order to formulate this, it is convenient to write all the terms appearing in
the definition of $\Pi^K_\BPHZ \Gamma$
as integrals over the same set of variables. For this, we henceforth fix a connected Feynman diagram 
$\Gamma$ once and for all, together with an arbitrary total order for its vertices. 

We then define the space $\hat \CT_\Gamma$ generated by connected Feynman diagrams
$\bar \Gamma$ with edges \textit{and vertices}
in bijection with those of $\Gamma$ via a map $\tau \colon (\bar \CE, \bar \CV) \to (\CE,\CV)$, 
together with a vertex labelling
$\Labn$, as well as a map $\c\colon \bar \CE \to \N$ which vanishes on all legs of $\bar \Gamma$. 
The goal of this map is to allow us to keep track on which parts of $\Gamma$ were contracted,
as well as the structure of nested contractions: $\c$ measures how ``deep'' a given edge lies
within nested contractions. In particular, it is natural to impose that $\c$ vanishes
on legs since they are never contracted.
We furthermore
impose that for every $j > 0$, every connected component $\hat \gamma$ of $\c^{-1}(j)$
has the following two properties. 
\begin{claim}
\item The highest vertex $v_{\star}(\hat \gamma)$ of $\hat \gamma$ has an incident edge $e$
 with $\c(e) < j$. (Here, ``highest'' refers to the total order we fixed on vertices of $\Gamma$,
 which is transported to $\bar \Gamma$ by the bijection between vertices of $\Gamma$ and $\bar \Gamma$.)
\item All edges $e$ incident to a vertex of $\hat \gamma$ other than $v_{\star}(\hat \gamma)$
satisfy $\c(e) \ge j$.
\end{claim}
Writing $\bar \CV^c \subset \bar \CV$ for those vertices $v$ with at least one edge  
$e$ incident to $v$ such that $\c(e) > 0$, we also
impose that $\Labn(v) = 0$ for $v \not \in \bar \CV^c$. 
We view $\Gamma$ itself as an element of $\hat \CT_\Gamma$ by setting $\c \equiv 0$.
Note that this data defines a map $v \mapsto v_\star$ from $\bar \CV^c$ to $\bar \CV^c$ 
such that $v \mapsto v_{\star}(\hat \gamma)$ for $\hat \gamma$
the connected component of $\c^{-1}(j)$ with the lowest possible value of $j$ containing
$v$.

For $\gamma \subset \Gamma$ as above, we then define maps $\hat \CC_\gamma$
on $\hat \CT_\Gamma$ similarly to \eqref{e:contraction}. This time however,
we set $\hat \CC_\gamma \bar \Gamma = 0$ unless the following conditions are met.
\begin{claim}
\item The graph $\tau^{-1}(\gamma) \subset \bar \Gamma$ is connected. 
\item For every edge $e$ adjacent to $\tau^{-1}(\gamma)$, one has 
$\c(e) \le \inf_{\hat e \in \hat \CE} \c(\hat e)$.
\end{claim}
We also restrict the sum over labels $\ell$ supported on edges with $\c(e) = \inf_{\hat e \in \hat \CE} \c(\hat e)$.
In order to remain in $\hat \CT_\Gamma$, instead
of extracting $\hat \gamma = \tau^{-1}(\gamma)$, we reconnect the edges of $\bar \Gamma$ 
adjacent to $\hat \gamma$ to the highest vertex $\hat v$ of $\hat \gamma$
and we increase $\c(e)$ by $1$ on all edges $e$ of $\hat \gamma$.
We similarly define elements $\hat \CR_\M \Gamma$ as above with every instance of $\CC_\gamma$
replaced by $\hat \CC_\gamma$.
We also view $\Gamma$ itself as an element of $\hat \CT_\Gamma$ by setting both $\c$  and 
$\Labn$ to $0$. 

Let us illustrate this by taking for $\Gamma$ the diagram of Figure~\ref{fig:subgraph}
and for $\gamma$ the triangle shaded in grey. In this case, assuming that the order on
our vertices is such that the first vertex is the leftmost one and that the degree of 
$\gamma$ is above $-1$ so that no node-decorations are needed, we have
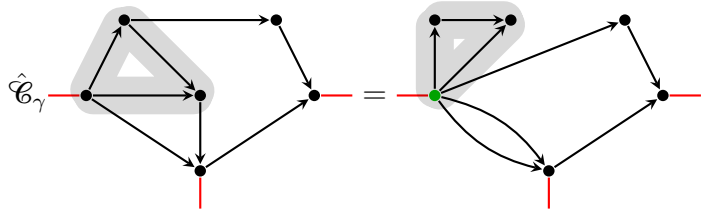
\begin{equ}
\hat \CC_\gamma
\begin{tikzpicture}[style={thick},baseline=-0.1cm]
\node[dot] (l) at (0,0) {};
\node[dot] (r) at (3,0) {};
\node[dot] (ul) at (0.5,1) {};
\node[dot] (ur) at (2.5,1) {};
\node[dot] (d) at (1.5,-1) {};
\node[dot] (c) at (1.5,0) {};
\draw[->] (l) -- (ul);
\draw[->] (ul) -- (ur);
\draw[->] (ur) -- (r); 
\draw[->] (l) -- (c);
\draw[->] (ul) -- (c);
\draw[->] (l) -- (d);
\draw[->] (c) -- (d);
\draw[->] (d) -- (r);
\draw[thick,red] (l) -- ++(180:0.5);
\draw[thick,red] (r) -- ++(0:0.5);
\draw[thick,red] (d) -- ++(-90:0.5);
\draw[line width=0.5cm,draw opacity=0.15, line cap=round] (l) -- (ul) -- (c) -- (l);
\end{tikzpicture} 
=
\begin{tikzpicture}[style={thick},baseline=-0.1cm]
\node[dot,boundary] (l) at (0,0) {};
\node[dot] (r) at (3,0) {};
\node[dot] (ul) at (0,1) {};
\node[dot] (ur) at (2.5,1) {};
\node[dot] (d) at (1.5,-1) {};
\node[dot] (c) at (1,1) {};
\draw[->] (l) -- (ul);
\draw[->] (l) -- (ur);
\draw[->] (ur) -- (r); 
\draw[->] (l) -- (c);
\draw[->] (ul) -- (c);
\draw[->] (l) to[bend right=20] (d);
\draw[->] (l) to[bend left=20] (d);
\draw[->] (d) -- (r);
\draw[thick,red] (l) -- ++(180:0.5);
\draw[thick,red] (r) -- ++(0:0.5);
\draw[thick,red] (d) -- ++(-90:0.5);
\draw[line width=0.5cm,draw opacity=0.15, line cap=round] (l) -- (ul) -- (c) -- (l);
\end{tikzpicture} 
\end{equ}
with $\c(v)$ equal to $1$ in the shaded region of the diagram on the right. 
The green node then denotes the 
element $v_\star$ for all the nodes $v$ in that region.
This time, it follows in virtually the same way as the proof of Proposition~\ref{prop:coass} 
that if $\gamma_1$ and $\gamma_2$ are either vertex disjoint
or such that one is included in the other, then the operators
$\hat \CC_{\gamma_1}$ and $\hat \CC_{\gamma_2}$ commute.
In particular, we can simply write 
\begin{equ}[e:defRM]
\hat \CR_\M \Gamma
= \Big(\prod_{\gamma \in \delta(\M)}(\id - \hat \CC_{\gamma}) \prod_{\bar\gamma \in \underline\M} 
(-\hat \CC_{\bar\gamma})\Big)\Gamma \;,
\end{equ}
without having to worry about the order of the operations as in \eqref{e:defRGamma}.

For every $K \in \CK^-_\infty$ and every test function $\phi$, we then have a 
linear map $\CW^K \colon \hat \CT_\Gamma \to \CC^\infty(\S^{\CV_\star})$
given by
\begin{equs}
\bigl(\CW^K \bar \Gamma\bigr) (x) &= 
\prod_{e \in \bar \CE_\star} K_{\Labhom(e)}(x_{\tau(e_+)} - x_{\tau(e_-)})
\prod_{v \in \bar \CV_\star} (x_{\tau(v)} - x_{\tau(v_\star)})^{\Labn(v)}\\
&\qquad \times \bigl(D_1^{\ell_1}\cdots D_k^{\ell_k} \phi\bigr)(x_{v_1},\ldots,x_{v_k})\;,
\end{equs}
where $\tau \colon \bar \CV \cup \bar \CE \to \CV \cup \CE$ is the bijection between
edges and vertices of $\bar \Gamma$ and those of $\Gamma$,
$v_i$ are the vertices to which the $k$ legs of $\Gamma$ are attached,
and $\ell_i$ are the corresponding multiindices as in \eqref{e:evalBis}.
With this notation, our definitions show that, for every partition $\CP$
of $\CF_\Gamma^-$ into forest intervals, one has
\begin{equ}
\bigl(\Pi^K_\BPHZ \Gamma\bigr)(\phi) =
\sum_{\M \in \CP} \int_{\S^{\CV_\star}} \bigl(\CW^K\hat \CR_{\M} \Gamma\bigr)(x)\,dx\;.
\end{equ}
We bound this rather brutally by
\begin{equs}[e:terms]
\bigl|\bigl(\Pi^K_\BPHZ \Gamma\bigr)(\phi)\bigr| &\le
\sum_{\bT} \sum_{\M \in \CP_\bT} \int_{D_\bT} \bigl|\bigl(\CW^K\hat \CR_{\M} \Gamma\bigr)(x)\bigr|\,dx \\
&\le
\sum_{\bT} \sum_{\M \in \CP_\bT} 
\sup_{x \in D_\bT} \bigl|\bigl(\CW^K\hat \CR_{\M} \Gamma\bigr)(x)\bigr|\prod_{u \in T} 2^{-d\bn_u}\;.
\end{equs}
At this stage, we would like to make a smart choice for the partition
$\CP_\bT$ which allows us to obtain a summable bound for this expression. 
In order to do this, we would like to guarantee that a cancellation
$(\id - \hat \CC_\gamma)$ appears for all of the subgraphs $\gamma$ that are
such that the length of all adjacent edges (as measured by the quantity
$|x_{\tau(e_+)} - x_{\tau(e_-)}|$) is much greater than the diameter of $\gamma$
(measured in the same way).

In order to achieve this, we first note that by Proposition~\ref{prop:full} and 
\eqref{e:forestFormulaFull} below, we can restrict ourselves in 
\eqref{e:terms} to the case where $\CP_\bT$ is a partition of the subset $\hat \CF_\Gamma^- \subset \CF_\Gamma^-$
of all forests containing only subgraphs that are full in $\Gamma$. 
(Recall that a subgraph $\bar\gamma\subset \Gamma$ is full in $\Gamma$ 
if it is induced by a subset of the vertices of $\Gamma$ in the sense that it consists of all edges
of $\Gamma$ connecting two vertices of the subset in question.)
We then consider the following construction. For 
any forest $\CF \in \hat \CF_\Gamma^-$, write $\K_\CF \Gamma$ for the Feynman 
diagram obtained by performing the contractions of $\hat \CC_\CF \Gamma$.
(So that $\hat \CC_\CF \Gamma$ is a linear combination of terms obtained from
$\K_\CF \Gamma$ by adding node-labels $\Labn$ and the corresponding derivatives on
incident edges.) As above, write $\tau$ for the corresponding bijection between
edges and vertices of $\K_\CF \Gamma$ and those of $\Gamma$. 
Given a Hepp sector $\bT = (T,\bn)$ for $\Gamma$ and an edge $e$ of $\Gamma$, we then write
$\scale_\bT^\CF(e) = \bn(v_e)$, where $v_e = \tau(\tau^{-1}(e)_-) \wedge \tau(\tau^{-1}(e)_+)$
is the common ancestor in $T$ of the two vertices incident to $e$, but when viewed as
an edge of $\K_\CF\Gamma$. (Since we only consider forests consisting of full
subgraphs, $\tau^{-1}(e)_-$ and $\tau^{-1}(e)_+$ are distinct, so this is well-defined.)
Given $\gamma \in \CF$, we then set 
\begin{equ}
\inte_\bT^\CF(\gamma) = \inf_{e \in \CE_\gamma^\CF}\scale_\bT^\CF(e)\;,\qquad
\exte_\bT^\CF(\gamma) = \sup_{e \in \d\CE_\gamma^\CF}\scale_\bT^\CF(e)\;,
\end{equ}
where $\CE_\gamma^\CF$ denotes the edges belonging to $\gamma$,
but \textit{not} to any of the children of $\gamma$ in $\CF$, while $\d\CE_\gamma^\CF$
denotes the edges adjacent to $\gamma$ and belonging to the parent $\CA(\gamma)$ of $\gamma$ in $\CF$ 
(with the convention that if $\gamma$ has no parent, then $\CA(\gamma) = \Gamma$).
With these notations, we then make the following definition.

\begin{definition}
Fix a Hepp sector $\bT$.
Given a forest $\CF\in \hat \CF_\Gamma^-$, we say that $\gamma \in \CF$ is \textit{safe in $\CF$} if  
 $\exte_\bT^\CF(\gamma) \ge \inte_\bT^\CF(\gamma)$ and
 that it is \textit{unsafe in $\CF$} otherwise.
Given a forest $\CF$ and a subgraph $\gamma \in \CG_\Gamma^-$,
we say that $\gamma$ is \textit{safe / unsafe for $\CF$} if $\CF \cup \{\gamma\} \in \hat \CF_\Gamma^-$
and $\gamma$ is safe / unsafe in $\CF \cup \{\gamma\}$. Finally, we say that a forest $\CF$ is safe
if every $\gamma \in \CF$ is safe in $\CF$.
\end{definition}

The following remark is then crucial.

\begin{lemma}\label{lem:decompSafe}
Let $\CF_s \in \hat \CF_\Gamma^-$ be a safe forest and write $\CF_u$ for the collection
of all $\gamma \in \CG_\Gamma^-$ that are unsafe for $\CF_s$. 
Then, one has $\CF_s \cup \CF_u \in \hat \CF_\Gamma^-$ and furthermore every $\gamma$ in $\CF_s$ / $\CF_u$ is
safe / unsafe in $\CF_s \cup \CF_u$.
\end{lemma}

\begin{proof}
Fix $\CF_s$ and write again $\tau$ for the corresponding bijection between
edges and vertices of $\K_{\CF_s} \Gamma$ and those of $\Gamma$. For each $\gamma \in \CF_s$, 
write $\CV_\gamma^{\CF_s} \subset \CV$ for the set of vertices of the form $\tau(\tau^{-1}(e)_\pm)$
for $e \in \CE_\gamma^{\CF_s}$, as well as $v_{\star,\gamma}^{\CF_s} \in \CV_\gamma^{\CF_s}$ for the 
highest one of these vertices. (This is the vertex that edges outside of $\gamma$ were reconnected to  by
the operation $\K_{\CF_s}$.)
We also write $\d\CV_\gamma^{\CF_s} \subset \CV$ for all vertices 
of the form $\tau(\tau^{-1}(e)_\pm)$
for $e \in \d\CE_\gamma^{\CF_s}$ that are \textit{not} in $\CV_\gamma^{\CF_s}$.

With this notation, $\inte_\bT^\CF(\gamma) = \bn((\CV_\gamma^{\CF_s})^\uparrow)$ and there exists a vertex $w \in \d\CV_\gamma^{\CF_s}$
for $\CA(\gamma)$ the parent of $\gamma$ in $\CF_s$ (with the convention as above) such that
$\exte_\bT^\CF(\gamma) = \bn(v_{\star,\gamma}^{\CF_s} \wedge w)$.
Since both $(\CV_\gamma^{\CF_s})^\uparrow$ and $v_{\star,\gamma}^{\CF_s} \wedge w$ lie on the path connecting
the root of $T$ to $v_{\star,\gamma}$, it follows from the definition of a safe forest that
one necessarily has $v_{\star,\gamma}^{\CF_s} \wedge w > (\CV_\gamma^{\CF_s})^\uparrow$.

Let now $\bar\gamma \in \CG_\Gamma^- \setminus \CF_s$ be such that $\CF_s \cup \{\bar \gamma\} \in \hat \CF_\Gamma^-$
and set $\CV_{\bar\gamma} = \CV_{\bar\gamma}^{\CF_s \cup \{\bar \gamma\}}$
as well as $\d\CV_{\bar\gamma} = \d\CV_{\bar\gamma}^{\CF_s \cup \{\bar \gamma\}}$.
It follows from the definitions that $\bar \gamma \in \CF_u$ if and only if 
none of the descendants of $\CV_{\bar\gamma}^\uparrow$ in $T$ belongs to $\d\CV_{\bar\gamma}$.
As a consequence of this characterisation, 
any two graphs $\gamma_1, \gamma_2 \in \CF_u$ are either vertex-disjoint,
or one of them is included in the other one. Indeed, assume by contradiction that neither is included in the
other one and that their intersection $\gamma_\cap$ contains at least one vertex. 
Writing $\hat \gamma_\cap$ for one of the connected components of $\gamma_\cap$,
there exist edges $e_i$ in $\gamma_i$ that are adjacent to $\hat \gamma_\cap$: otherwise, since the $\gamma_i$
are connected, one of them would be contained in $\hat \gamma_\cap$. Write $v_i$ for the vertex of $e_i$ that does
not belong to $\hat \gamma_\cap$. Such a vertex exists since otherwise it would not be the case
that $\hat \gamma_\cap$ is full in $\gamma^\uparrow = \CA(\gamma_1) = \CA(\gamma_2)$.
Since $\gamma_1$ is unsafe, it follows that $v_2$ is not a descendent 
of $(\CV_{\hat \gamma_\cap} \cup \{v_1\})^\uparrow$,
so that in particular, for every vertex $v \in \hat \gamma_\cap$, one has
$v_1 \wedge v > v_2\wedge v$. The same argument with the roles of $\gamma_1$ and $\gamma_2$
reversed then leads to a contradiction.

This shows that $\CF_s \cup \CF_u$ is indeed again a forest so that
it remains to show the last statement.
We will show a slightly stronger statement namely that, given an arbitrary forest 
$\CF$, the property of $\gamma \in \CF$
being safe or unsafe does not change under the operation of adding to $\CF$ a graph 
$\bar \gamma$ that is unsafe for $\CF$.
Given the definitions, there are three potential cases that could
affect the ``safety'' of $\gamma$: either $\bar\gamma \subset  \gamma$, or $\gamma \subset \bar\gamma$,
or $\bar \gamma \subset \CA(\gamma)$ and there exists an edge $e$ adjacent to both $\gamma$ and $\bar \gamma$.
We consider these three cases separately and we write $\bar \CF = \CF \cup \{\bar \gamma\}$.

In the case $\bar\gamma \subset  \gamma$, it follows from the ultrametric property
and the fact that $\bar \gamma$ is unsafe that $\inte_\bT^{\bar \CF}(\gamma) = \inte_\bT^\CF(\gamma)$
whence the desired property follows.
In the case $\gamma \subset \bar\gamma$, it is $\exte_\bT^{\CF}(\gamma)$ which could 
potentially change since $\d\CE_\gamma^\CF$ becomes smaller when adding $\bar \gamma$.
Note however that by the ultrametric property, combined with the fact that $\bar \gamma$ is
unsafe, the edges $e$ in $\d\CE_\gamma^\CF \setminus \d\CE_\gamma^{\bar\CF}$
satisfy $\scale_\bT^\CF(e) = \scale_\bT^{\bar \CF}(e)$. Furthermore, again as a consequence of
$\bar \gamma$ being unsafe, one has $\scale_\bT^{\bar \CF}(e) < \scale_\bT^{\bar \CF}(\bar e)$
for every edge $\bar e$ in $\bar \gamma$ which is not in $\gamma$, so in particular
for $\bar e \in \d\CE_\gamma^{\bar\CF}$. This shows again that $\exte_\bT^{\CF}(\gamma) = \exte_\bT^{\bar \CF}(\gamma)$
as required. The last case can be dealt with in a very similar way, thus concluding the proof.
\end{proof}

As a corollary of the proof, we see that the definition of the notion of ``safe forest''
as well as the construction of $\CF_u$ given a safe forest $\CF_s$ only depend on 
the topology of the tree $T$ and not on the 
specific scale assignment $\bn$. 
It also follows that, given an arbitrary $\CF \in \hat\CF_\Gamma^-$, there exists a unique way
of writing $\CF = \CF_s \cup \CF_u$ with $\CF_s$ a safe forest and $\CF_u$ being unsafe for $\CF_s$
(and equivalently for $\CF$). In particular, writing $\CF_\Gamma^{(s)}(T)$ for the collection of 
safe forests for the tree $T$, 
the collection $\CP_\T = \{[\CF_s, \CF_s \cup \CF_u]\,:\, \CF_s \in \CF_\Gamma^{(s)}(T)\}$
where, for any $\CF_s$, the forest $\CF_u$ is defined as in Lemma~\ref{lem:decompSafe}, forms a partition of 
$\hat\CF_\Gamma^-$ into forest intervals.
It then follows from \eqref{e:terms} that
\begin{equ}
\bigl|\bigl(\Pi^K_\BPHZ \Gamma\bigr)(\phi)\bigr| \le
\sum_{T} \sum_{\CF_s \in \CF_\Gamma^{(s)}(T)} 
\sum_{\bn}
\sup_{x \in D_\bT} \bigl|\bigl(\CW^K\hat \CR_{[\CF_s,\CF_s\cup \CF_u]} \Gamma\bigr)(x)\bigr|\prod_{v \in T} 2^{-d\bn_v}\;,
\end{equ}
where $\bn$ runs over all monotone integer labels for $T$ and the construction of
$\CF_u$ given $\CF_s$ and $T$ is as above. We note that the first two sums are finite, so
that as in the proof of Proposition~\ref{prop:WeinFancy} it is sufficient, 
for any given choice of $T$ and safe forest $\CF_s$, to find a  collection real-valued 
function $\{\eta_i\}_{i \in I}$ (for some \textit{finite} index set $I$) on the interior vertices of $T$ such that
\begin{equ}[e:wantedboundeta]
\sum_{\bn}
\sup_{x \in D_\bT} \bigl|\bigl(\CW^K\hat \CR_{[\CF_s,\CF_s\cup \CF_u]} \Gamma\bigr)(x)\bigr|\prod_{v \in T} 2^{-d\bn_v}
\le \sum_{i \in I} \sum_{\bn}\prod_{v \in T} 2^{-\eta_i(v)\bn_v}\;,
\end{equ}
and such that
\begin{equ}[e:wantedeta]
\sum_{w \ge v}\eta_i(w) > 0 \;,\qquad \forall v \in T\;,\quad \forall i \in I\;,
\end{equ}
which then 
guarantees that the above expression converges.

\begin{wrapfigure}{R}{5cm}
\begin{center}
\vspace{-.5em}
\begin{tikzpicture}[style=thick,scale=0.85]
	\draw[fill=blue!5]  plot[smooth cycle, tension=.7] coordinates {(-1.5,0) (-1.3,0.5) (-0.5,0) (2,0.5) (1.1,-1) (0,-0.2) (-1,-.8)};
	\draw[fill=blue!15]  plot[smooth cycle, tension=.7] coordinates {(2,.5) (2,1.3) (3,1)};
	\draw[fill=blue!15]  plot[smooth cycle, tension=.7] coordinates {(-.5,0) (-1,1.3) (0,1) (0.3,1.2) (0.3,0.5)};
	\draw[fill=blue!30]  plot[smooth cycle, tension=.7] coordinates {(0.3,1.2) (-0.1,1.4) (0.1,1.6) (0.4,1.3)};
\node[dot,red] (l) at (2,.5) {};
\node[dot,red] (l) at (-.5,0) {};
\node[dot,red] (l) at (0.3,1.2) {};
\end{tikzpicture}
\end{center}
\vspace{-1em}
\caption{Structure of $\K_\CF\Gamma$.}\label{fig:collapsedGraph}
\vspace{0.2em}
\end{wrapfigure}
Before we turn to the construction of the $\eta_i$, let us examine in a bit more detail the 
structure of the graph $\K_\CF \Gamma = (\CV_\CF,\CE_\CF)$. 
Writing $\tau$ for the bijection between $\K_\CF\Gamma$ and $\Gamma$,
every $\gamma \in \CF$ yields a subgraph $\K(\gamma) = (\CV_\gamma, \CE_\gamma)$ of $\K_\CF \Gamma$
whose edge set is given by the preimage under $\tau$ of the edge set of 
$\gamma \setminus \bigcup \CC(\gamma)$, where $\CC(\gamma)$ denotes the 
set of all children of $\gamma$ in $\CF$. Furthermore, $\K(\gamma)$ is connected by exactly one vertex 
to $\K(\bar \gamma)$, for $\bar \gamma \in \CC(\gamma) \cup \{\CA(\gamma)\}$, and it is 
disconnected from $\K(\bar \gamma)$ for all other elements $\gamma \in \CF$. 
This is also the case if $\gamma$ is a root of $\CF$, so that $\CA(\gamma) = \Gamma$ by our usual convention, 
if we set $\K(\Gamma)$ to be
the preimage in $\K_\CF\Gamma$ of the complement of all roots of $\CF$.
We henceforth write $v_\star(\gamma)$ for the unique vertex connecting $\K(\gamma)$ to $\K(\CA(\gamma))$
and we write $\CV_\gamma^\star = \CV_\gamma \setminus \{v_\star(\gamma)\}$,
so that one has a partition $\CV_\CF = \CV_\Gamma \sqcup \bigsqcup_{\gamma \in \CF}\CV_\gamma^\star$.

In this way, the 
tree structure of $\CF$ is reflected in the topology of $\K_\CF \Gamma$, as illustrated in 
Figure~\ref{fig:collapsedGraph}, where each $\K(\gamma)$ is stylised by a coloured shape, with
parents having lighter shades than their children and connecting vertices drawn in red.
Recall that we also fixed a total order on the vertices of $\Gamma$ (and therefore those of $\K_{\CF}\Gamma$)
and that the construction of $\K_{\CF}\Gamma$ implies that the corresponding order on
$\{v_\star(\gamma)\}_{\gamma \in \CF}$ is compatible with the partial order on $\CF$ given by inclusion.
For $e \in \CE_\CF$, write $M_e \subset \{+,-\}$ for those ends such that 
$\tau(e)_\bullet \neq \tau(e_\bullet)$ for $\bullet \in M_e$ and set
\begin{equ}
\CE_\CF^{\bullet} = \{(e,\bullet)\,:\, e\in \CE_\CF^m,\, \bullet \in M_e\}\;.
\end{equ}
Then, by the construction of $\K_\CF\Gamma$, for every $\bullet \in M_e$ there exists a 
unique $\gamma_\bullet(e) \in \CF$ and vertex $e_\circ \in \CV_\CF$ such that 
\begin{equ}[e:propse]
e_\bullet = v_\star(\gamma_\bullet(e))\;,\quad e_\circ = \tau^{-1}(\tau(e)_\bullet) \in \CV_{\gamma_\bullet(e)}\;, 
\quad e \in \CE_{\CA(\gamma_\bullet(e))}\;.
\end{equ}
Given $\ell\colon \CE_\CF^\bullet \to \N^d$, we then define a canonical basis element 
$\CD_\CF^\ell\Gamma \in \hat \CT_\Gamma$ by
\begin{equ}
\CD_\CF^\ell\Gamma = (\K_\CF \Gamma, \Labhom^{(\ell)}, \pi\ell)\;,
\end{equ}
where $\Labhom^{(\ell)}$ is the edge-labelling given by $\Labhom^{(\ell)}(e) = \Labhom(\tau(e)) + \sum_{\bullet \in M_e}\ell(e,\bullet)$,
with $\Labhom$ the original edge-labelling of $\Gamma$, and 
$\pi\ell$ is the node-labelling given by
$\pi\ell(v) = \sum\{\ell(e,\bullet)\,:\, e_\circ = v\}$.
Given $\gamma \in \CF$ and $\ell$ as above, we also set 
$\ell(\gamma) = \sum\{|\ell(e,\bullet)|\,:\, \gamma_\bullet(e) = \gamma\}$.

We now return to the bound \eqref{e:wantedboundeta} and first consider the 
special case when $\CF_s$ is a safe forest such that $\CF_u = \emptyset$.
By \eqref{e:contraction} and \eqref{e:defRM}, $\hat \CR_{\CF_s} \Gamma$ can then be written as
\begin{equ}[e:hatRF]
\hat \CR_{\CF_s} \Gamma = (-1)^{|\CF_s|}\sum_{\ell\colon \CE_{\CF_s}^\bullet \to \N^d} {(-1)^{\ell_{\out}}\over \ell!}  \CD_{\CF_s}^\ell\Gamma\;,
\end{equ}
where $\ell_{\out} = \sum\{|\ell(e,\bullet)|\,:\, \bullet = -\}$ and
the sum in \eqref{e:hatRF} is restricted to those choices of $\ell$ such that, for every $\gamma \in \CF_s$,
one has $\deg\gamma + \ell(\gamma) \le 0$.

In this case, we take as the index set $I$ appearing in \eqref{e:wantedeta} 
all those functions $\ell$ appearing in the sum \eqref{e:hatRF} (recall that the sum is restricted to
finitely many such functions) and we set
\begin{equ}
\eta_{\ell} (u) = d + \sum_{e \in \CE_{\CF_s}} \Labhom^{(\ell)}(e)\one_{e^\uparrow}(u)
+ \sum_{(e,\bullet)\in \CE_{\CF_s}^\bullet} |\ell(e)| \one_{(e,\bullet)^\uparrow} (u)\;,
\end{equ}
where, for $e \in \CE_{\CF_s}$, $e^\uparrow$ denotes the node of $T$ given by 
$\tau(e_-)\wedge \tau(e_+)$ and, for $(e,\bullet)\in \CE_{\CF_s}^\bullet$, $(e,\bullet)^\uparrow$ denotes the 
node $\tau(e_\circ) \wedge \tau(e_\bullet)$.

It follows from the definition of $\CW^K$ that this choice does indeed satisfy \eqref{e:wantedboundeta}.
We now claim that as a consequence of the fact that $\CF_s$ 
is such that $\CF_u = \emptyset$, it also satisfies \eqref{e:wantedeta}.
Assume by contradiction that there exists a node $u$ of $T$ and a labelling $\ell$ such that 
$a \eqdef \sum_{v \ge u} \eta_\ell(v) \le 0$.
Write $\CV_0 \subset \CV_{\CF_s}$ for the vertices $v$ such that $\tau(v) \ge u$ in $T$
and $\Gamma_0 = (\CE_0,\CV_0) \subset \K_{\CF_s}\Gamma$ for the corresponding subgraph. 
In general, $\Gamma_0$ does not need to be connected, so we write $\Gamma_0^{(i)} = (\CE_0^{(i)},\CV_0^{(i)})$ for its
connected components.
We then set
\begin{equ}
a_i \eqdef |\CV_{0}^{(i)}| - 1 + \sum_{e \in \CE_0^{(i)}} \Labhom^{(\ell)}(e)
+ \sum_{(e,\bullet)\in \CE_{\CF_s}^\bullet} |\ell(e)| \one_{\{e_\bullet, e_\circ\} \subset \CV_0^{(i)}}\;,
\end{equ}
so that $\sum_i a_i \le a$, with equality if $\Gamma_0$ happens to be connected.
Since $a\le 0$, there exists $i$ such that $a_i \le 0$. Furthermore, $i$ can be chosen
such that $|\CV_{0}^{(i)}| \ge 2$, since $|\CV_0| \ge 2$ and we would otherwise have 
$a = |\CV_0| - 1 \ge 1$.

Set
$\CV_{0,\gamma} = \CV_0\cap \CV_\gamma$ and let $\CF_s^{(i)} \subset \CF_s\cup \{\Gamma\}$ be the
subtree consisting of those $\gamma$ such that either $\CE_\gamma \cap \CE_0^{(i)} \neq \emptyset$
or $v_\star(\gamma) \in \CV_0^{(i)}$ (or both). 
We also break $a_i$
into contributions coming from each $\gamma \in \CF_s^{(i)}$ by setting
\begin{equ}[e:defagamma]
a_{i,\gamma} \eqdef |\CV_{0,\gamma}| - 1 + \sum_{e \in \CE_{\gamma} \cap \CE_0} \Labhom^{(\ell)}(e)
+  \sum_{(e,\bullet)\in \CE_{\CF_s}^\bullet} \one_{\gamma_\bullet(e) = \gamma} |\ell(e)| \one_{\{e_\bullet, e_\circ\} \subset \CV_0^{(i)}}\;.
\end{equ}
We claim that $\sum_\gamma a_{i,\gamma} = a_i$: recalling that one always has $\Gamma \in \CF_s^{(i)}$ by definition,
the only part which is not immediate is
that $\sum_\gamma(|\CV_{0,\gamma}| - 1) = |\CV_{0}^{(i)}| - 1$. This is
a consequence of the fact that in the sum $\sum_\gamma |\CV_{0,\gamma}|$, each ``connecting vertex''
is counted double. Since $\CF_s^{(i)}$ is a tree, the number of these  equals
$|\CF_s^{(i)}|-1$, whence the claim follows.

We introduce the following terminology. An element $\gamma \in \CF_s \cup \{\Gamma\}$ is said to be ``full''
if $\CE_\gamma \cap \CE_0^{(i)} = \CE_\gamma$, ``empty'' if  
$\CE_\gamma \cap \CE_0^{(i)} = \emptyset$,
and ``normal'' otherwise. We also set $a_{i,\gamma} = 0$ for all empty $\gamma$ with 
$\CV_{0,\gamma} = \emptyset$.
Recall furthermore the definition of $\deg \gamma$ for $\gamma \in \CF_s$ given in \eqref{e:degreeSubgraph}
and the definition of $\ell(\gamma)$ given above. With this terminology, we then have the following.

\begin{lemma}\label{lem:propsgraph}
A full subgraph $\gamma$ cannot have an empty parent 
and one has
\begin{equs}[2]
	a_{i,\gamma} &= \deg \gamma + \ell(\gamma) - \sum_{\bar \gamma \in \CC(\gamma)} (\deg \bar \gamma + \ell(\bar \gamma)) &\quad& \text{if $\gamma$ is full,} \\
	a_{i,\gamma} &= 0 && \text{if $\gamma$ is empty,} \label{e:wantedagamma}\\
	a_{i,\gamma} &>  - \sum_{\bar \gamma \in \CC_\star(\gamma)} (\deg \bar \gamma + \ell(\bar \gamma))&& \text{if $\gamma$ is normal,}
\end{equs}
where $\CC_\star(\gamma)$ consists of those children $\bar \gamma$ 
of $\gamma$ such that $v_\star(\bar \gamma) \in \CV_0^{(i)}$.
\end{lemma}

Before we proceed to prove
Lemma~\ref{lem:propsgraph}, let us see how this leads to a contradiction. 
By \eqref{e:hatRF}, one has $\deg \bar \gamma + \ell(\bar \gamma) \le 0$ for every $\gamma \in \CF_s$
and a fortiori $\deg \bar \gamma < 0$. Furthermore, since $|\CV_0^{(i)}| \ge 2$, there exists
at least one subgraph $\gamma$ which is either full or normal. Since full subgraphs can only have parents
that are either full or normal and since $\Gamma$ itself cannot be full (since legs are never contained
in $\CE_0^{(i)}$), we have at least one normal subgraph.
Since each of the negative terms $\deg \gamma + \ell(\gamma)$ appearing in the right hand side of 
the bound of $a_{i,\gamma}$ for $\gamma$ full is compensated by a corresponding term in its parent,
and since we use the strict inequality appearing for normal $\gamma$ at least once, we conclude that
one has indeed $\sum_\gamma a_{i,\gamma} > 0$ as required.

\begin{proof}[of Lemma~\ref{lem:propsgraph}]
Let us first show that the bounds \eqref{e:wantedagamma} hold. If $\gamma$ is empty, one has
either $\gamma \not \in \CF_s^{(i)}$ in which case $\CV_{0,\gamma} = \emptyset$
and $a_{i,\gamma} = 0$ by definition, or $\CV_{0,\gamma} = v_{\star}(\gamma)$ in which case
$a_{i,\gamma} = 0$ by \eqref{e:defagamma}. If $\gamma$ is full, then it follows immediately
from the definition of $\deg\gamma$ that one would have
$a_{i,\gamma} = \deg \gamma  - \sum_{\bar \gamma \in \CC(\gamma)} \deg \bar \gamma$ if it weren't
for the presence of the labels $\ell$. If $\gamma$ is full then, whenever 
$(e,\bullet)$ is such that $\gamma_\bullet(e) = \gamma$, one also 
has $\{e_\bullet,e_\circ\} \subset \CV_0^{(i)}$ by \eqref{e:propse} and the definition of being full.
Similarly, one has $e \in \CE_\gamma \cap \CE_0$ whenever $\gamma_\bullet(e) \in \CC(\gamma)$.
The first identity in \eqref{e:wantedagamma} then follows from the 
fact that each edge with $\gamma_\bullet(e) = \gamma$ contributes $|\ell(e)|$ to the last term in 
\eqref{e:defagamma} while each edge with $\gamma_\bullet(e) \in \CC(\gamma)$ contributes $-|\ell(e)|$
to the penultimate term.

Regarding the last identity in \eqref{e:wantedagamma}, given a normal subgraph $\gamma$, write
$\hat \gamma$ for the subgraph of $\Gamma$ with edge set given by 
\begin{equ}[e:defgammahat]
\hat \CE = \tau(\CE_\gamma \cap \CE_0^{(i)}) \cup \bigcup_{\bar \gamma \in \CC_\star(\gamma)} \CE(\bar \gamma) \;. 
\end{equ}
In exactly the same way as for a full subgraph, one then has
$a_{i,\gamma} \ge \deg \hat \gamma - \sum_{\bar \gamma \in \CC_\star(\gamma)} (\deg \bar \gamma + \ell(\bar \gamma))$.
The reason why this is an inequality and not an equality is that we may have additional positive
contributions coming from those $\ell(e)$ with $\gamma_m(e) = \gamma$ and such that $e_\circ \in \CV_0^{(i)}$,
while we do not have any negative contributions from those 
 $\ell(e)$ with $\gamma_m(e) \in \CC_\star(\gamma)$ but $e \not \in \CE_0^{(i)}$.
The claim then follows from the fact that one necessarily has $\deg\hat \gamma > 0$ by the assumption
that $\CF_u = \emptyset$. Indeed, it follows from its definition and the construction of the 
Hepp sector $T$ that the subgraph $\Gamma_0$ satisfies that $\scale_\bT^{\CF_s}(e) > \scale_\bT^{\CF_s}(e)$ for
every edge $e \in \CE_0$ and every edge $\bar e$ adjacent to $\Gamma_0$ in $\K_{\CF_s}\Gamma$, so that one
would have $\hat \gamma \in \CF_u$ otherwise.

It remains to show that if $\gamma$ is a full subgraph, then it cannot have empty parents. This follows in
essentially the same way as above, noting that if it were the case that $\gamma$ has an empty parent, then
it would be unsafe in $\CF_s$, in direct contradiction with the fact that $\CF_s$ is a safe forest.
\end{proof}

In order to complete the proof of Theorem~\ref{theo:BPHZ}, it remains to consider the general case when
$\CF_u \neq \emptyset$. In this case, setting $\M = [\CF_s, \CF_s\cup \CF_u]$, we have
\begin{equ}[e:defRMgen]
\hat \CR_\M \Gamma
= (-1)^{|\CF_s|}\sum_{\ell\colon \CE_{\CF_s}^m \to \N^d} {(-1)^{\ell_{\out}}\over \ell!} \Big(\prod_{\gamma \in \CF_u}(\id - \hat \CC_{\gamma})\Big) \CD_{\CF_s}^\ell\Gamma\;,
\end{equ}
with the sum over $\ell$ restricted in the same ways as before. Again, we 
bound each term in this sum separately, so that our index set $I$ consists again of the subset of 
functions $\ell\colon \CE_{\CF_s}^m \to \N^d$ such that $\deg\gamma + \ell(\gamma) < 0$ for every
$\gamma \in \CF_s$, but this time each of these summands is still comprised of several terms
generated by the action of the operators $\hat \CC_{\gamma}$ for the ``unsafe'' graphs $\gamma$.

For any $\gamma \in \CF_u$, we define a subgraph $\K(\gamma)$ of $\K_{\CF_s}\Gamma$ as before,
with the children of $\gamma$ being those in $\CF_s \cup \{\gamma\}$ \textit{not} in all of $\CF_s\cup \CF_u$.
The definition of $\gamma$ being ``unsafe'' then guarantees that 
there exists a vertex $\gamma^\uparrow$ in $T$ such that $\tau(\CV_\gamma) = \{v \in \CV\,:\, v \ge \gamma^\uparrow\}$.
We furthermore define
\begin{equ}
\gamma^{\uuparrow} = \sup \{e^\uparrow\,:\, e \in \CE_{\CA(\gamma)}\;\&\; e \sim \K(\gamma)\}\;,
\end{equ}
with ``$\sim$'' meaning ``adjacent to'',
which is well-defined since all of the elements appearing under the $\sup$ lie on the path joining
$\gamma^\uparrow$ to the root of $T$. In particular, one has $\gamma^\uparrow > \gamma^{\uparrow\uparrow}$.
We also set $N(\gamma) = 1 + \lfloor-\deg\gamma\rfloor$ with the 
convention  that $N(\gamma) = 0$ for $\gamma \not \in \CF_u$.

We claim that this time, if we set 
\begin{equ}[e:fulleta]
\eta_{\ell} (u) = d + \sum_{e \in \CE_{\CF_s}} \Labhom^{(\ell)}(e)\one_{e^\uparrow}(u)
+ \sum_{e\in \CE_{\CF_s}^m} |\ell(e)| \one_{e_m^\uparrow} (u) +
\sum_{\gamma \in \CF_u} N(\gamma) \bigl(\one_{\gamma^\uparrow}(u) - \one_{\gamma^\uuparrow}(u)\bigr) \;,
\end{equ}
then $\eta_\ell$ does indeed satisfy the required properties, which then concludes the proof.
As before, we assume by contradiction that there is $u$ such that 
$a = \sum_{v \ge u}\eta_\ell(v) \le 0$ and we define, for each connected component $\Gamma_0^{(i)}$
of $\Gamma_0$,
\begin{equ}
a_i \eqdef |\CV_{0}^{(i)}| - 1 + \sum_{e \in \CE_0^{(i)}} \Labhom^{(\ell)}(e)
+ \sum_{e\in \CE_{\CF_s}^m} |\ell(e)| \one_{\{e_\bullet, e_\circ\} \subset \CV_0^{(i)}}
+ \sum_{\gamma \in \CF_u} N(\gamma)\one_{\K(\gamma) = \Gamma_0^{(i)} \cap \K(\CA(\gamma))}\;.
\end{equ}
It is less obvious than before to see that $\sum a_i \le a$ because of the 
presence of the last term. Given $\gamma \in \CF_u$, there are two possibilities regarding the 
corresponding term in \eqref{e:fulleta}.
If $\gamma^\uparrow < u$ in $T$, then it
does not contribute to $a$ at all. Otherwise, $\tau^{-1}(\gamma)$ is included
in $\Gamma_0$ and 
we distinguish two cases. In the first case, one has $\K(\gamma) = \K(\CA(\gamma)) \cap \Gamma_0$.
In this case, since the inclusion $\gamma \subset \CA(\gamma)$ is strict,
there is at least one edge in $\K(\CA(\gamma))$ adjacent to $\K(\gamma)$. 
Since this edge is also adjacent to $\Gamma_0$, it follows that in this case $\gamma^\uuparrow < u$
so that we have indeed a contribution $N(\gamma)$ to $a$.
In the remaining case, the corresponding term may or may not
contribute to $a$, but if it does, then its contribution is necessarily positive, so we can
discard it and still have $\sum a_i \le a$ as required.

As before, we then write $a_i = \sum_{\gamma \in \CF_s^{(i)}} a_{\gamma,i}$ with
\begin{equs}
a_{i,\gamma} &\eqdef |\CV_{0,\gamma}| - 1 + \sum_{e \in \CE_{\gamma} \cap \CE_0} \Labhom^{(\ell)}(e)
+  \sum_{e\in \CE_{\CF_s}^m} \one_{\gamma_m(e) = \gamma} |\ell(e)| \one_{\{e_\bullet, e_\circ\} \subset \CV_0^{(i)}}\\
&\qquad + \sum_{\bar\gamma \in \CF_u} N(\bar\gamma)\one_{\K(\bar\gamma) = \Gamma_0^{(i)} \cap \K(\gamma)}\;.\label{e:defagammafull}
\end{equs}
We claim that the statement of Lemma~\ref{lem:propsgraph} still holds in this case. 
Indeed, the only case that requires a slightly different argument is that when $\gamma$ is ``normal''.
In this case, defining again $\hat \gamma$ as in \eqref{e:defgammahat}, we have
\begin{equ}
a_{i,\gamma} \ge \deg \hat \gamma + N(\hat\gamma) - 
\sum_{\bar \gamma \in \CC_\star(\gamma)} (\deg \bar \gamma + \ell(\bar \gamma))\;,
\end{equ}
since the last term in \eqref{e:defagammafull} contributes precisely when $\hat \gamma \in \CF_u$
and then only the term with $\bar \gamma = \hat \gamma$ is selected by the indicator function.
The remainder of the argument, including the fact that this then yields a contradiction with the
assumption that $a \le 0$, is then identical to before since one always has $\deg \hat \gamma + N(\hat\gamma) > 0$.

In order to complete the proof of our main theorem, it thus remains to show that the choice
of $\eta_\ell$ given in \eqref{e:fulleta} allows to bound from above the contribution of the Hepp
sector indexed by $T$, in the sense that the bound \eqref{e:wantedboundeta} holds.
The only non-trivial part of this is the presence of a term
\begin{equ}
N(\gamma) \bigl(\one_{\gamma^\uparrow}(u) - \one_{\gamma^\uuparrow}(u)\bigr)
\end{equ} 
for each factor of $(1-\hat \CC_\gamma)$ in \eqref{e:defRMgen}. This will be a consequence
of the following bound.

\begin{lemma}\label{lem:Taylor}
Let $K_i \colon \S\to\R$ be kernels satisfying the bound \eqref{e:propKer}
with $\deg\Labhom = -\alpha_i < 0$ for $i\in I$ with $I$ a finite index set, and write $I_\star = I \sqcup \{\star\}$.
Let furthermore $x_i, y_i \in \S$ such that 
$|x_i - x_j| \le \delta < \Delta  \le |x_i - y_j|$ for all $i,j \in I_\star$ and let $N \ge 0$ be an integer.
Then, one has the bound
\begin{equs}
\Bigl|\prod_{i \in I} K_i(x_i-y_i)
- \sum_{|\ell| < N}{1\over \ell!} \prod_{i \in I} (x_i-x_\star)^{\ell_i}&\bigl(D^{\ell_i}K_i\bigr)(x_\star-y_i)\Bigr|\label{e:Taylor}\\
&\qquad \lesssim \delta^{N} \Delta^{-N}\prod_{i \in I} |y_i-x_i|^{-\alpha_i}\;.
\end{equs}
\end{lemma}

\begin{proof}
The proof is a straightforward application of Taylor's theorem to the function
$x \mapsto \prod_{i \in I} K_i(x_i)$ defined on $\S^I$.
For example, the version given in \cite[Prop.~A.1]{Regularity}
shows that for every $\tilde\ell\colon I \to \N^d$ with $|\tilde \ell| = N$, 
there exist measures $\CQ_{\tilde \ell}$ on $\S^I$ with total variation 
${1\over \tilde \ell!}\prod_i |(x_i - x_\star)^{\tilde \ell_i}| \lesssim \delta^N$
and support in the ball of radius $K\delta$ around $(x_\star,\ldots,x_\star)$
(for some $K$ depending only on $|I|$ and $d$) such that 
\begin{equs}
\prod_{i \in I} K_i(x_i-y_i)
- \sum_{|\ell| < N}&{1\over \ell!} \prod_{i \in I} (x_i-x_\star)^{\ell_i} \bigl(D^{\ell_i}K_i\bigr)(x_\star-y_i)\label{e:remainder}\\
& =\sum_{|\tilde \ell| = N}\int \prod_{i \in I} \bigl(D^{\tilde\ell_i}K_i\bigr)(z_i-y_i)\,\CQ_{\tilde \ell}(dz)
\end{equs}
If $\Delta > (K+1)\delta$, then the claim follows at once from the fact that
\begin{equ}
\bigl|\bigl(D^{\tilde\ell_i}K_i\bigr)(z_i-y_i)\bigr| \lesssim |z_i-y_i|^{-\alpha_i - |\tilde \ell_i|}
\lesssim \Delta^{ - |\tilde \ell_i|} |x_i-y_i|^{-\alpha_i}\;.
\end{equ}
If $\Delta \le (K+1)\delta$ on the other hand, each term in the left hand side of \eqref{e:Taylor}
already satisfies the required bound individually.
\end{proof}

It now remains to note that each occurence of $(1-\hat \CC_\gamma)$ in \eqref{e:defRMgen}
produces precisely one factor of the type considered in Lemma~\ref{lem:Taylor}, with 
the set $I$ consisting of the edges in $\CA(\gamma)$ adjacent to $\gamma$, 
$\delta = 2^{-\bn(\gamma^\uparrow)}$ and $\Delta = 2^{-\bn(\gamma^\uuparrow)}$.
The additional factor $\delta^{N(\gamma)} \Delta^{-N(\Gamma)}$ produced in this way
precisely corresponds to the additional term $N(\gamma) \bigl(\one_{\gamma^\uparrow}(u) 
- \one_{\gamma^\uuparrow}(u)\bigr)$ in our definition of $\eta$.
The only potential problem that could arise is when some edges are involved in 
the renormalisation of more than one different subgraph.
The explicit formula \eqref{e:remainder} however shows that this is not a problem.
The proof of Theorem~\ref{theo:BPHZ} is complete.

\subsection{Properties of the BPHZ valuation}
\label{sec:props}

In this section, we collect a few properties of the BPHZ valuation $\Pi^K_\BPHZ$. 
In order to formulate the main tool for this, we first introduce a ``gluing operator''
$\Contr \colon \hat\CT_- \to \hat\CT_-$ such that $\Contr \Gamma$ is the connected vacuum diagram
obtained by identifying all the marked vertices of $\Gamma$, for example
\begin{equ}
\Contr \Bigl(\<markedtriang>\<markedloop>\Bigr) = \<markedboth>\;,
\end{equ}
where the marked vertices are indicated in green. It follows from the definition 
\eqref{e:PiK} that the linear map $\Pi_-^K$ satisfies the identity
\begin{equ}[e:contr]
\Pi_-^K \Contr\tau = \Pi_-^K\tau\;,\qquad \tau \in \hat\CT_-\;.
\end{equ}
We claim that the same also holds for $\Pi_-^K \hat \CA \pi$, where $\pi \colon \hat \CT_- \to \CH_-$ is the
canonical projection.

\begin{lemma}\label{lem:prodglue}
One has $\Pi_-^K \hat \CA \pi\Contr\tau = \Pi_-^K\hat \CA\pi\tau$ for all $\tau \in \hat \CT_-$.
\end{lemma}

\begin{proof}
By induction on the number of connected components and since $\Pi_-^K$, $\hat \CA$ and $\pi$ are
all multiplicative, it suffices to show that, for
every element $\tau$ of the form $\tau = \gamma_1 \gamma_2$ where the $\gamma_i$ are connected
and non-empty, one has the identity
\begin{equ}
\Pi_-^K \hat \CA \pi \Contr \tau = \Pi_-^K \hat \CA \pi \gamma_1\cdot \Pi_-^K \hat \CA \pi \gamma_2
= \Pi_-^K \bigl(\hat \CA \pi \gamma_1\cdot \hat \CA \pi \gamma_2\bigr)\;.
\end{equ}
In particular, one has $\Pi_-^K \hat \CA \tau = 0$ for every $\tau$ with $\deg\tau \le 0$ of the form
$\Contr(\gamma_1 \gamma_2)$, as soon as one of the factors has strictly positive degree.

We will use the fact that, as a consequence of \eqref{e:twisted} combined with the definition of
$\Deltam$, one has for connected $\sigma = (\Gamma,v_\star,\kk)$ with $\deg\sigma \le 0$ the identity
\begin{equ}[e:idenAhat]
\hat \CA \sigma = - \sigma - 
 \sum_{\bar \Gamma \subset \Gamma \atop \bar \Gamma \not \in \{\emptyset, \Gamma\}} 
 \CM(\hat \CA\pi \otimes \id) \CX_{\bar \Gamma}\sigma
 \;,
\end{equ}
where we made use of the operators
\begin{equ}
\CX_{\bar \Gamma}\sigma = 
\sum_{\bar \ell\colon \d\bar \Gamma \to \N^d \atop \bar \kk\colon \bar \CV \to \N^d}
{(-1)^{|\out \bar \ell|} \over \bar \ell!}\binom{\kk}{\bar \kk} (\bar \Gamma,\star, \bar \kk + \pi\bar \ell) \otimes (\Gamma,v_\star,\kk - \bar\kk) / (\bar \Gamma, \bar \ell)
\end{equ}
and $\star$ denotes some arbitrary choice of distinguished vertex.
(Here, $\CX$ stands for ``extract''.) Note that the nonvanishing terms in \eqref{e:idenAhat}
are always such that
the degree of $\bar\Gamma$ (not counting node-decorations) is negative.

The proof of the lemma now goes by induction on the number of edges of $\tau = \gamma_1\gamma_2$. 
In the base case, each of the $\gamma_i$ has one edge and there are two non-trivial
cases. In the first case, $\deg \gamma_i \le 0$ for both values of $i$. In this case,
it follows from the above formula that, since $v_\star$ is the vertex in $\Contr \tau$
at which  both edges are connected and since $\hat \CA \gamma_1 = - \gamma_i$, one has
\begin{equ}
\hat \CA \Contr \tau = - \Contr\tau + 2\tau\;,
\end{equ}
so that the claim follows from \eqref{e:contr}, combined with the fact that $\hat \CA \pi \gamma_i = -\gamma_i$.
In the second case, one has $\deg \gamma_1 \le 0$ and $\deg \gamma_2 > 0$, but $\deg \gamma_1 + \deg \gamma_2 \le 0$
so that $\pi \Contr\tau = \Contr\tau$. In this case, the only subgraph of $\Contr \tau$ of negative degree is
$\gamma_1$, so that 
\begin{equ}
\hat \CA \Contr \tau = - \Contr\tau + \tau\;,
\end{equ}
thus yielding $\Pi_-^K \hat \CA \Contr \tau = 0$ as required.

We now write $\Gamma$ for the graph associated to $\Contr\tau$ and $\Gamma_i \subset \Gamma$ for the 
subgraphs associated to each of the factors $\gamma_i$. 
Writing $\CU_\Gamma$ for the set of all non-empty proper subgraphs of
$\Gamma$, we then have a natural bijection
\begin{equs}
\CU_\Gamma = \CU_{\Gamma_1}
&
\sqcup \{\bar \gamma_1 \sqcup \Gamma_2\,:\, \bar \gamma_1 \in \CU_{\Gamma_1}\}
\sqcup \CU_{\Gamma_2} \sqcup \{\bar \gamma_2 \sqcup \Gamma_1\,:\, \bar \gamma_2 \in \CU_{\Gamma_2}\}\\
&\sqcup \{\bar \gamma_1 \sqcup \bar \gamma_2\,:\, \bar \gamma_1 \in \CU_{\Gamma_1},\; \bar \gamma_2 \in \CU_{\Gamma_2}\} \sqcup \{\Gamma_1,\Gamma_2\}\;.\label{e:bijection}
\end{equs}
Take now an element of the form $\bar \gamma_1 \sqcup \Gamma_2$ from the first set above.
As before, there are no edges in $\Gamma$ adjacent to $\Gamma_1$ other than those incident to $v_\star$.
Furthermore, $\bar \gamma_1 \sqcup \Gamma_2$ has strictly less edges than $\Gamma$, so we can apply our induction
hypothesis, yielding
\begin{equs}
\Pi_-^K\CM(\hat \CA\pi \otimes \id) \CX_{\bar \gamma_1 \sqcup \Gamma_2}\Contr\tau
&= \Pi_-^K\CM(\hat \CA\pi \CM_{\gamma_2} \otimes \id) \CX_{\bar \gamma_1}\gamma_1 \\
&= \Pi_-^K\bigl(\hat \CA\pi\gamma_2\cdot \CM(\hat \CA\pi \otimes \id) \CX_{\bar \gamma_1}\gamma_1\bigr)\;,
\end{equs}
where $\CM_{\gamma_2}\colon \gamma \mapsto \Contr(\gamma\cdot\gamma_2)$.
In a similar way, we obtain the identities
\begin{equs}
\Pi_-^K\CM(\hat \CA\pi \otimes \id) \CX_{\bar \gamma_1}\Contr\tau
&= \Pi_-^K\bigl(\gamma_2\cdot \CM(\hat \CA\pi \otimes \id) \CX_{\bar \gamma_1}\gamma_1\bigr)\;, \\
\Pi_-^K\CM(\hat \CA\pi \otimes \id) \CX_{\bar \gamma_1\sqcup \bar \gamma_2}\Contr\tau
&= \Pi_-^K\bigl(\CM(\hat \CA\pi \otimes \id) \CX_{\bar \gamma_1}\gamma_1 \cdot \CM(\hat \CA\pi \otimes \id) \CX_{\bar \gamma_2}\gamma_2\bigr)\;,\\
\Pi_-^K\CM(\hat \CA\pi \otimes \id) \CX_{\gamma_1}\Contr\tau
&= \Pi_-^K\bigl(\gamma_2\cdot \hat \CA \pi\gamma_1\bigr)\;,
\end{equs}
as well as the corresponding identities with $1$ and $2$ exchanged.
Inserting these identities into \eqref{e:idenAhat} (with the sum broken up according
to \eqref{e:bijection}), we obtain
\begin{equs}
\Pi_-^K \hat \CA \Contr \tau &= - \Pi_-^K \Contr \tau 
- \sum_{\bar \gamma_1 \in \CU_{\Gamma_1}}
\Pi_-^K\bigl((\gamma_2 +\hat \CA\pi\gamma_2) \cdot \CM(\hat \CA\pi \otimes \id) \CX_{\bar \gamma_1}\gamma_1\bigr) \\
&\qquad - \sum_{\bar \gamma_2 \in \CU_{\Gamma_2}}
\Pi_-^K\bigl((\gamma_1 +\hat \CA\pi\gamma_1) \cdot \CM(\hat \CA\pi \otimes \id) \CX_{\bar \gamma_2}\gamma_2\bigr) \\
&\qquad- \sum_{\bar \gamma_1 \in \CU_{\Gamma_1}}\sum_{\bar \gamma_2 \in \CU_{\Gamma_2}} \Pi_-^K\bigl(\CM(\hat \CA\pi \otimes \id) \CX_{\bar \gamma_1}\gamma_1 \cdot \CM(\hat \CA\pi \otimes \id) \CX_{\bar \gamma_2}\gamma_2\bigr) \\
&\qquad - \Pi_-^K\bigl(\gamma_2\cdot \hat \CA \pi\gamma_1\bigr) - \Pi_-^K\bigl(\gamma_1\cdot \hat \CA \pi\gamma_2\bigr)\;.
\end{equs}
At this stage, we differentiate again between the case in which $\deg\gamma_i \le 0$ for both $i$ and the
case in which one of the two has positive degree. (The case in which both have positive degree is
again trivial.) In the former case, $\pi\gamma_i = \gamma_i$ and one has
\begin{equ}
\gamma_2 +\hat \CA\pi\gamma_2 = -\sum_{\bar \gamma_2 \in \CU_{\Gamma_2}} \CM(\hat \CA\pi \otimes \id) \CX_{\bar \gamma_2}\gamma_2\;.
\end{equ}
In particular, the second and third terms are the same as the fourth, but with opposite sign and one has
\begin{equs}
\Pi_-^K \hat \CA \Contr \tau &= 
 \sum_{\bar \gamma_1 \in \CU_{\Gamma_1}}\sum_{\bar \gamma_2 \in \CU_{\Gamma_2}} \Pi_-^K\bigl(\CM(\hat \CA\pi \otimes \id) \CX_{\bar \gamma_1}\gamma_1 \cdot \CM(\hat \CA\pi \otimes \id) \CX_{\bar \gamma_2}\gamma_2\bigr) \\
&\qquad - \Pi_-^K (\gamma_1\cdot \gamma_2)  - \Pi_-^K\bigl(\gamma_2\cdot \hat \CA\gamma_1\bigr) - \Pi_-^K\bigl(\gamma_1\cdot \hat \CA \gamma_2\bigr)\\
&= \Pi_-^K \bigl((\hat \CA\gamma_1+\gamma_1)\cdot(\hat \CA\gamma_2+\gamma_2)\bigr) \\
&\qquad - \Pi_-^K (\gamma_1\cdot \gamma_2)  - \Pi_-^K\bigl(\gamma_2\cdot \hat \CA\gamma_1\bigr) - \Pi_-^K\bigl(\gamma_1\cdot \hat \CA \gamma_2\bigr) \\
&= \Pi_-^K \bigl(\hat \CA\gamma_1 \cdot \hat \CA\gamma_2\bigr)\;,
\end{equs}
as claimed. Consider now the case $\deg\gamma_1 > 0$. Then, the two terms containing $\hat \CA \pi\gamma_1$
vanish and we obtain similarly
\begin{equs}
\Pi_-^K \hat \CA \Contr \tau &= - \Pi_-^K \Contr \tau 
- \sum_{\bar \gamma_1 \in \CU_{\Gamma_1}}
\Pi_-^K\bigl(\gamma_1 \cdot \CM(\hat \CA\pi \otimes \id) \CX_{\bar \gamma_2}\gamma_2\bigr) 
- \Pi_-^K\bigl(\gamma_1\cdot \hat \CA \gamma_2\bigr)\\
&=- \Pi_-^K (\gamma_1\cdot\gamma_2) + \Pi_-^K \bigl(\gamma_1 \cdot (\hat \CA\gamma_2+\gamma_2)\bigr)
- \Pi_-^K\bigl(\gamma_1\cdot \hat \CA \gamma_2\bigr) = 0\;,
\end{equs}
as claimed, thus concluding the proof.
\end{proof}

As a consequence of this result, we have the following. Recall that $\CS_k^{(c)}$ is the space of 
translation invariant compactly supported (modulo translations) distributions in $k$ variables.
Given $x \in \S^k$, $y \in \S^\ell$, we also write $x\sqcup y = (x_1,\ldots,x_k,y_1,\ldots,y_\ell) \in \S^{k+\ell}$.  
For any $k,\ell \ge 1$, we 
then have a bilinear ``convolution operator'' $\star\colon \CS_k^{(c)} \times \CS_\ell^{(c)} \to \CS_{k+\ell-2}^{(c)}$
obtained by setting
\begin{equ}
\bigl(\eta \star \zeta\bigr)(x \sqcup y) = \int_\S 
\eta(x \sqcup z)\zeta(z \sqcup y)\,dz\;,\quad x \in \S^{k-1}, y \in \S^{\ell-1}\;,
\end{equ}
whenever $\eta$ and $\zeta$ are represented by continuous functions. It is straightforward to see 
that this extends continuously to all of $\CS_k^{(c)} \times \CS_\ell^{(c)}$, and that it coincides with
the usual convolution in the special case $k = \ell = 2$.

Similarly, we have a convolution operator
$\star\colon \CH_k \times \CH_\ell \to \CH_{k+\ell-2}$ obtained in the following way.
Let
$\Gamma \in \CT_k$ and $\bar \Gamma \in \CT_\ell$ be Feynman diagrams such that the label of the 
$k$th leg of $\Gamma$ and the first leg of $\bar \Gamma$ are both given by $\delta$. 
We then define $\Gamma \star \bar \Gamma \in \CT_{k+\ell-2}$ to be the Feynman diagram
with $k+\ell-2$ legs obtained by removing the $k$th leg of $\Gamma$ as well as the first leg
of $\bar \Gamma$, and identifying the two vertices these legs were connected to.
(We also need to relabel the legs of $\bar \Gamma$ accordingly.)
This operation extends to all of $\CH_k \times \CH_\ell$ by noting that given a Feynman diagram
$\Gamma \in \CT_k$, there always exists $\Gamma_n \in \CT_k$ with 
$\Gamma_n =  \Gamma$ in $\CH_k$ which is a linear combination of diagrams with label
$\delta$ on the $n$th leg: if the $n$th leg of $\Gamma$ has label $\delta^{(m)}$ with $m \neq 0$, 
one obtains $\Gamma_n$ by performing $|m|$ ``integrations by parts'' using \eqref{e:IBP}.
We then define in general $\Gamma \star \bar \Gamma$ by setting
$\Gamma \star \bar \Gamma \eqdef \Gamma_k \star \bar \Gamma_0$ and we can check that this is
indeed well-defined in $\CH_{k+\ell-2}$. We then have the following consequence of
Lemma~\ref{lem:prodglue}.

\begin{proposition}\label{prop:convolution}
The BPHZ valuation satisfies $\Pi_\BPHZ(\Gamma \star \bar \Gamma)
= \Pi_\BPHZ\Gamma \star \Pi_\BPHZ\bar \Gamma$.
\end{proposition}

\begin{proof}
Write $\CM^\star \colon \CH \otimes \CH \to \CH$ for the convolution operator
introduced above and note that the canonical valuation $\Pi$ (we suppress the dependence on $K$)
does satisfy the property of the statement. It therefore suffices to show that one has the identity
\begin{equ}[e:wanted]
(\Pi_-\hat \CA \otimes \id)\Delta \CM^\star
=
\CM^\star \bigl((\Pi_-\hat \CA \otimes \id)\Delta \otimes
(\Pi_-\hat \CA \otimes \id)\Delta\bigr)
\end{equ}
between maps $\CH \otimes \CH \to \CH$.

Suppose that $\Gamma\in\CT_k$ and $\bar \Gamma \in \CT_\ell$,
write $v$ for the vertex of $\Gamma$ adjacent to the $k$th leg, and let $\bar v$ be the 
vertex of $\bar \Gamma$ adjacent to its first leg.
Fix furthermore an arbitrary map $\sigma\colon (\CV_\star \sqcup \bar \CV_\star)/\{v,\bar v\} \to \N$
which is injective and such that $\sigma(v) = 0$.
Since internal edges of $\Gamma \star \bar \Gamma$ are in bijection with the disjoint
union of the internal edges of $\Gamma$ and those of $\bar \Gamma$, we
have an obvious bijection between subgraphs $\gamma$ of $\Gamma \star \bar \Gamma$ 
and pairs $(\gamma_1,\gamma_2)$ of subgraphs of $\Gamma$ and $\bar \Gamma$.
We also have a natural choice of distinguished vertex for each connected subgraph of 
$\Gamma$, $\bar \Gamma$ or $\Gamma \star \bar \Gamma$ by choosing the vertex with the lowest
value of $\sigma$. 
If we then write $\hat \Delta \tau \in \hat \CT_- \otimes \CH$ for the right hand side of \eqref{e:coaction}
with this choice of distinguished vertices, then we see that 
\begin{equ}
(\Contr \otimes \id)\hat \Delta (\Gamma \star \bar \Gamma) = (\Contr\CM \otimes \CM^\star)(\id \otimes \tau \otimes\id) (\hat\Delta \Gamma \otimes \hat \Delta \bar \Gamma)\;,
\end{equ}
where $\tau \colon \hat \CT_- \otimes \CH \to \CH \otimes \hat \CT_-$ is the map that exchanges 
the two factors. Applying $\Pi_-\hat \CA \pi$ to both sides and making use of 
Lemma~\ref{lem:prodglue}, the required identity \eqref{e:wanted} follows at once.
\end{proof}

\begin{wrapfigure}{R}{5cm}
\begin{center}
\vspace{-1em}
\begin{tikzpicture}[style=thick,scale=0.85]
	\draw[fill=blue!5]  plot[smooth cycle, tension=.7] coordinates {(-1.5,0) (-1.3,1.5) (-0.5,1) (2,1.5) (1.1,-1) (0,-0.2) (-1,-.8)};
	\draw[fill=green!5]  plot[smooth cycle, tension=.7] coordinates {(2,1.5) (2,2.3) (3,2)};
\node[dot,boundary] (l) at (2,1.5) {};
\draw (l) -- ++(35:0.3);
\draw (l) -- ++(80:0.3);
\draw (l) -- ++(-120:0.5);
\draw (l) -- ++(-160:0.5);
\node[dot] (l1) at (1.1,-1) {};
\node[dot] (l2) at (-1.5,0) {};
\draw[red] (l1) -- ++(-50:0.4);
\draw (l1) -- ++(90:0.5);
\draw (l1) -- ++(120:0.5);
\draw[red] (l2) -- ++(180:0.4);
\draw (l2) -- ++(60:0.5);
\draw (l2) -- ++(0:0.5);
\draw (l2) -- ++(-50:0.5);
\node at (0.2,0.4) {$\Gamma_0$};
\node at (2.4,1.9) {$\gamma$};
\end{tikzpicture}
\end{center}
\vspace{-1em}
\caption{Generalised self-loop.}\label{fig:collapse}
\vspace{0.5em}
\end{wrapfigure}
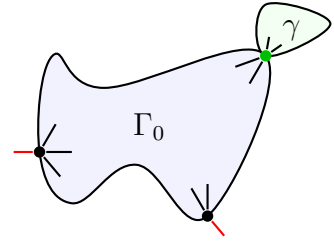
Consider the situation of a Feynman diagram $\Gamma$ containing a vertex $v$
and a subgraph $\gamma$ which is a ``generalised self-loop at $v$'' in
the sense that
\begin{claim}
\item The vertex $v$ is the only vertex of $\gamma$ that is adjacent to any edge
not in $\gamma$.
\item No leg of $\Gamma$ is adjacent to any vertex of $\gamma$, except possibly for 
$v$. 
\end{claim}
We then obtain a new diagram $\Gamma_0$ by collapsing all of $\gamma$ onto the vertex $v$,
as illustrated in Figure~\ref{fig:collapse}, where the vertex $v$ is indicated 
in green and legs are drawn in red.

As a consequence of Proposition~\ref{prop:convolution}, we conclude that 
in such a situation there exists a constant $c_\gamma \in \R$ such that 
\begin{equ}
\Pi_\BPHZ \Gamma = c_\gamma \Pi_\BPHZ \Gamma_0\;,  
\end{equ}
and that furthermore $c_\gamma = 0$ as soon as $\deg\gamma \le 0$ as a consequence of Proposition~\ref{prop:poly}.
One particularly important special case is that of actual self-loops, where
$\gamma$ consists of a single edge connecting $v$ to itself,
thus showing that
$\Pi_\BPHZ \Gamma = 0$ for every $\Gamma$ containing self-loops since
the degree of a self-loop of type $\Labhom$ is given by $\deg\Labhom$, which is always negative.

Finally, it would also appear natural to restrict the sums in \eqref{e:coaction} and \eqref{e:coprod} to
subgraphs $\bar \Gamma$ that are c-full in $\Gamma$ (in the sense that each connected component of $\bar \Gamma$ 
is a full subgraph of $\Gamma$), especially in view of the proof of the BPHZ theorem where
we saw that the ``dangerous'' connected subgraphs are always the full ones. We can then perform the exact same
steps as before, including the construction of a corresponding twisted antipode and the verification of
the forest formula. Writing $\hat \CF_\Gamma^-$ for the subset of $\CF_\Gamma^-$ consisting of
forests $\CF$ such that each $\gamma \in \CF$ is a full subgraph of its parent $\CA(\gamma)$ 
(as usual with the convention that the parent of the maximal elements is $\Gamma$ itself), it is therefore
natural in view of \eqref{e:forest} to define a valuation
\begin{equ}[e:forestFormulaFull]
\Pi_\BPHZ^\full \Gamma = (\Pi_-\otimes \Pi)\sum_{\CF \in \hat \CF_\Gamma^-} (-1)^{|\CF|} \CC_\CF \Gamma\;,
\end{equ}
where $\Pi$ and $\Pi_-$ are the canonical valuations associated to some $K \in \CK^-_\infty$.
It turns out that, maybe not so surprisingly in view of Proposition~\ref{prop:poly}, this actually 
yields the exact same valuation:

\begin{proposition}\label{prop:full}
One has $\Pi_\BPHZ^\full = \Pi_\BPHZ$.
\end{proposition}

\begin{proof}
In order to show that 
\begin{equ}
(\Pi_-\otimes \Pi)\sum_{\CF \in  \CF_\Gamma^-\setminus \hat\CF_\Gamma^-} (-1)^{|\CF|} \CC_\CF \Gamma = 0\;,
\end{equ}
we will partition $\CF_\Gamma^-\setminus \hat\CF_\Gamma^-$ into sets such that the 
above sum vanishes, when restricted to any of the sets in the partition. 
In order to formulate our construction, given $\gamma \in \CG_\Gamma^-$, we write
$\gamma^\cl \in \CG_\Gamma^-$ for the ``closure'' of $\gamma$ in $\Gamma$, i.e.\ the full
subgraph of $\Gamma$ with the same vertex set as $\gamma$.
For $\CF \in \CF_\Gamma^-\setminus \hat\CF_\Gamma^-$,
we then have a unique decomposition $\CF = \CF^\full \cup \CF^p$ such that 
each $\gamma \in \CF^\full$ is full in $\Gamma$, no element of $\CF^\full$ is 
contained in an element of $\CF^p$, and no root of
$\CF^p$ is full in $\Gamma$.

Write $\CF^p_\mmax$ for the set of roots of $\CF^p$ and set
\begin{equ}
\overline{\CF^p} = \{\gamma^\cl\,:\, \gamma \in \CF^p_\mmax\}\;.
\end{equ}
In general, one may have $\CF^\full \cap \overline{\CF^p} \neq \emptyset$, so
we also set $\CF^\full_\circ = \CF^\full \setminus \overline{\CF^p}$.
If we write $\NN\colon \CF \mapsto (\CF^p,\CF^\full_\circ)$, then we see that 
the preimage of $(\CF^p,\CF^\full_\circ)$ under $\NN$ consists of all forests of the form
$\CF^p \cup \CF^\full_\circ \cup \CB$, where $\CB$ is an arbitrary subset of $\overline{\CF^p}$.
Furthermore, $\CF_\Gamma^-\setminus \hat\CF_\Gamma^-$ consists precisely of those forests
$\CF$ such that $\CF^p \neq \emptyset$. 
Since $\sum_{\CB \subset \overline{\CF^p}} (-1)^{|\CB|} = 0$, it thus remains to show that 
the quantity
\begin{equ}[e:quantityB]
(\Pi_-\otimes \Pi) \CC_{\CF^p \cup \CF^\full_\circ \cup \CB} \Gamma
\end{equ}
is independent of $\CB \subset \overline{\CF^p}$.

To see that this is the case, consider the space $\hat \CT_\Gamma$ and
the operators $\hat \CC_\CF$ as in the proof of the BPHZ theorem and denote  
by $\hat \Pi\colon \hat \CT_\Gamma \to \CS$ the composition of $\Pi\colon \hat \CT \to \CS$ with the natural
injection $\hat \CT_\Gamma \hookrightarrow \hat \CT$. One then has for every forest $\CG$ the identity
\begin{equ}
(\Pi_-\otimes \Pi) \CC_{\CG} \Gamma = \hat \Pi \hat \CC_{\CG} \Gamma\;,\qquad \hat \CC_\CG \eqdef \prod_{\gamma \in \CG} \hat \CC_\gamma\;.
\end{equ}
(As already pointed out before, the order of the operations does not matter here.)
Let now $\gamma \in \CG_\Gamma^-$ and consider 
the elements $\hat \CC_\gamma \Gamma$ and $\hat \CC_\gamma\hat \CC_{\gamma^\cl}\Gamma$.
It follows from the definition of the operators $\hat \CC_\gamma$ that 
all the terms appearing in both expressions consist of the same graph where edges in $\Gamma \setminus \gamma^\cl$ 
adjacent to $\gamma^\cl$ are reconnected to the distinguished vertex $v_\star$ of $\gamma$ and the edges in
$\gamma^\cl$ that are not in $\gamma$ are turned into self-loops for $v_\star$. 

Regarding the edge and vertex-labels $\ell$ and $\Labn$ generated by these operations,
a straightforward application of the Chu-Vandermonde theorem shows that they yield
the exact same terms in both cases.
The only difference
is that the function $\c$ is equal to $1$ on $\gamma$ in the first case, while it equals $2$ on $\gamma$
and $1$ on edges of $\gamma^\cl$ that are not in $\gamma$  in the second case.
This however would only make a difference if we were to compose this with an operator
of the type $\hat \CC_{\bar \gamma}$ for some $\bar \gamma$ with $\gamma \subset \bar \gamma \subset \gamma^\cl$.
In our case however, we only use this in order to compare
$\CC_{\CF^p \cup \CF^\full_\circ \cup \CB}$ to $\CC_{\CF^p \cup \CF^\full_\circ}$, so that 
we consider the situation $\gamma \in \CF^p_\mmax$. Since these graphs are all vertex-disjoint, 
it follows that $(\prod_{\gamma \in \CF^p_\mmax} \hat \CC_\gamma)\Gamma$
and $(\prod_{\gamma \in \CF^p_\mmax \cup \CB} \hat \CC_\gamma)\Gamma$ only differ by the value of $\c$ 
in the way described above.

Our construction of the sets
$\CF_\circ^\full$ and $\CF^p$ then guarantees that this discrepancy is irrelevant when
further applying $\hat\CC_{\bar \gamma}$ for $\bar \gamma \in \CF_\circ^\full \cup (\CF^p \setminus \CF^p_\mmax)$,
so that \eqref{e:quantityB} is indeed independent of $\CB$ as claimed.
\end{proof}

\section{Large-scale behaviour}
\label{sec:largeScale}

We now consider the case of kernels $K_\Labhom$ that don't have compact support.
In order to encode their behaviour at infinity, we assign to each label $\Labhom \in \Lab$
a second degree $\deg_\infty\colon \Lab \to \R_- \cup \{-\infty\}$ with 
$\deg_\infty\delta^{(k)} = -\infty$ and satisfying this time the consistency condition
$\deg_\infty \Labhom^{(k)} = \deg_\infty \Labhom$.\footnote{It would have looked more natural to impose
the stronger condition $\deg_\infty \Labhom^{(k)} = \deg_\infty \Labhom - |k|$ as before. One may further
think that in this case one would be able to extend Theorem~\ref{theo:infinity} to all diagrams $\Gamma$,
not just those in $\CH_+$. This is wrong in general, although we expect it to be true after
performing a suitable form of positive renormalisation as in \cite{BHZalg,BPHZana}. This is not performed
here, and as a consequence we are unable to take advantage of the additional large-scale cancellations
that the stronger condition $\deg_\infty \Labhom^{(k)} = \deg_\infty \Labhom - |k|$ would offer.
} We furthermore
assume that we are given a collection of smooth kernels $R_\Labhom \colon \R^d \to \R$
for $\Labhom \in \Lab_\star$ satisfying the bounds
\begin{equ}[e:boundInfinity]
|D^k R_\Labhom(x)| \lesssim (2+|x|)^{\deg_\infty\Labhom}\;,
\end{equ}
for all multiindices $k$, uniformly over all $x \in \R^d$, and such that
\begin{equ}[e:compatR]
R_{\Labhom^{(k)}} = D^kR_\Labhom\;.
\end{equ}
Similarly to before, we extend this to $\Lab$ by using the convention $R_{\delta^{(m)}} \equiv 0$ and we write
$\CK^+_\infty$ for the set of all smooth \textit{compactly supported} kernel 
assignments $\Labhom \mapsto R_\Labhom$, as well as $\CK^+_0$ for its closure
under the system of seminorms defined by \eqref{e:boundInfinity}.

Consider then the formal expression \eqref{e:evalBis}, but with each instance of 
$K_\Labhom$ replaced by $G_\Labhom = K_\Labhom + R_\Labhom$. 
The aim of this section is to exhibit a sufficient condition on $\Gamma$ which 
guarantees that this expression can also be renormalised, using the same procedure as 
in the previous sections. The conditions we require in Theorem~\ref{theo:infinity} below
can be viewed as a large-scale
analogue to the conditions of Weinberg's theorem. They are required because, unlike
in \cite{BHZalg,BPHZana}, we do not perform any ``positive renormalisation'' in
the present article.

To formulate our main result, we introduce the following construction. Given a Feynman
diagram $\Gamma$ with at least one edge, consider a partition $\CP_\Gamma$ of its 
inner vertex set, i.e.\ elements of $\CP_\Gamma$ are non-empty subsets of $\CV_\star$ 
and $\bigcup \CP_\Gamma = \CV_\star$. We always consider the case where the partition
$\CP_\Gamma$ consists of at least two subsets, in other words $|\CP_\Gamma| \ge 2$.
Given such a partition, we then set
\begin{equ}
\deg_\infty \CP_\Gamma \eqdef \sum_{e \in \CE(\CP_\Gamma)}\Labhom(e) + d(|\CP_\Gamma|-1)\;,
\end{equ}
where $\CE(\CP_\Gamma)$ consists of all internal edges $e \in \CE_\star$ such that 
both ends $e_+$ and $e_-$ are contained in different elements of $\CP_\Gamma$.
Note the strong similarity to \eqref{e:degreeSubgraph}, which is of
course not a coincidence. We will call a partition $\CP_\Gamma$ ``tight'' if
there exists one single element $A \in \CP_\Gamma$ containing all of the vertices
$v_{i,\star} \in \CV_\star$ that are connected to legs of $\Gamma$.

Given $K$ and $R$ in $\CK_\infty^-$ and $\CK_\infty^+$ respectively, we furthermore 
define a valuation $\Pi^{K,R}$ by setting as in \eqref{e:evalBis}
\begin{equ}[e:PiKR]
\bigl(\Pi^{K,R} \Gamma\bigr)(\phi) = \int_{\S^{\CV_\star}} \prod_{e \in \CE_\star} G_{\Labhom(e)}(x_{e_+} - x_{e_-}) 
\bigl(D_1^{\ell_1}\cdots D_k^{\ell_k} \phi\bigr)(x_{v_1},\ldots,x_{v_k})\,dx\;,
\end{equ}
where we used again the notation $G_\Labhom = K_\Labhom + R_\Labhom$.
We then have the following result which is the analogue in this context of 
Proposition~\ref{prop:WeinFancy}. 

\begin{proposition}\label{prop:largeScale}
Let $\Gamma$ be such that every tight partition $\CP_\Gamma$ of its inner vertices
satisfies $\deg_\infty \CP_\Gamma < 0$. Then, the map 
$(K,R) \mapsto \Pi^{K,R}\Gamma$ extends continuously to all of $(K,R) \in \CK_\infty^-\times \CK_0^+$.
\end{proposition}

\begin{proof}
This is a corollary of Theorem~\ref{theo:infinity} below: given \eqref{e:PiKR} and 
given that we restrict ourselves to $K \in \CK_\infty^-$, it suffices to note that 
$\Pi^{K,R} = \Pi^{0,K+R}_\BPHZ$.
\end{proof}

\begin{remark}
The reason why it is natural to restrict oneself to tight partitions can best be seen with the 
following very simple example. Consider the case
\begin{equ}
\Gamma = 
\begin{tikzpicture}[style={thick},baseline = -0.1cm]
\node[dot,label={[shift={(0,0.1)}]below:{\tiny$v_1$}}] (l) at (0,0) {};
\node[dot,label={[shift={(0,0.1)}]below:{\tiny$v_2$}}] (m) at (1,0) {};
\node[dot,label={[shift={(0,0.1)}]below:{\tiny$v_3$}}] (r) at (2,0) {};
\draw[->] (l) -- (m) node[midway,above=-0.1] {\tiny$\Labhom_1$};
\draw[->] (m) -- (r) node[midway,above=-0.1] {\tiny$\Labhom_2$};
\draw[thick,red] (l) -- ++(180:0.5) node[midway,above=-0.1] {\tiny$0$};
\draw[thick,red] (r) -- ++(0:0.5) node[midway,above=-0.1] {\tiny$0$};
\end{tikzpicture}\;.
\end{equ}
Writing $G_i = K_{\Labhom_i} + R_{\Labhom_i}$ and identifying functions
with distributions as usual, one then has $(\Pi^{K,R}\Gamma)(x,y) = (G_1 \star G_2)(y-x)$.
If the $G_i$ are smooth functions, then this is of course well-defined as soon as their combined
decay at infinity is integrable, which naturally leads to the condition $\deg_\infty \Labhom_1
+\deg_\infty \Labhom_1 < -d$, which corresponds indeed to the condition $\deg_\infty \CP_\Gamma < 0$ for
$\CP_\Gamma = \{\{v_1,v_3\},\{v_2\}\}$, the only tight partition of the inner vertices
of $\Gamma$. Considering instead \textit{all} partitions would lead to the condition
$\deg_\infty \Labhom_i < -d$ for $i =1,2$, which is much stronger than necessary.
\end{remark}

Note now the following two facts.
\begin{claim}
\item The condition of Proposition~\ref{prop:largeScale} is compatible with the definition of the space
$\CH$ in the sense that if it is satisfied for one of the summands in the left hand side of
\eqref{e:IBP}, then it is also satisfied for all the others, as an immediate consequence of
the fact that $\deg_\infty \Labhom^{(k)} = \deg_\infty \Labhom$. In particular, we have a well-defined
subspace $\CH_+ \subset \CH$ on which the condition of Proposition~\ref{prop:largeScale} holds and therefore
$\Pi^{K,R} \Gamma$ is well-defined for $(K,R) \in \CK_\infty^-\times \CK_0^+$.
\item If $\Gamma$ satisfies the assumption of Proposition~\ref{prop:largeScale}, then it is also 
satisfied for all of the Feynman
diagrams appearing in the second factor of the summands of $\Delta \Gamma$, 
so that $\CH_+$ is invariant
under the action of $\CG_-$ on $\CH$.
\end{claim} 
This suggests that if we define a BPHZ renormalised valuation on $\CH_+$ by
\begin{equ}[e:fullBPHZ]
\Pi_\BPHZ^{K,R} = \bigl(\Pi^K_- \hat \CA \otimes \Pi^{K,R}\bigr)\Delta \;,
\end{equ}
then it should be possible to extend it to kernel assignments exhibiting self-similar 
behaviour both at the origin and at infinity. This is indeed the case, as demonstrated by the
main theorem of this section. 

\begin{theorem}\label{theo:infinity}
The map $(K,R) \mapsto \Pi_\BPHZ^{K,R}\Gamma$ extends continuously to 
$(K,R) \in \CK_0^-\times \CK_0^+$ for all $\Gamma \in \CH_+$.
\end{theorem}

\begin{proof}
Consider the space $\tilde \CT$ defined as the vector space generated by the set of 
pairs $(\Gamma, \tilde \CE)$, where $\Gamma$ is a Feynman diagram as before and 
$\tilde \CE \subset \CE_\star$ is a subset of its internal edges. 
We furthermore define a linear map 
$\CX \colon \CT \to \tilde \CT$ by
$\CX \Gamma = \sum_{\tilde \CE \subset \CE_\star} (\Gamma, \tilde \CE)$,
and we define a valuation on $\tilde \CT$ by setting
\begin{equs}
\bigl(\tilde \Pi^{K,R} (\Gamma, \tilde \CE)\bigr)(\phi) &= \int_{\S^{\CV_\star}} \prod_{e \in \CE_\star\setminus \tilde \CE} K_{\Labhom(e)}(x_{e_+} - x_{e_-}) \prod_{e \in \tilde \CE}
R_{\Labhom(e)}(x_{e_+} - x_{e_-}) \\
&\quad \times \bigl(D_1^{\ell_1}\cdots D_k^{\ell_k} \phi\bigr)(x_{v_1},\ldots,x_{v_k})\,dx\;,\label{e:tildePiKR}
\end{equs}
so that $\Pi^{K,R} = \tilde \Pi^{K,R} \CX$.
Similarly to before, we define $\d\tilde \CT$ by the analogue of \eqref{e:IBP}
and we set $\tilde \CH = \tilde \CT / \d\tilde \CT$, noting that $\tilde \Pi^{K,R}$ is well-defined
on $\tilde \CH$.

We also define a map $\tilde \Delta\colon \tilde \CH \to \CH_- \otimes \tilde \CH$
in the same way as \eqref{e:coaction}, but with the sum restricted to subgraphs $\gamma$ whose edge sets
are subsets of $\CE_\star \setminus \tilde \CE$. (This condition guarantees that $\tilde \CE$ 
can naturally be identified with a subset of the quotient graph $\Gamma / \gamma$.)
With this definition, one has the identity
\begin{equ}
\tilde \Delta \CX = (\id \otimes \CX)\Delta\;,
\end{equ}
as a consequence of the fact that the set of pairs $(\tilde \CE,\gamma)$ such that 
$\tilde \CE \subset \CE_\star$ and $\gamma$ is a subgraph of $\Gamma$ containing only edges
in $\CE_\star \setminus \CE$ is the same as the set of pairs such that 
$\gamma$ is an arbitrary subgraph of $\Gamma$ and $\tilde \CE$ is a subset of the edges of
$\Gamma / \gamma$. This in turn implies that one has the identity
\begin{equ}[e:BPHZglobal]
\Pi_\BPHZ^{K,R}\Gamma
= \bigl(\Pi^K_- \hat \CA \otimes \Pi^{K,R}\bigr)\Delta \Gamma
= \bigl(\Pi^K_- \hat \CA \otimes \tilde \Pi^{K,R}\bigr)\tilde \Delta \CX \Gamma\;.
\end{equ}
Let now $\tilde \CT_+$ be the subspace of $\tilde \CT$ consisting of pairs $(\Gamma, \tilde \CE)$ 
such that $\deg_\infty \CP < 0$ for every tight partition $\CP$ with $\CE(\CP) \subset \tilde \CE$.
Again, this defines a subspace $\tilde \CH_+ \subset \tilde\CH$ invariant under the action
of $\CG_-$ by $\tilde \Delta$ and $\CX$ maps $\CH_+$ (defined as in the statement of the theorem) 
into $\tilde \CH_+$, so that it
remains to show that $\bigl(\Pi^K_- \hat \CA \otimes \tilde \Pi^{K,R}\bigr)\tilde \Delta$
extends to kernels $(K,R) \in \CK_0^-\times \CK_0^+$ on all of $\tilde \CH_+$.

For this, we now fix $\tau = (\Gamma, \tilde \CE) \in \tilde \CT_+$ and we remark that for $R \in \CK_\infty^+$
we can interpret the factor $\prod_{e \in \tilde \CE}
R_{\Labhom(e)}(x_{e_+} - x_{e_-})$ in \eqref{e:tildePiKR} as being part of the test function.
More precisely, we set
\begin{equ}
\phi \otimes_\tau R = \phi(x_{1},\ldots,x_{k})\prod_{e \in \tilde \CE}
R_{\Labhom(e)}(x_{[e]_+} - x_{[e]_-})\;,
\end{equ}
where $[\cdot]_\pm \colon \tilde \CE \to \{k+1,\ldots,k+2|\tilde \CE|\}$ is an arbitrary but fixed numbering
of the half-edges of $\tilde \CE$.
We then have $(\tilde \Pi^{K,R} \tau)(\phi) = (\Pi^{K} \CU \tau)(\phi \otimes_\tau R)$,
where $\CU \tau \in \CT_+$ is the Feynman diagram obtained by cutting each of the edges $e \in \tilde \CE$
open, replacing them by two legs with label $\delta$ and numbers given by $[e]_\pm$. 
It is immediate from the definitions
and the condition \eqref{e:compatR} that this is compatible with the 
actions of $\tilde \Delta$ and $\Delta$ in the sense that one has 
\begin{equ}
((g \otimes \tilde \Pi^{K,R}) \tilde \Delta \tau)(\phi)  =
((g \otimes \Pi^{K}) \Delta \CU\tau)(\phi \otimes_\tau R)\;,\qquad \forall g \in \CG_-\;.  
\end{equ}
Inserting this into \eqref{e:BPHZglobal}, we conclude that
\begin{equ}
\bigl(\Pi_\BPHZ^{K,R}\Gamma\bigr)(\phi) = \sum_{\tilde \CE \subset \CE_\star} \bigl(\Pi_\BPHZ^K \CU(\Gamma,\tilde \CE)\bigr)(\phi \otimes_{(\Gamma,\tilde \CE)} R)\;,
\end{equ}
so that it remains to bound separately each of the terms in this sum. 

For this, we write $\S_d = \Z^d$ for the discrete analogue of our state space $\S = \R^d$,
we set $N = k+2|\tilde \CE|$, and we write
$1 = \sum_{x \in \S_d^N} \Psi_x$ for a partition of unity with the property that 
$\Psi_x(y) = \Psi_0(y-x)$ and that $\Psi_0$ is supported in a cube of sidelength $2$ centred 
at the origin, so that it remains to show that 
\begin{equ}
\sum_{x \in \S_d^N}S_x\;,\qquad S_x \eqdef \bigl(\Pi_\BPHZ^K \CU(\Gamma,\tilde \CE)\bigr)((\phi \otimes_{(\Gamma,\tilde \CE)} R)\Psi_x)\;,
\end{equ}
is absolutely summable. It then follows from Theorem~\ref{theo:BPHZ} that the summand in the
above expression is bounded by 
\begin{equ}
|S_x| \lesssim \prod_{e \in \tilde \CE} (1 + |x_{[e]_+} - x_{[e]_-}|)^{\deg_\infty\Labhom(e)}\;,
\end{equ}
for all $(K,R) \in \CK_0^-\times \CK_0^+$.
This expression is not summable in general, so we need to exploit the fact that there
are many terms that vanish. For instance, since the test function $\phi$ is compactly supported, there exists
$C$ such that $S_x = 0$ as soon as $|x_i| \ge C$ for some $i \le k$.
Similarly, since the kernels $K_\Labhom$ are compactly supported, there exists $C$ such that 
$S_x = 0$ as soon as there are two legs $[i]$ and $[j]$ of $\CU_\tau$ attached to the same connected component
and such that $|x_i-x_j| \ge C$.

Let now $\CP$ be the finest tight partition for $\Gamma$ with $\CE(\CP) \subset \tilde \CE$ and
let $L \in \CP$ denote the (unique) set which contains all the vertices adjacent to the legs of $\Gamma$.
We conclude from the above consideration that one has  
\begin{equ}[e:wantedBoundLS]
\sum_{x \in \S_d^N}|S_x| \lesssim \sum_{y \in \S_d^{\CP}} \one_{\{y_L = 0\}} \prod_{e \in \CE(\CP)} (1 + |y_{[e_+]} - y_{[e_-]}|)^{\deg_\infty\Labhom(e)}\;,
\end{equ}
where $[v] \in \CP$ denotes the element of $\CP$ containing the vertex $v$.
At this stage, the proof is virtually identical to that of Weinberg's theorem, with the difference
that we need to control the large-scale behaviour instead of the small-scale behaviour.
We define Hepp sectors $D_\bT \subset \S_d^{\CP}$ for $\bT = (T,\bn)$ in exactly the same way as before,
the difference being that this time no two elements can be at distance \textit{less} than $1$, so that 
we can restrict ourselves to scale assignments with $\bn_v \le 0$ for every inner vertex of $T$.
Also, in view of \eqref{e:wantedBoundLS}, the leaves of $T$ are this time given by elements of $\CP$.
In the same way as before, the number of elements of $D_\bT$ is of the order of 
$\prod_{u \in T} 2^{-d \bn_u}$ so that one has again a bound of the type
\begin{equ}[e:goodbound]
\sum_{x \in \S_d^N}|S_x| \lesssim  \sum_{\bT} \prod_{u \in T} 2^{-\bn_u\eta_u}\;,\quad \eta_u = d + \sum_{e \in \CE(\CP)}\one_{e^\uparrow}
\deg_\infty\Labhom(e)\;,
\end{equ}
where $e^\uparrow$ denotes the common ancestor in $T$ of the two elements of $\CP$ containing
the two endpoints of $e$.
Our assumption on $\Gamma$ now implies that for every initial segment $T_i$ of $T$\footnote{i.e.\ $T_i$
is such that if $u \in T_i$ and $v \le u$, then $v \in T_i$.},
one has $\sum_{u \in T_i} \eta_u < 0$. This is because one has
$\sum_{u \in T_i} \eta_u = \deg \CP_{T_i}$, where $\CP_{T_i}$ is the coarsest coarsening of $\CP$
such that for every edge $e \in \CE(\CP_{T_i})$, one has $e^\uparrow \not \in \CP_{T_i}$.

We claim that any such $\eta$ satisfies 
\begin{equ}
S(T,\eta) \eqdef \sum_\bn \prod_{u \in T} 2^{-\bn_u\eta_u} < \infty\;,
\end{equ}
where again the sum is restricted to negative $\bn$ that are monotone on $T$. 
This can be shown by induction over the number of leaves of $T$. If $T$ has only two leaves,
then this is a converging geometric series and the claim is trivial. 
Let now $T$ be a tree with $m \ge 3$ leaves and assume that the claim holds for all trees with $m-1$ leaves.
Pick an inner vertex $u$ of $T$ which has exactly two descendants (such a vertex always exists since
$T$ is binary) and write $\tilde T$ for the new tree obtained from $T$ by deleting $u$ and coalescing 
its two descendants into one single leaf. Write furthermore $u^\uparrow$ for the parent of $u$ in $T$,
which exists since $T$ has at least three leaves. The following example illustrates this construction:
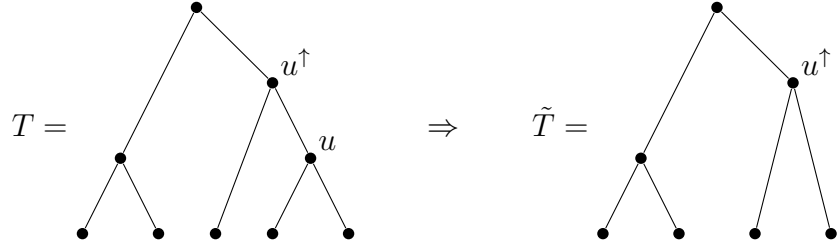
\begin{equ}
T = 
\begin{tikzpicture}[baseline=1.3cm]
\node[dot] (r) at (0,3) {};
\node[dot] (0) at (-1,1) {};
\node[dot,label={[shift={(-0.1,-0.1)}]above right:{$u^\uparrow$}}] (1) at (1,2) {};
\node[dot] (00) at (-1.5,0) {};
\node[dot] (01) at (-0.5,0) {};
\node[dot] (10) at (0.25,0) {};
\node[dot,label={[shift={(-0.1,-0.1)}]above right:{$u$}}] (11) at (1.5,1) {};
\node[dot] (110) at (1,0) {};
\node[dot] (111) at (2,0) {};
\draw (r) -- (0);
\draw (0) -- (00);
\draw (0) -- (01);
\draw (r) -- (1);
\draw (1) -- (10);
\draw (1) -- (11);
\draw (11) -- (110);
\draw (11) -- (111);
\end{tikzpicture}\qquad \Rightarrow\qquad
\tilde T = 
\begin{tikzpicture}[baseline=1.3cm]
\node[dot] (r) at (0,3) {};
\node[dot] (0) at (-1,1) {};
\node[dot,label={[shift={(-0.1,-0.1)}]above right:{$u^\uparrow$}}] (1) at (1,2) {};
\node[dot] (00) at (-1.5,0) {};
\node[dot] (01) at (-0.5,0) {};
\node[dot] (10) at (0.5,0) {};
\node[dot] (11) at (1.5,0) {};
\draw (r) -- (0);
\draw (0) -- (00);
\draw (0) -- (01);
\draw (r) -- (1);
\draw (1) -- (10);
\draw (1) -- (11);
\end{tikzpicture}
\end{equ}
Since the condition on $\eta$ is open and since $S_\eta$ increases when increasing $\eta_u$,
we can assume without loss of generality that $\eta_u \neq 0$.
There are then two cases:
\begin{claim}
\item If $\eta_u < 0$, we have $\sum_{\bn_u > \bn_{u^\uparrow}} 2^{-\bn_u\eta_u} \approx 1$,
so that 
\begin{equ}[e:induction]
S(T,\eta) \approx S(\tilde T,\tilde \eta)\;,
\end{equ}
where $\tilde \eta_v$ is just the restriction of $\eta_v$ to the tree $\tilde T$.
Since initial segments of $\tilde T$ are also initial segments of $T$ and since $\tilde \eta = \eta$ on them,
we can make use of the induction hypothesis to conclude.
\item If $\eta_u > 0$, we have $\sum_{\bn_u > \bn_{u^\uparrow}} 2^{-\bn_u\eta_u} \approx 2^{-\bn_{u^\uparrow}\eta_u}$,
so that \eqref{e:induction} holds again, but this time 
$\tilde \eta_{u^\uparrow} = \eta_{u^\uparrow} + \eta_{u}$ and $\tilde \eta_v = \eta_v$ otherwise.
We conclude in the same way as before since the only ``dangerous'' case is that of initial
segments $\tilde T_i$ containing $u^\uparrow$, but these are in bijection with the initial
segment $T_i = \tilde T_i \cup \{u\}$ of $T$ such that $\sum_{v \in \tilde T_i} \tilde \eta_v
= \sum_{v \in T_i} \eta_v$, so that the induction hypothesis still holds.
\end{claim}
Applying this to \eqref{e:goodbound} completes the proof of the theorem.
\end{proof}

\begin{remark}
While the definition of $\Pi_\BPHZ^{K,R}$ is rather canonical, given kernel assignments $K$ and $R$,
the decomposition $G = K + R$ is \textit{not}. Using the fact that $\CG_-$ is a group, it is 
however not difficult to see that, for any two choices $(K,R), (\bar K, \bar R) \in \CK_0^-\times \CK_0^+$
such that
\begin{equ}
K_\Labhom + R_\Labhom = \bar K_\Labhom + \bar R_\Labhom\;,\qquad \forall \Labhom \in \Lab_\star\;,
\end{equ}
there exists an element $g \in \CG_-$ such that
$\Pi_\BPHZ^{\bar K,\bar R} = (g \otimes \Pi_\BPHZ^{K,R})\Delta$.
\end{remark}

\bibliographystyle{Martin}

\bibliography{refs}

\end{document}